\documentclass{article}
\usepackage{graphicx} 
\usepackage{authblk}  
\usepackage{subcaption} 
\usepackage{booktabs} 
\usepackage{colortbl}
\usepackage{tgpagella}

\def\notes{0}

\usepackage{amsmath,amssymb,amsthm,verbatim,relsize,mathrsfs,hyperref,soul}
\usepackage{algorithm}
\usepackage{mathtools}
\usepackage[noend]{algpseudocode}
\hypersetup{colorlinks = true,linkcolor = violet, anchorcolor = violet, citecolor = violet, filecolor = red,urlcolor = red}
\usepackage [english]{babel}
\usepackage{blindtext}
\usepackage{framed}
\usepackage{fullpage}
\usepackage[shortlabels]{enumitem}
\usepackage{graphicx}
\usepackage{xcolor}
\usepackage{subcaption}
\usepackage{subfiles}
\usepackage{dsfont}
\usepackage{extarrows}
\usepackage{bbm}
\usepackage[normalem]{ulem}
\usepackage{todonotes}
\usepackage{subcaption}
\usepackage{graphicx}
\usepackage{tcolorbox}
\usetikzlibrary{shapes.geometric}
\tcbuselibrary{skins}
\usepackage[numbers]{natbib}  

\usepackage{setspace}

\usepackage{cleveref}
\usepackage{hyperref}

\graphicspath{ {./figures/} }

\makeatletter
\newcommand\footnoteref[1]{\protected@xdef\@thefnmark{\ref{#1}}\@footnotemark}
\makeatother

\newcommand{\ignore}[1]{}

\ifnum\notes=1
    \newcommand{\mynote}[3]{\marginpar{\tiny\sf \color{#2} {#1}: {#3}}}
    \newcommand{\inlinenotes}[3]{{\color{#2} [{\bf {#1}:} {\sf {#3}}]}}
    \newcommand{\strikeout}[2]{{\color{#1}  \sout{\color{black} #2}}}
\else 
    \newcommand{\mynote}[3]{}
    \newcommand{\inlinenotes}[3]{}
    
    \newcommand{\strikeout}[2]{}
\fi

\newcommand{\asnote}[1]{\mynote{A}{purple}{#1}}

\newcommand{\shnote}[1]{\mynote{S}{blue}{#1}}

\newcommand{\jnote}[1]{\mynote{J}{teal}{#1}}
\newcommand{\jsout}[1]{\strikeout{teal}{#1}}
\newcommand{\jadd}[1]{\textcolor{teal}{#1}}

\newcommand{\R}{\mathbb{R}}

\newcommand{\eps}{\varepsilon}

\newcommand{\cA}{\mathcal{A}}
\newcommand{\cG}{\mathcal{G}}

\newcommand{\paren}[1]{\left( #1 \right)}

\newcommand{\defeq}{\vcentcolon=}

\algrenewcommand\algorithmicrequire{\textbf{Input:}}
\algrenewcommand\algorithmicensure{\textbf{Output:}}

\newtheorem{theorem}{Theorem}[section]

\newtheorem{definition}[theorem]{Definition}
\newtheorem{lemma}[theorem]{Lemma}

\newcommand{\Sec}[1]{\hyperref[sec:#1]{Section~\ref*{sec:#1}}} 
\newcommand{\Eqn}[1]{\hyperref[eq:#1]{(\ref*{eq:#1})}} 
\newcommand{\Fig}[1]{\hyperref[fig:#1]{Fig.\,\ref*{fig:#1}}} 
\newcommand{\Tab}[1]{\hyperref[tab:#1]{Tab.\,\ref*{tab:#1}}} 
\newcommand{\Thm}[1]{\hyperref[thm:#1]{Theorem\,\ref*{thm:#1}}} 
\newcommand{\Fact}[1]{\hyperref[fact:#1]{Fact\,\ref*{fact:#1}}} 
\newcommand{\Lem}[1]{\hyperref[lem:#1]{Lemma\,\ref*{lem:#1}}} 
\newcommand{\Lems}[2]{\hyperref[lem:#1]{Lemmas\,\ref*{lem:#1}} and~\hyperref[lem:#2]{\ref*{lem:#2}}} 
\newcommand{\Prop}[1]{\hyperref[prop:#1]{Prop.\,\ref*{prop:#1}}} 
\newcommand{\Cor}[1]{\hyperref[cor:#1]{Corollary~\ref*{cor:#1}}} 
\newcommand{\Conj}[1]{\hyperref[conj:#1]{Conjecture~\ref*{conj:#1}}} 
\newcommand{\Def}[1]{\hyperref[def:#1]{Definition~\ref*{def:#1}}} 
\newcommand{\Alg}[1]{\hyperref[alg:#1]{Algorithm~\ref*{alg:#1}}} 
\newcommand{\Proc}[1]{\hyperref[proc:#1]{Procedure~\ref*{proc:#1}}} 
\newcommand{\Step}[1]{\hyperref[step:#1]{Step~\ref*{step:#1}}} 
\newcommand{\Steps}[2]{\hyperref[step:#1]{Steps~\ref*{step:#1}} and~\hyperref[step:#2]{\ref*{step:#2}}} 
\newcommand{\Stepss}[3]{\hyperref[step:#1]{Steps~\ref*{step:#1}},~\hyperref[step:#2]{\ref*{step:#2}}, and~\hyperref[step:#3]{\ref*{step:#3}}} 
\newcommand{\Ex}[1]{\hyperref[ex:#1]{Ex.~\ref*{ex:#1}}} 
\newcommand{\Clm}[1]{\hyperref[clm:#1]{Claim~\ref*{clm:#1}}} 
\newcommand{\Inv}[1]{\hyperref[inv:#1]{Invariant~\ref*{inv:#1}}} 
\newcommand{\Rem}[1]{\hyperref[rem:#1]{Remark~\ref*{rem:#1}}} 
\newcommand{\Obs}[1]{\hyperref[obs:#1]{Observation~\ref*{obs:#1}}} 

\newtcolorbox{findingbox}[1]{
    colback=white,              
    colframe=black,             
    fonttitle=\bfseries\color{white},  
    colbacktitle=black,         
    title={Finding #1},         
    arc=3mm,                    
    boxrule=0.5pt,             
    left=5pt,right=5pt,top=5pt,bottom=5pt  
}

\title{Differentially Private Modeling of Disease Transmission within Human Contact Networks}
\date{}

\author[1]{Shlomi Hod\thanks{Equal contribution.}}
\author[1]{Debanuj Nayak$^{*}$}
\author[2,3]{Jason R. Gantenberg}
\author[1]{Iden Kalemaj}
\author[2,4]{Thomas A. Trikalinos}
\author[1]{Adam Smith}

\affil[1]{Department of Computer Science, Boston University}
\affil[2]{Department of Health Services, Policy \& Practice, School of
Public Health, Brown University}
\affil[3]{Department of Epidemiology, School of Public Health, Brown University}
\affil[4]{Department of Biostatistics, School of Public Health, Brown University}

\begin{document}

\maketitle

\vspace{-3em}


        
        





\begin{abstract}
    Epidemiologic studies of infectious diseases often rely on models of contact networks to capture the complex interactions that govern disease spread, and ongoing projects aim to vastly increase the scale at which such data can be collected. 
    However, contact networks may include sensitive information, such as sexual relationships or drug use behavior. 
    Protecting individual privacy while maintaining the scientific usefulness of the data is crucial.
    We propose a privacy-preserving pipeline for disease spread simulation studies based on a sensitive network that integrates differential privacy (DP) with statistical network models such as stochastic block models (SBMs) and exponential random graph models (ERGMs).
    Our pipeline comprises three steps: (1) compute network summary statistics using \emph{node-level} DP (which corresponds to protecting individuals' contributions); 
    (2) fit a statistical model, such as an ERGM, using these summaries, which allows generating synthetic networks reflecting the structure of the original network; and (3) simulate disease spread on the synthetic networks using an agent-based model. 
    We evaluate the effectiveness of our approach using a simple Susceptible-Infected-Susceptible (SIS) disease model under multiple configurations. We compare both numerical results, such as simulated disease incidence and prevalence, as well as qualitative conclusions such as intervention effect size, 
    on networks generated with and without differential privacy constraints. 
    Our experiments are based on egocentric sexual network data from the ARTNet study (a survey about HIV-related behaviors). Our results show that the noise added for privacy is small relative to the other sources of error (sampling and model misspecification, for example). This suggests that, in principle, curators of such sensitive data can provide valuable epidemiologic insights while protecting privacy. 
\end{abstract}

\newpage
\setcounter{tocdepth}{2}
\tableofcontents
\newpage

\section{Introduction}
\label{sec:intro}

Simulation models of infectious disease transmission are important tools for developing causal explanations \cite{Epstein2008-why,Ackley2021-dynamical} and public health responses to both endemic and epidemic pathogens \cite{Lessler2016-mechanistic,Lofgren2014-opinion}. For many pathogens, the structures of human contacts act as the substrate through which the pathogen propagates, governing transmission dynamics \cite{Keeling2005,Keeling2005-implications} and epidemic potential \cite{Althouse2020-superspreading}. For this reason, epidemiologic models (e.g., network agent-based models \cite{Jenness2017EpiModelAR}) often incorporate detailed information about human contact networks as a means to simulate realistic transmission scenarios. 

Information on individuals and their relationships (connections) is naturally modeled as a a graph (hereon, \emph{network graph} or \emph{network}), where nodes represent individuals and edges the relationships between them~\cite{Jenness2017EpiModelAR, Morris2004-network,Rathkopf2018-network}. In infectious disease modeling, edges represent relationships relevant to transmission.
For instance, edges in a  model of HIV or sexually transmitted disease may represent sexual \cite{goodreau2012drives,jenness2016impact,Aral2008-understanding} or injection drug use partnerships \cite{Rutherford2016-control}, whereas in a model for influenza edges may represent sufficient conditions for transmission based on physical proximity \cite{mapps2023,mappsWorkshop2023}. 

The detailed network data poses risks to privacy \cite{Ji2017Graph}. Study participants who provide data on disease status or illegal behavior risk re-identification, with potentially grave social or legal consequences \cite{Lazer2009-computational}. This risk is not limited to the release of individual-level information \cite{Ji2017Graph,Backstrom2027Wherefore,Narayanan2009De} and may be incurred even when researchers present data summaries, including network summary statistics \cite{Ellers2019-embeddings}. Such summaries could (indirectly) reveal information about an individual's presence in the dataset. Simply knowing an individual's position in a network---for example, a person that represents the sole connection between two otherwise isolated groups of individuals \cite{Rathkopf2018-network}---could reveal information about this person's identity. Given the extent to which network structures govern transmission dynamics, we suspect that even summary outputs from infectious disease simulations could reveal information that increases the risk of privacy violation, possibly in combination with external information or other summaries derived from the same data. That is, stochastic epidemic simulations may not be inherently privacy-preserving.

\emph{Differential privacy}~\cite{DworkMNS16} (DP) is a widely studied and deployed formal privacy guarantee for data analysis. We say that an algorithm or a statical analysis is differentially private if the probability distribution of its output looks nearly indistinguishable for any two input datasets that differ only in the data of a single individual. In other words, a DP algorithm is robust to changing an individual record in the dataset. A common building block in constructing a differentially private version of a statistical query, such as count and median, is to add calibrated noise that masks the contribution of an individual data point while still providing sufficiently accurate results.

In the context of databases representing human contact networks, where nodes represent individuals, the change of one person's data amounts to changing their connections in the network, and possibly all of them. This concept is captured by \emph{node differential privacy} (\Cref{def:node-dp}) (or, more concisely node-DP)~\cite{HayLMJ09} where the indistinguishability guarantee applies to two networks which can differ on one node (individual) and all its edges (connections). The current work focuses on node-DP given that we aim to protect all of the information pertaining to one individual in a human contact network.

In order to support the responsible collection, use, and sharing of these growing network datasets for epidemic simulations, we present and evaluate in this paper a pipeline to analyze such networks through node-DP algorithms. In our pipeline, illustrated in \Cref{fig:pipeline}, a differentially private algorithm produces network summaries that are released and subsequently used to generate synthetic networks in which epidemics are simulated. The input data is a network, possibly with additional attributes for nodes such as age and race. The final outputs typically include statistics on these simulations, such as the effect of an intervention on a disease's prevalence.

Our core research question in this work is \textbf{whether a differentially private synthetic network, which provides rigorous privacy constraints, can yield accurate estimation of epidemic simulation outcomes.} That is, whether DP synthetic networks can match the accuracy of both simulations on the original network and non-private synthetic networks commonly used in epidemic research.

\subsection{Summary of Contributions}

\paragraph{(1) A privacy-preserving pipeline for epidemic simulation.} We design and implement a pipeline (\Cref{sec:pipeline}) that enables epidemiologists to perform disease spread simulations on sensitive network data while satisfying node-level differential privacy guarantees.

\paragraph{(2) A comprehensive evaluation framework.} We designed an experimental framework (\Cref{sec:experiment}) to systematically assess the pipeline across multiple dimensions, including network modeling choices, population-level and subgroup-level epidemic statistics, and various sources of variance. This enables studying the effect of DP noise in comparison to other error sources such as modeling, simulation, and network sampling.

\paragraph{(3) Empirical findings on sources of error and variance.} Through extensive experiments on realistic network data, we demonstrate that (\Cref{sec:findings}):
\begin{enumerate}
    \item Network model misspecification, when present, dominates all other sources of bias and is unrelated to differential privacy (\Cref{sec:find1});
    \item Introducing differential privacy has a limited impact on epidemic analysis accuracy for reasonable parameter choices (\Cref{sec:find2});
    \item Accurate epidemic statistics are maintained at the granular subgroup level when network models are correctly specified (\Cref{sec:find3});
    \item Variance introduced by differential privacy is smaller than variance from the network sampling and epidemic simulation steps already present in standard epidemiological practice (\Cref{sec:find4}).
\end{enumerate}

These findings demonstrate the feasibility of privacy-preserving epidemic modeling using differentially private synthetic networks, without substantially compromising the accuracy of simulation analysis. In \Cref{sec:discussion}, we discuss the broader implications of these findings, along with the limitations and directions for future work. Lastly, we present related work in \Cref{sec:related-work}. Our code, including the complete pipeline implementation and evaluation framework, is publicly available.\footnote{\url{https://github.com/shlomihod/epidp}}


\begin{figure}[t]
\includegraphics[width=\linewidth]{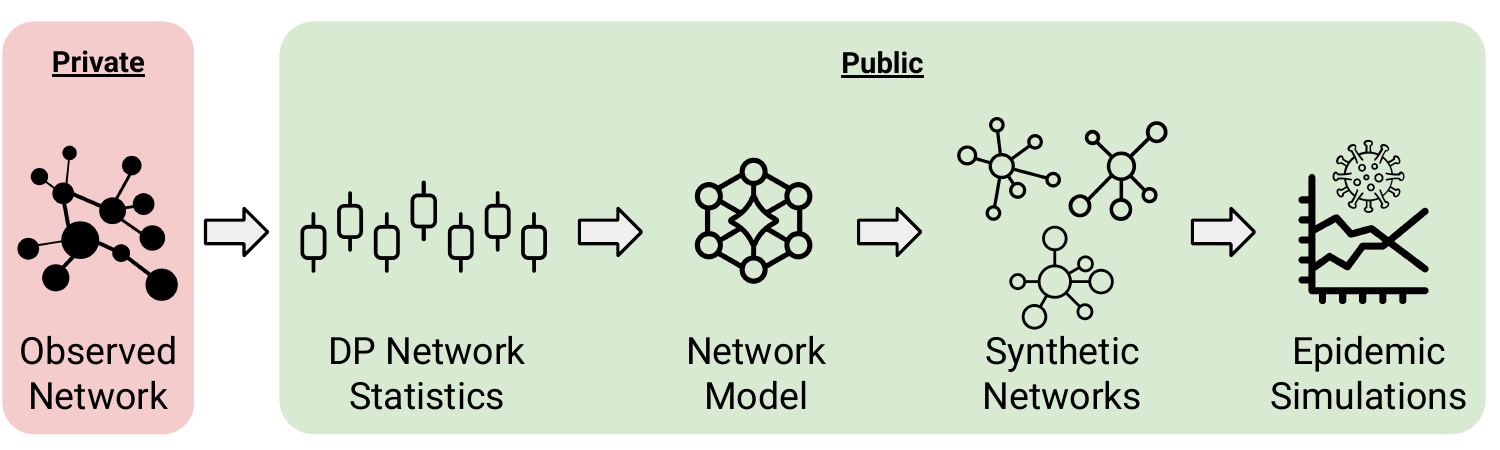}
\centering
\caption{An overview of our pipeline for differentially private (DP) simulation of infectious disease transmission. We begin by estimating network statistics using empirical network data. Next, we render these statistics private by adding noise and use them to fit network models (e.g., exponential random graph models). Finally, we run epidemic simulations over networks simulated based on these fitted models. These epidemic simulations inherit the privacy protections of the noised network statistics.}
\label{fig:pipeline}
\end{figure}

\begin{figure}[tb]
\includegraphics[width=0.7\linewidth]{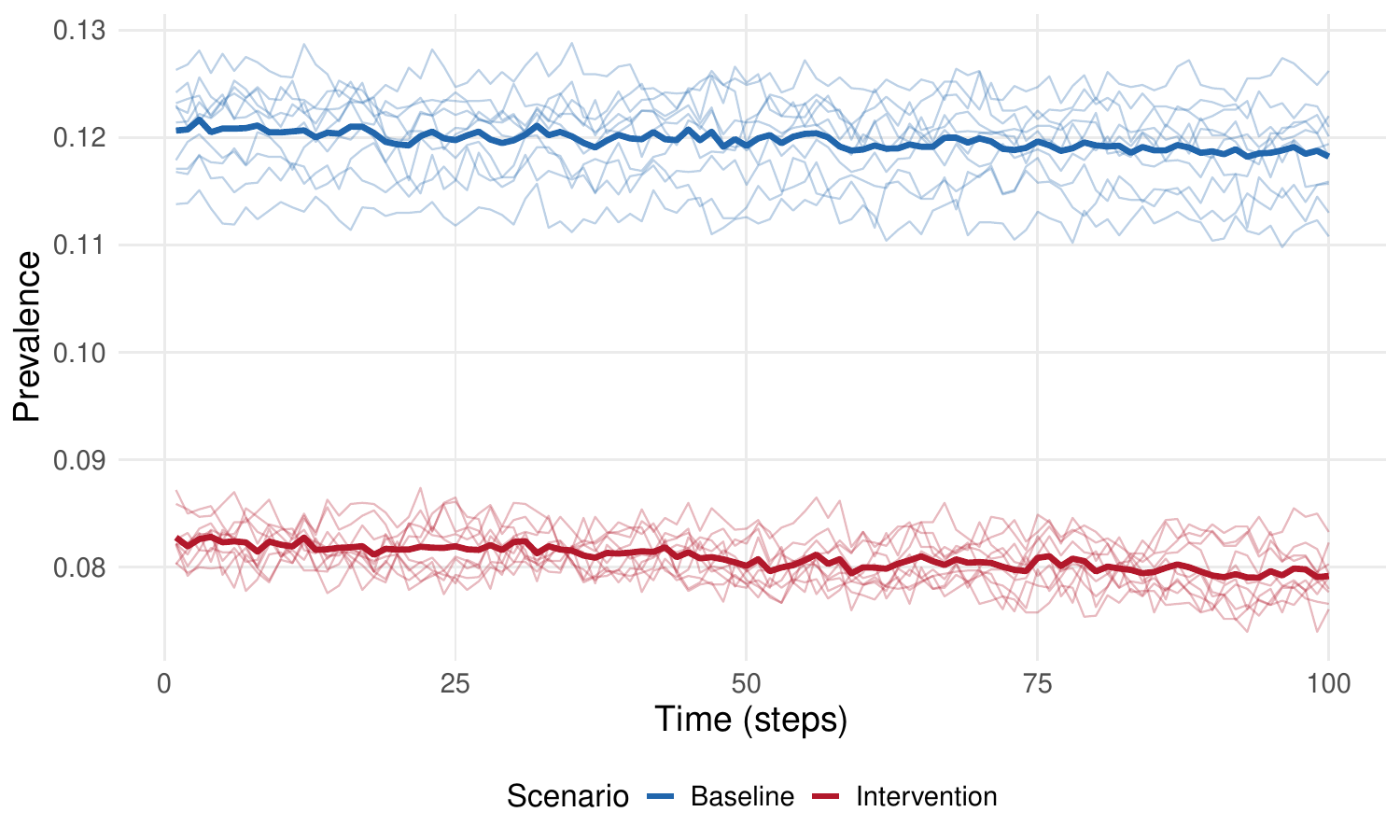}
\centering
\caption{Epidemic dynamics comparing baseline and intervention scenarios under the ``high prevalence'' condition. The agent-based model simulates transmission over a synthetic network generated from an ERGM fitted to an observed network estimated based on ARTNet data, without differential privacy (refer to \Cref{sec:setup} for details). Individual simulations ($n=10$) shown with transparency; solid lines indicate mean prevalence.}
\label{fig:epicurve}
\end{figure}

\section{Pipeline Walkthrough and Background}
\label{sec:pipeline}

This section provides an overview of our pipeline and its components (\Cref{fig:pipeline}). The pipeline is designed to align with common practices in infectious disease modeling and ultimately includes one additional step that provides privacy protection via DP.

The pipeline takes as input a human contact network, the \emph{observed network}, in which nodes represent individuals---possibly attached to demographic attributes (e.g., age and race)---and edges represent relationships between nodes. The definition of a ``relationship'' will vary and could indicate friendships, sexual partnerships, or other modes of social contact relevant for a given research question \cite{Morris2004-network}. The pipeline consists of three steps. 

Differentially private algorithms compute key network statistics (1st step)---such as mixing matrices and degree distributions---which are then used to parameterize statistical network models (2nd step), such as stochastic block models (SBMs) and exponential random graph models (ERGMs). These fitted models generate synthetic networks (3rd step) that preserve essential structural properties while providing differential privacy protection. We then simulate disease transmission on these synthetic networks and conduct multifaceted measurements (4th step).

In the rest of this section, we describe in detail the three components of the pipeline  together. Although our pipeline is general, for readability, we situate it within the context of the experiment we conduct in this work. We work backward in the description of the pipeline, starting with its purpose: conducting epidemic simulations, We first discuss briefly the epidemic simulation step (\Cref{sec:epidemic-simulation}). Subsequently, we present the network modeling step, and in particular, the  statistical models employed for synthetic network generation (\Cref{sec:net-gen-models}). Finally we deep dive into the use differential privacy to release the network statistics utilized to parameterize these models (\Cref{sec:dp-stats}).

\subsection{Epidemic Simulation}
\label{sec:epidemic-simulation}

We focus here on a network agent-based model (ABM)~\cite{Jenness2017EpiModelAR,Hunter2008-ergm} in which agents represent men who have sex with men (MSM) and epidemics are simulated over synthetic networks generated by statistical network models. These network ABMs are a natural place to start because they rely directly on network statistics estimated from empirical data, though one could apply our pipeline to an ABM in which contacts were generated algorithmically and use the DP network statistics instead as targets to calibrate these algorithms.
Our epidemic simulation is based loosely on gonorrhea transmission within sexual networks of MSM. We model a susceptible-infectious-susceptible (SIS) pathogen in which susceptible individuals who are infected become susceptible again once they recover. We simulate two treatment scenarios and two prevalence levels. The two treatment scenarios consist of a ``baseline case'' in which individuals remain untreated for gonorrhea, recovering at the same rate, and an ``intervention case'' in which individuals may be tested and treated, recovering at different rates based on agents' treatment status. 
These stylized scenarios aim to understand the effect of node-DP on epidemic simulations; they mimic a study in which ABMs are used to simulate counterfactual policies and estimate their effects on population-level measures of disease. 
We simulate both the ``baseline'' and ``intervention'' scenarios at two prevalence levels, in which the baseline case model was calibrated to achieve ``high'' ($\approx 12\%$) and ``
low ($\approx 2\%$) stable prevalence of the pathogen.

We defer a detailed presentation of the epidemic simulation design to the experimental design section (\Cref{sec:simulation-technical}). \Cref{fig:epicurve} illustrates the dynamics from multiple simulations conducted on the same network under the high-prevalence condition, using our empirical data.

\subsection{Statistical (Generative) Models of Networks}
\label{sec:net-gen-models}

Although our pipeline is general enough to handle many types of generative network models, in this work we focus on \emph{stochastic block models} (SBM) \cite{holland1983stochastic} and \emph{exponential random graph models} (ERGMs) \cite{hunter2006inference, Krivitsky2014separable}. We choose to utilize these models because they are common within the literature of disease spread simulations \cite{jenness2016effectiveness, jenness2016impact, goodreau2012concurrent, goodreau2012drives, jenness2017incidence, goodreau2017sources,nelson2020modeling,jenness2021dynamic}.

Stochastic block models (SBMs) partition nodes into communities, where the probability of an edge between two vertices depends only on their group memberships. In the simplest case of a network $G$ with $K$ groups, the model is parameterized by a $K \times K$ matrix $f_\mathrm{SBM}(G) \defeq M$, called the mixing matrix, where $M_{ij}$ gives the number of edges between vertices in groups $i$ and $j$.\footnote{Alternatively, Stochastic Block Models (SBMs) are defined using a $K \times K$ probability matrix $P$ where the $P_{ij}$ gives the probability of an edge between a pair of vertices in groups $i$ and $j$. Since the group membership of each node is publicly known, the number of possible vertex pairs for any two choices of groups $i$ and $j$, denoted by $n_{\{i,j\}}$ is also known. It is $n_i \cdot n_j$ for $i \neq j$ and ${n_i \choose 2}$ for $i = j$ where $n_i$ is the number of vertices in group $i$. This allows us relate the two alternative definitions as $P_{ij} = M_{ij}/ n_{\{i,j\}}$ and conclude that the two defintions are equivalent.} This allows modeling of assortative structures (higher within-group connectivity) or disassortative patterns depending on the block structure. The mixing matrix $M$ is the minimal sufficient statistic for an SBM. \Cref{fig:graph-mixing-matrix-example} provides an example of a graph and its mixing matrix.

\begin{figure}[t]
    \centering
\definecolor{tab10blue}{HTML}{0072B2}
\definecolor{tab10orange}{HTML}{D55E00}
\definecolor{tab10magenta}{HTML}{009E73}
    
    \begin{subfigure}[b]{0.45\textwidth}
        \centering
        \begin{tikzpicture}[scale=1, transform shape]
            \node[draw, circle, fill=tab10blue, text=white, minimum size=0.8cm] (A) at (0,2) {A};
            \node[draw, circle, fill=tab10blue, text=white, minimum size=0.8cm] (B) at (-1,1) {B};
            \node[draw, circle, fill=tab10blue, text=white, minimum size=0.8cm] (C) at (1,1) {C};
            \node[draw, rectangle, fill=tab10orange, text=white, minimum size=0.8cm] (D) at (-2,0) {D};
            \node[draw, rectangle, fill=tab10orange, text=white, minimum size=0.8cm] (E) at (2,0) {E};
            \node[draw, diamond, fill=tab10magenta, text=white, minimum size=0.8cm] (F) at (-1,-1) {F};
            \node[draw, diamond, fill=tab10magenta, text=white, minimum size=0.8cm] (G) at (1,-1) {G};
            \draw (A) -- (B);
            \draw (A) -- (C);
            \draw (B) -- (D);
            \draw (C) -- (E);
            \draw (D) -- (E);
            \draw (E) -- (G);
            \draw (D) -- (F);
            \draw (F) -- (G);
        \end{tikzpicture}
        \caption{Example Network with Attributes}
        \label{fig:graph}
    \end{subfigure}
    \hfill
    \begin{subfigure}[b]{0.45\textwidth}
        \centering
        \[
        \begin{array}{c|ccc}
            & \cellcolor{tab10blue}\textcolor{white}{\textbf{Circle}} & \cellcolor{tab10orange}\textcolor{white}{\textbf{Square}} & \cellcolor{tab10magenta}\textcolor{white}{\textbf{Diamond}} \\
            \hline
            \cellcolor{tab10blue}\textcolor{white}{\textbf{Circle}}  & 2 & 2 & 0 \\
            \cellcolor{tab10orange}\textcolor{white}{\textbf{Square}}   & 2 & 1 & 2 \\
            \cellcolor{tab10magenta}\textcolor{white}{\textbf{Diamond}} & 0 & 2 & 1
        \end{array}
        \]
        \caption{Mixing Matrix}
        \label{fig:matrix}
    \end{subfigure}
    \caption{Network representation and its corresponding mixing matrix.} 
    \label{fig:graph-mixing-matrix-example}
\end{figure}

Exponential random graph models (ERGMs) encode network patterns through sufficient statistics and are more expressive than SBMs \cite{Hunter2008-ergm,Krivitsky2023-ergm,Keeling2005}. In particular, an ERGM which only uses the mixing matrix as its sufficient statistic is equivalent to an SBM. Consider a set of descriptive statistics given by a function $f_\mathrm{ERGM}(G)$ taking values in $\mathbb{R}^d$. For example,  $f$ might return a pair of values (so $d=2$) representing the number of edges and the number of triangles in the graph. An ERGM with parameter vector $\theta\in \mathbb{R}^d$ induces the  probability distribution \[\Pr\paren{G\mid \theta}  \propto \exp\paren{\theta^T f_\mathrm{ERGM}(G)} \, . \]

While SBMs capture only edge-based community structure, ERGMs can incorporate a broader set of structural dependencies, including degree distribution and homophily. 

In our experiments, we consider SBMs based on mixing matrix of a single node attribute, \emph{age} ($f_\mathrm{SBM}$). For ERGMs, we use various types of sufficient statistics ($f_\mathrm{ERGM}$): number of edges, number of nodes with a specific degree, nodematch (number of edges between nodes sharing the same attribute value), and nodefactor (number of edges incident to nodes with each attribute value). See \Cref{app:net-stats-defs} for formal definitions of these network statistics.
SBMs and ERGMs based on these statistics are commonly used in epidemiological modeling \cite{jenness2016effectiveness, jenness2016impact, goodreau2012concurrent, goodreau2012drives, jenness2017incidence, goodreau2017sources,nelson2020modeling,jenness2021dynamic}.

\subsection{Releasing Differentially Private Statistics}
\label{sec:dp-stats}

Differential privacy (DP) \cite{DworkMNS16} is a rigorous mathematical framework for releasing aggregate statistics from a dataset while protecting the privacy of individuals whose data is contained within the dataset. At its core, DP is a stability criterion\jsout{;}\jadd{:} it ensures that the output of an algorithm remains essentially the same even if one person’s data is arbitrarily substituted for another's. The philosophical argument for this protection relies on a comparison with the notion of non-participation. Consider a scenario where a person, say Alice's data is not present in a dataset; in this case, intuitively, Alice's privacy is protected as no statistic derived from the dataset can reveal information specifically about her. If an analysis is differentially private, the output produced when Alice contributes data is indistinguishable from the output produced when Alice does not contribute her data. This indistinguishability---an observer cannot distinguish Alice's presence from her absence (or substitution)---captures the notion of privacy that differential privacy provides. 

The level of privacy protection that differential privacy provides is parameterized by a quantity called the \emph{privacy budget} denoted by $\eps$ (epsilon) which takes values in $[0, \infty)$. This parameter controls the maximum level of distinguishability allowed between two datasets which are the same except for, one which includes Alice and the other, where someone else is included instead of Alice. Smaller the value of $\eps$, stronger the privacy protection. 

There is a vast toolbox of algorithms satisfying differential privacy developed through a large body of work over the last 20 years. This research has led to the creation of numerous differentially private algorithms for various types of statistics, each backed by utility guarantees and implemented in software libraries \cite{Holohan2019DiffprivlibTI,Gaboardi2020APF,Yousefpour2021Opacus,Berghel2022Tumult}. Consequently, differential privacy has moved from theory to significant real-world adoption, primarily by governments and major technology companies. Most notably related to this work are the 2020 U.S. Census release \cite{Abowd2018TheUC}, Israel's Live Birth Registry release \cite{Hod2025Differentially}, and Google's COVID-19 Community Mobility Reports \cite{aktay2022google}.

\subsubsection{Differential Privacy for Network Data}

There are two natural adaptations of differential privacy for network datasets: \emph{edge differential privacy} (\Cref{def:edge-dp}) and \emph{node differential privacy} (\Cref{def:node-dp}) (or, more concisely, \emph{edge-DP} and \emph{node-DP})~\cite{HayLMJ09}. For edge-DP, first investigated by Nissim et al.~\cite{NissimRS07}, the indistinguishability requirement applies to any two networks that differ in one edge. In contrast, for node-DP (formulated by \cite{HayLMJ09} and first studied by three concurrent works~\cite{BlockiBDS13,KasiviswanathanNRS13,ChenZ13}), the indistinguishability requirement applies to any two networks that differ in one node and all its adjacent edges. Node-DP is more suitable for datasets representing networks of people because, in this context, it protects each individual (node) and all their connections (adjacent edges).

The current work focuses on node-privacy given that we aim to protect all of the connectivity information pertaining to one individual in a human contact network. Although node-DP is more challenging to satisfy than edge-DP, since it is a more stringent requirement, we find that existing frameworks for node-DP noise addition are sufficient in terms of accuracy for the pipeline we design and study.

\subsubsection{Differential Privacy via Noise Addition}

A blueprint for achieving differential privacy is to introduce a controlled measurement error to a statistic of interest, namely, adding symmetric distributed random noise to the statistic's true value. The \textit{Laplace mechanism} (\Cref{thm:laplace}) is the classic approach to do so, though we note that other algorithms exist. The Laplace mechanism adds calibrated random noise to the true value $f(G)$, where the noise is sampled from a Laplace distribution with scale parameter $\text{GS}(f)/\eps$. Recall that $\eps$ is the privacy budget the tune the level of protection; $\text{GS}(f)$ denotes the \textit{global sensitivity} of the statistic $f$ (\Cref{def:GS})

The global sensitivity of a function $f$, in the context of node-DP, is defined as the maximum change in $f$'s output over all pairs of neighboring graphs---that is, over all possible graphs that differ by a single node (along with its incident edges). Formally, this is $\text{GS}(f) = \max_{G, G'} |f(G) - f(G')|$ where $G$ and $G'$ are neighboring graphs.
Intuitively, adding noise at the scale of the global sensitivity masks the contribution of any single node to the output of the function $f$.

The privacy parameter $\eps$ precisely tunes this masking, introducing a fundamental \textit{privacy-utility tradeoff}. Smaller $\eps$ makes it harder for an adversary to reliably learn whether a particular individual is present in the network or to gain new information attributable to their presence, regardless of what auxiliary information they possess, but this comes at the cost of adding more noise to $f(G)$ and reducing accuracy. From a theoretical perspective, values of $\eps \leq 1$ are generally considered to provide strong privacy protection. At the extremes: $\eps = 0$ yields perfect privacy but completely uninformative outputs, while $\eps = \infty$ releases the true statistic with no privacy protection. Larger $\eps$ improves the utility (accuracy) of the released statistic by reducing noise, but weakens privacy guarantees. Selecting an appropriate value of $\eps$ requires balancing these competing concerns based on the specific application and privacy requirements.

\subsubsection{Addressing Excessive Global Sensitivity}

For all statistics we consider in this work, the global sensitivity is unacceptably high. To illustrate this challenge, consider the mixing matrix entries as a concrete example.
The global sensitivity of an entry in the mixing matrix is $n-1$ where $n$ is the number of nodes in the network: the largest possible change in a mixing matrix entry due to a replacement of a node is from edge count 0 to $n-1$ (or vice versa).
Consider two groups where group $i$ is small (containing, say, a handful of nodes) and group $j$ is large. If these groups are initially disconnected ($M_{ij} = 0$), adding a single node to group $i$ that connects to all nodes in group $j$ can change the edge count $M_{ij}$ substantially. (The maximal change $n-1$ occurs when a group $i$ consists of a single node and a group $j$ of all the rest of nodes; but even with small groups of constant size, the sensitivity can be a large constant.)

This presents a fundamental obstacle: adding noise calibrated to a sensitivity of $n-1$ would completely obscure the true values of quantities that themselves lie in $[0,n-1]$, wiping out all information and rendering the released statistics that might be useless for downstream analysis.

To counter this effect, we preprocess the network before computing the network statistics by bounding node degrees. Specifically, we cap all node degrees at some value, called \emph{truncated degree} and denoted with $\Delta$, by carefully pruning the graph.\footnote{To cap degrees, we use the algorithm of \cite{DayLL16}, which has stability properties that are important for the final privacy guarantee.}
Returning to the mixing matrix example: with this degree bound in place, adding or removing a single node can contribute at most $\Delta$ edges between groups, rather than potentially connecting to all nodes in a large group. For a network of size $n$, this reduces the global sensitivity from $n-1$ to $\Delta$, fundamentally limiting the influence that adding or removing a single node can have on the statistic. When $\Delta \ll n$---a reasonable assumption for many real-world networks, such as human contact networks where individual degrees are typically bounded---this sensitivity reduction is dramatic, enabling substantially more accurate private releases.

Formally, this approach is an instance of the Lipschitz extension framework \cite{KasiviswanathanNRS13,ChenZ13}. Our algorithms take the truncated degree $\Delta$ as an additional parameter and satisfy the usual notion of differential privacy regardless of its value, but are only guaranteed to give accurate answers on graphs where all nodes have degree at most $\Delta$. We explore the role that $\Delta$ plays in the accuracy of epidemic simulations.

While this paper does not propose new differentially private algorithms, our pipeline demonstrates the value of developing algorithms tailored to specific epidemiological models. We provide a detailed description of the differentially private algorithms used in this work, together with all necessary formal definitions, in \Cref{app:dp-algs}.

\section{Experimental Design}
\label{sec:experiment}

\subsection{Sources of Error \& Experimental Conditions}
\label{sec:source-error}

Our differentially private pipeline pipeline introduces errors in its different steps.
To evaluate the accuracy of the simulations produced by the pipeline, we decompose the total error into distinct sources and types (biases and stochastic variation) and design experiments to measure each component independently.

\subsubsection{Sources of Error in the Pipeline}
\label{sec:sources-of-error}

\paragraph{Model Misspecification Error.}
Model misspecification occurs when the statistical model cannot fully capture the structural complexity of the observed network. This error arises from the choice of the model and its network statistics $f$, and is independent of differential privacy. When the model's summary statistics are insufficient, the resulting synthetic networks will systematically differ from the observed network. For example, a stochastic block model fitted only on age mixing matrices cannot capture concurrent partnerships or race homophily.

\paragraph{Transformation Error (Degree Truncation).}
Our DP algorithm requires projecting the input network into a degree-capped version controlled by the truncated degree parameter $\Delta$, a hyperparameter of the network statistic release algorithm that does not control the privacy level. When $\Delta$ is smaller than the true maximum degree, this preprocessing distorts network structure by removing edges from high-degree nodes. This error introduced before noise addition and is most pronounced when $\Delta$ is substantially smaller than the observed maximum degree.

\paragraph{Noise Error (Differential Privacy).}
To achieve differential privacy, calibrated noise is added to the graph statistics based on the privacy budget $\eps$. As discussed earlier, this introduces the privacy-utility tradeoff: smaller $\eps$ provides stronger privacy but more noise, while larger $\eps$ improves accuracy but weakens privacy.

\paragraph{Network Sampling Variance.}
Even when a statistical model is correctly specified and fitted, sampling different synthetic networks from the same model parameters introduces stochastic variation in network structure and subsequent epidemic outcomes.

\paragraph{Epidemic Simulation Variance.}
Disease spread is inherently stochastic. Running multiple epidemic simulations on the same network produces variation in epidemic outcomes due to the random nature of transmission events and disease progression.

\subsubsection{Analyst Perspective}
\label{sec:analyst-prespective}

Thanks to our experimental design (\Cref{fig:experimental-design}), discussed below, we are able to measure each source of error. However, an analyst who executes the pipeline (\Cref{fig:pipeline}) cannot assess all of them. In particular, an analyst who receives the differentially private summary statistics cannot characterize how the randomness introduced by differential privacy propagates through to impact the epidemic simulation---though they do know the distribution of noise added to those statistics. From the analyst's perspective, the pipeline begins at the point where they choose a model and fit it to the summary statistics. They can observe variation arising from sampling multiple networks from the fitted model and running simulations on those networks, but this represents only part of the true variation.

\subsubsection{Experimental Conditions}

We establish several experimental conditions to isolate each error source.

The \textbf{Observed Network} (\Cref{sec:observed-network}) serves as ground truth, generated using an ERGM with ARTNet parameters. This represents the ideal case with no modeling, transformation, or privacy error.

The \textbf{Without DP} condition represents standard epidemiological practice, where models (SBM or ERGM) are fitted directly on statistics from the observed network. This contains only model misspecification error (it presents), eliminating transformation and noise errors.

The \textbf{With DP} conditions introduce privacy protection with $\eps \in \{0.5, 1, 5, 10, \infty\}$ and $\Delta \in \{2, 3, 4, 5\}$. We generate five independent DP releases per configuration (model family, $\eps$, $\Delta$). These conditions contain all three error sources. Notably, the case $\eps = \infty$ applies degree truncation without noise, isolating transformation error alone.

\subsubsection{Comparisons}

We test two model families to understand misspecification effects. The stochastic block model (SBM) uses only the age mixing matrix, creating intentional misspecification---it cannot capture concurrent partnerships, race homophily, or detailed degree distributions. The ERGM specification matches the terms used to generate the observed network, creating minimal misspecification by design.

Comparing \textbf{Observed Network to Without DP} isolates model misspecification error and establishes reference variability from network sampling and epidemic simulation.

Comparing \textbf{Without DP to With DP} isolates the combined error from differential privacy (transformation plus noise). The special case $\eps = \infty$ is particularly informative, revealing transformation error alone by applying degree truncation without noise.

To quantify variance contributions, we execute each step---DP release, network sampling, and epidemic simulation---multiple times. This allows us to measure not only systematic biases but also the relative magnitude of variance from DP release variability, network sampling variance, and epidemic simulation variance.

\subsection{Our Experimental Setup}
\label{sec:setup}

\subsubsection{Observed Network}
\label{sec:observed-network}

The observed network (``ground truth'') used in our pipeline contains 10,000 nodes and is generated from an ERGM with target statistics estimated from the ARTNet study \cite{weiss2020egocentric}, specifically the subset of respondents residing in Atlanta. This dataset represents a large convenience sample of MSM who responded to a survey that included questions about their individual demographic characteristics, the characteristics of their sexual partners, and the frequency of several sexual behaviors. These egocentric network data allow researchers to derive target statistics that describe the demographics, degree distribution, concurrency, and other attributes of a sexual network, making it valuable for epidemiological modeling. In the literature, it is used for studying the transmission of HIV and other sexually transmitted infections (STIs) \cite{Jenness2022-role,Onwubiko2025-optimizing,Maloney2020-projected}.

We opt for using an observed network generated from this data for three primary reasons. First, the ARTNet survey provides detailed cross-sectional information on sexual partnerships among MSM in the United States. These egocentric network data are well-suited for specifying the network models we rely upon in our simulation scenarios. Second, real-world large-scale human contact networks suitable for this research are not publicly available---a limitation that fundamentally motivates this work. Such networks contain highly sensitive personal information, and sharing them without adequate privacy protections risks substantial privacy harms to individuals. Finally, because we are interested in part in understanding the error introduced by model misspecification relative to that introduced by DP, we require control over all factors governing the network structure. Taking as our ``ground truth'' synthetic networks generated from a known ERGM allows us to isolate and quantify the impact of modeling assumptions systematically, which would be impossible with networks of unknown or partially specified generative processes. 

\subsubsection{Simulation of Disease Spread}
\label{sec:simulation-technical}

\paragraph{Epidemic Models}
\label{sec:epidemic-types}

Throughout our experiments, we simulate the spread of an \textit{SIS} (susceptible-infected-susceptible) pathogen using a network ABM. In contrast to compartmental models, which model aggregate transition rates between subgroups of the population within which contacts between individuals are assumed to occur at random (\textit{i.e.}, a ``random mixing'' assumption that leads to ``mass action'' dynamics \cite{Anderson2000-mathematical}), ABMs represent individuals as discrete entities \cite{Epstein1999-agentbased,Macal2016-everything} . These discrete entities may possess their own demographic, disease-related, and/or other attributes, some of which may change over time. In this way, an ABM models a virtual population of individuals and their interactions such that epidemic dynamics emerge in a ``bottom-up'' fashion \cite{Epstein1999-agentbased,Macal2016-everything}. Some ABMs model interaction between agents algorithmically \cite{Macal2016-everything}. Others, such as the network ABMs we use, rely upon a statistical representation of the whole network to generate persistent or transient links between individual agents so as to preserve network-level properties important to the phenomenon under study \cite{Jenness2017EpiModelAR}.

\paragraph{Simulation Protocol}
\label{sec:simulation-protocol}

To simulate an $SIS$ epidemic, the the ABM proceeds in discrete time steps, each representing a single week. Let $G$ be the underlying network of contacts in the population. At any given time step, the following actions take place:

\begin{itemize}
    \item \textbf{Infect step}: Each infected person $u \in I$ in the population can infect other susceptible people $v \in S$ in its neighborhood (i.e., $u$ and $v$  share an edge $\{u,v\} \in E(G)$) with probability $p_{\mathrm{inf}}$ called the \textit{probability of transmission}, independent of time and other nodes in the network. Each pair of agents is assumed to engage in a single contact during each time step.

    \item \textbf{Recover step}: Each infected person in the population has a chance to recover from the disease independently with probability $p_{\mathrm{recov}}$ called the \textit{recovery rate} and to become susceptible again.
\end{itemize}

Upon initializing the model we randomly infect a fraction of nodes in $G$. This initial fraction is a parameter of our simulation called the \emph{initial prevalence}. We then allow the pathogen to spread in the population for 500 time steps, also called the \emph{burn-in} period, bringing the epidemic to a stable equilibrium in the population. We then run the simulation for 100 more time steps and collect the model's epidemic output during this time period, which we refer to as the \emph{analytic window}. After each time step, we compute the following epidemic statistics:

\begin{itemize}
    \item \textbf{Prevalence}: Fraction of people who are infected at the current time step.

    \item \textbf{Incidence rate}: The number of people newly infected during the current time step divided by the total number of susceptible people at the end of the previous time step.
\end{itemize}

We average these epidemic statistics over the analytic window within each simulation run, eliciting the absolute measures of interest for each realization of the model.

\paragraph{Interventions}
\label{sec:interventions}

Often, researchers conducting epidemic simulations are interested in the expected effect of a given intervention on disease transmission or other outcomes. Rather than absolute measures of prevalence or incidence, these studies may focus on relative measures that compare the intervention scenario with a ``baseline scenario''. The description we provide in the prior section constitutes our baseline scenario. We also simulate an intervention scenario called \emph{test and treat}\footnote{\url{https://github.com/shlomihod/epidp/blob/main/R/z_scenario_test_and_treat.R} which is based on \url{https://github.com/EpiModel/EpiModel-Gallery/tree/main/2018-08-TestAndTreatIntervention}} in which individuals may receive diagnostic testing and, if diagnosed, receive treatment that increases their rate of recovery from the infection. In this scenario, the intervention is implemented at the beginning of each simulation run, during which we allow the epidemic to evolve for 500 steps and then collect model output over the subsequent 100 time steps, as before. Therefore, the simulations under the \emph{test and treat} scenario also reach an approximately stable equilibrium by the time the analytic window begins (Figure \ref{fig:epicurve}). The intervention scenario introduces three new parameters to the simulations:
\begin{itemize}
    \item \emph{test rate}, the probability with which eligible individuals in the population are chosen to be tested if they are infected.

    \item \emph{test duration}, once diagnosed, how long until they are eligible to tested again.

    \item Diagnosed individuals have an increased probability of recovery within each time step, called the \emph {recovery rate with treatment}.
\end{itemize}

When the intervention is active, the following actions take place at each time step. Each eligible individual $u$ from the population $V$ is selected independently with probability \emph{test rate}, and those selected are tested. If $u$ is infected the individual is put on treatment for \emph{test duration} time steps. During this time $u$ is not eligible to be tested again, and their rate of recovery increases to \emph {recovery rate with treatment} from the normal \textit{recovery rate}. If an agent on treatment does not recover within the \textit{test duration} window, they revert to "undiagnosed" status and to the normal \textit{recovery rate}.

To quantify the effect of the intervention, we calculate the \textbf{Prevalence Ratio}, defined as the ratio of disease prevalence under intervention to disease prevalence under the baseline scenario and calculated for each (network, simulation) pair as the average prevalence during the analytic window in the \textit{test-and-treat} scenario over the average prevalence during the analytic window in the \textit{baseline scenario}. We also calculate the \textbf{Incidence Rate Ratio} as the comparison of average incidence rates in an analogous manager (results shown in Appendix \ref{app:results}).

\subsubsection{Putting it all Together}
\label{sec:all-together}

\begin{table}[t]
\centering
\caption{Epidemic simulation parameters.}
\label{tab:simulation-parameters}
\begin{tabular}{@{}lll@{}}
\toprule
\textbf{Parameter} & \textbf{Symbol} & \textbf{Value} \\
\midrule
\multicolumn{3}{@{}l}{\textit{General Parameters}} \\
\midrule
Time step duration & -- & 1 week \\
Initial prevalence & -- & 0.2 \\
Burn-in period & -- & 500 \\
Analytic window & -- & 100 \\
\midrule
\multicolumn{3}{@{}l}{\textit{Settings}} \\
\midrule
High transmission \small{(target $\approx 12\%$ prevalence)} & $p_{\mathrm{inf}}$ & 0.75 \\
Low transmission \small{(target $\approx 2\%$ prevalence)} & $p_{\mathrm{inf}}$ & 0.05 \\
\midrule
\multicolumn{3}{@{}l}{\textit{Baseline Scenario Parameters}} \\
\midrule
Recovery rate & $p_{\mathrm{recov}}$ & 0.1 \\
\midrule
\multicolumn{3}{@{}l}{\textit{Intervention Scenario Parameters}} \\
\midrule
Test rate & -- & 0.1 \\
Test duration & -- & 2 \\
Recovery rate (treated) & -- & 0.5 \\
\bottomrule
\end{tabular}
\end{table}

\begin{figure}[t]
\includegraphics[width=\linewidth]{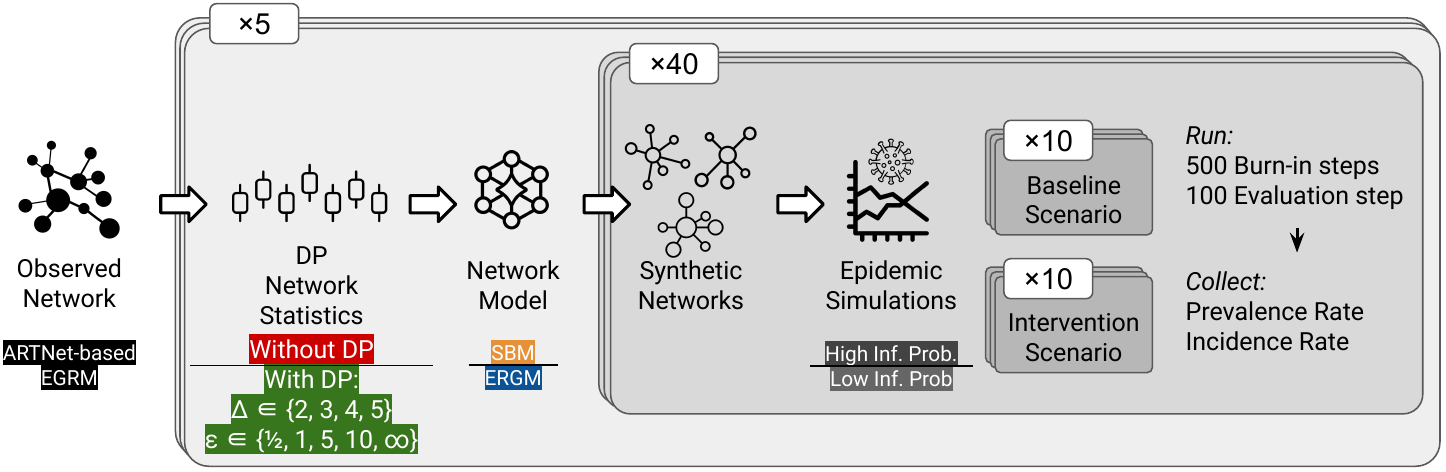}
\centering
\caption{An illustration of the experimental design to evaluate our pipeline as described in \Cref{sec:all-together}. The reader may find it helpful to contrast this diagram with \Cref{fig:pipeline}, which illustrates the pipeline in use.}
\label{fig:experimental-design}
\end{figure}

We conduct four experiment configurations based on two factors: the disease transmission probability and the graph model used for synthetic network generation.

For disease transmission, we consider a high prevalence scenario with infection probability $p_{\text{infection}} = 0.75$ and a low prevalence scenario with $p_{\text{infection}} = 0.05$. In both scenarios, the initial prevalence was set to $0.2$, and the recovery rate is $0.1$. As discussed in \Cref{sec:simulation-technical}, the intervention scenario has additional parameters. Test rate is set to $0.1$, test duration is set to $2$, and recovery rate with treatment is set to $0.5$. Simulation parameters are detailed in \Cref{tab:simulation-parameters}.

For network models, we use either stochastic block models (SBM) or exponential family random graph models (ERGM), following standard epidemiological practice. These two factors yield four experiment configurations: high prevalence with SBM, high prevalence with ERGM, low prevalence with SBM, and low prevalence with ERGM.

Within each experiment configuration, we assess the impact of differential privacy on the accuracy of the epidemic simulation by varying two parameters: the privacy parameter $\eps \in \{1/2, 1, 5, 10, \infty\}$ and the truncated degree parameter $\Delta \in \{2, 3, 4, 5\}$. For each combination of $\eps$ and $\Delta$, we produce 5 DP releases of the network statistics (as detailed in \Cref{sec:net-gen-models} and \Cref{app:net-stats-defs}). For each such DP release, we fit either a SBM or an ERGM on the private graph statistics computed---using different statistics for each model type---and then sample 40 synthetic networks from the fitted model. For each such synthetic network, we run two simulation scenarios, one with and without intervention. For each scenario, we run 10 disease spread simulations that first run for 500 burn-in time steps and then 100 time steps that are used to collect data for the experiment. In each such run, we store the prevalence and incidence rate at the population level and at the subgroup levels for both attributes \emph{age} and \emph{race}.

As baselines for comparison, we also run the complete experimental procedure described above on two settings: the true observed network and the no-DP condition (no degree truncation and using exact network statistics). For both baselines, we follow the same protocol of sampling 40 synthetic networks (for the no-DP case), running both intervention and non-intervention scenarios, and conducting 10 disease spread simulations per scenario with the same burn-in and data collection periods.

The experiments were run on Boston University Shared Computing Cluster (SCC)\footnote{\url{https://www.bu.edu/tech/support/research/computing-resources/scc}} for approximately $17,000$ CPU hours.

\section{Findings}
\label{sec:findings}

In this section, we describe our findings regarding the utility evaluation of our pipeline (\Cref{sec:pipeline}) based on our experimental design (\Cref{sec:experiment}). Recall that our objective is to study the impact of applying differential privacy to epidemiological simulations, so our findings focus on systematic comparisons with and without privacy protection. Here we present the results for prevalence only, as the incidence rate results follow a similar pattern and are presented in \Cref{app:results}.

In this section, we describe our findings regarding the utility evaluation of our pipeline (\Cref{sec:pipeline}) based on our experimental design (\Cref{sec:experiment}). Recall that our objective is to study the impact of introducing of differntial privacy to an epidemgoical simulation, so the findings centers around carefully comparing simulations with and withut diferential privacy. Heree present results for prevalence only, as the incidence rate results follow a similar pattern and are presented in \Cref{app:results}.

\subsection{Network Model Misspecification}
\label{sec:find1}

\begin{findingbox}{1}
If present, network model misspecification is the largest source of bias, which is unrelated to differential privacy. 
\end{findingbox}

\begin{figure}[h!]
    \centering
    \begin{subfigure}[b]{\textwidth}
        \centering
        \includegraphics[width=\textwidth]{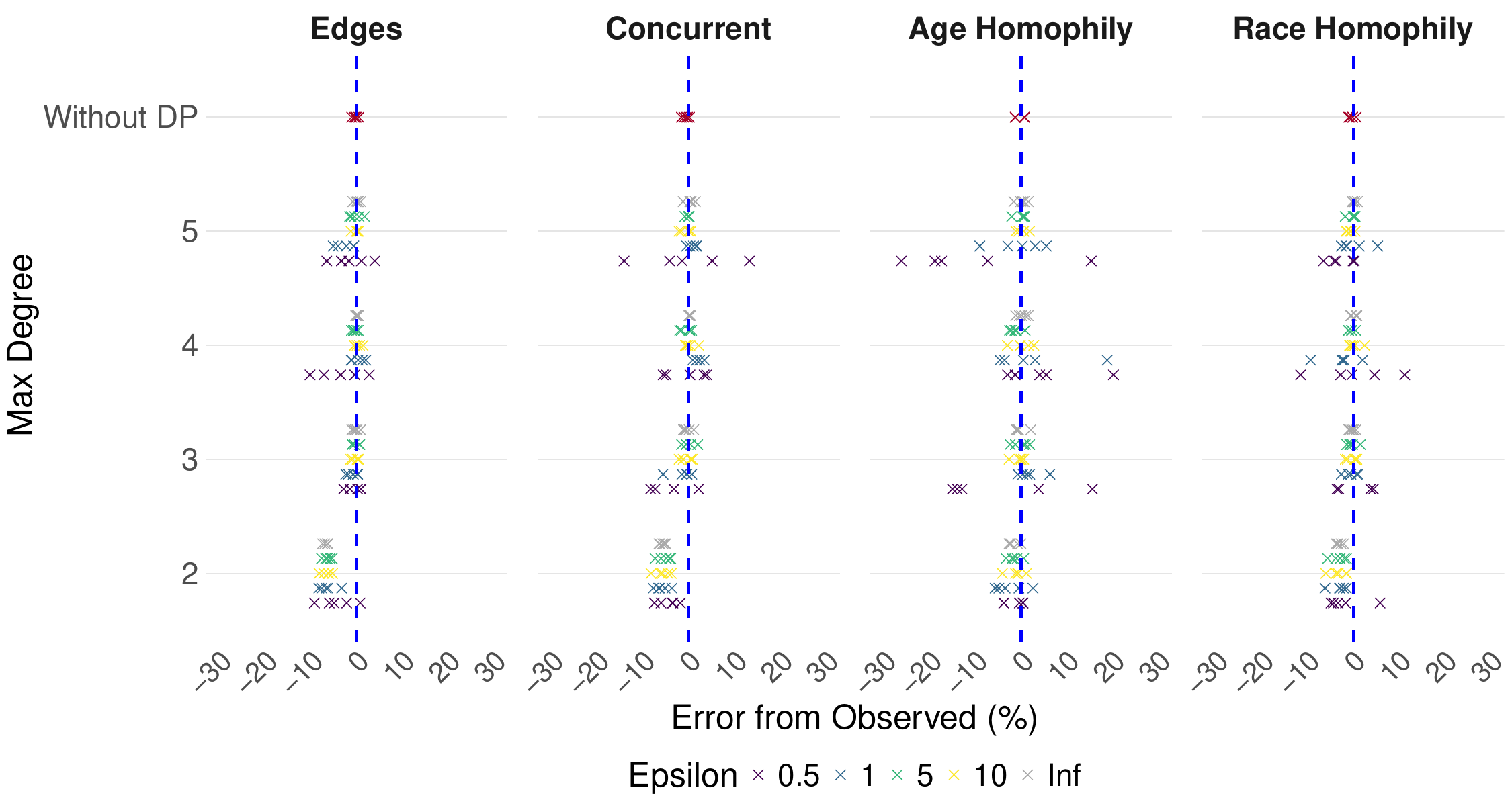}
        \caption{ERGM}
        \label{fig:quality_metrics_ergm}
    \end{subfigure}
    \hfill
    \begin{subfigure}[b]{\textwidth}
        \centering
        \includegraphics[width=\textwidth]{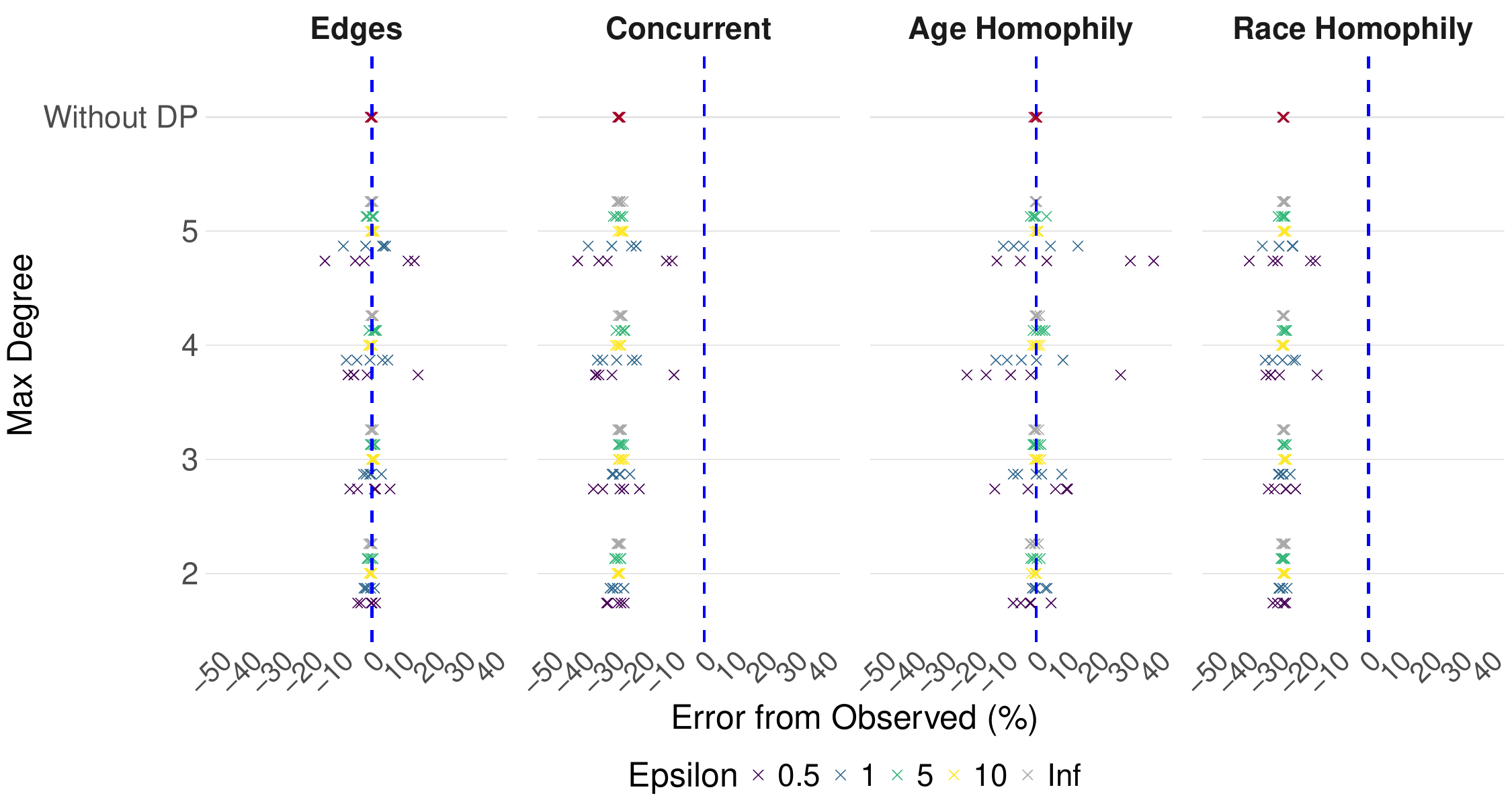}
        \caption{SBM}
        \label{fig:quality_metrics_bm}
    \end{subfigure}
    \caption{Edge-level network statistics presented on a relative scale, showing percentage difference from the observed network baseline (\textcolor{blue}{dashed line}). Each point represents the relative deviation of synthetic network statistics compared to the observed network value. Both plots refers to the network used in the low prevalence configuration.} 
    \label{fig:quality_metrics}
\end{figure}

Our pipeline samples synthetic networks from a statistical model trained on network statistics calculated from the observed network. If the model is misspecified, that is---cannot capture the complexity of the observed network---the synthetic networks may have different structural properties and produce inaccurate epidemic statistics.
To evaluate  this effect, we compare the observed network with synthetic networks generated from models trained \emph{without differential privacy} on two criteria: (1) network statistics; and (2) epidemic statistics. We test two types of models: (1) ERGMs that match the specification (terms)  used to generate the observed network from the ARTNet study; and (2) misspecified SBMs  capture only the mixing matrix age.

Overall, we observe that \textbf{specifying the model correctly is crucial for an unbiased estimation of graph statistics and epidemic simulation}.

Networks sampled from ERGMs without differential privacy accurately capture on average the structure of the observed network, as measured with network statistics (see \Cref{fig:quality_metrics_ergm}).

Similarly, ERGMs trained without differential privacy provide an accurate estimation, compared to simulations on the observed network, of the prevalence ratio (\Cref{fig:prevalence_ratio_ergm_high,fig:prevalence_ratio_ergm_low}), and the prevalences per scenario, baseline (\Cref{fig:prevalence_baseline_ergm_high,fig:prevalence_baseline_ergm_low}) and intervention (\Cref{fig:prevalence_intervention_ergm_high,fig:prevalence_intervention_ergm_low}).

In contrast, the SBM, which is only fitted on the \emph{age mixing matrix}, fails to capture structural features that are not specified. \Cref{fig:quality_metrics_bm} shows that both edge count and age homophily, which can be directly derived from the mixing matrix, do not present bias, while concurrent and race homophily demonstrate a large gap. This translates into modest discrepancies in epidemic statistics. In the prevalence ratio metric we see a max absolute gap of 0.06 between the observed and without-DP results (\Cref{fig:prevalence_ratio_bm_high,fig:prevalence_ratio_bm_low}). Similarly, the maximum absolute gap for the
per-scenario prevalence metric is 0.04 (\Cref{fig:prevalence_baseline_bm_high,fig:prevalence_baseline_bm_low}; \Cref{fig:prevalence_intervention_bm_high,fig:prevalence_intervention_bm_low}).

We  also observe similar patterns in degree statistics, and we discuss this in \Cref{app:quality-metrics-degree}.

\subsection{Bias in Epidemic Analysis}
\label{sec:find2}

\begin{findingbox}{2}
Differential privacy generally does not introduce new bias to absolute and relative epidemic analyses.
\end{findingbox}

\begin{figure}[h!]
    \centering

    \begin{subfigure}[b]{0.48\textwidth}
    \centering
    \includegraphics[width=\textwidth]{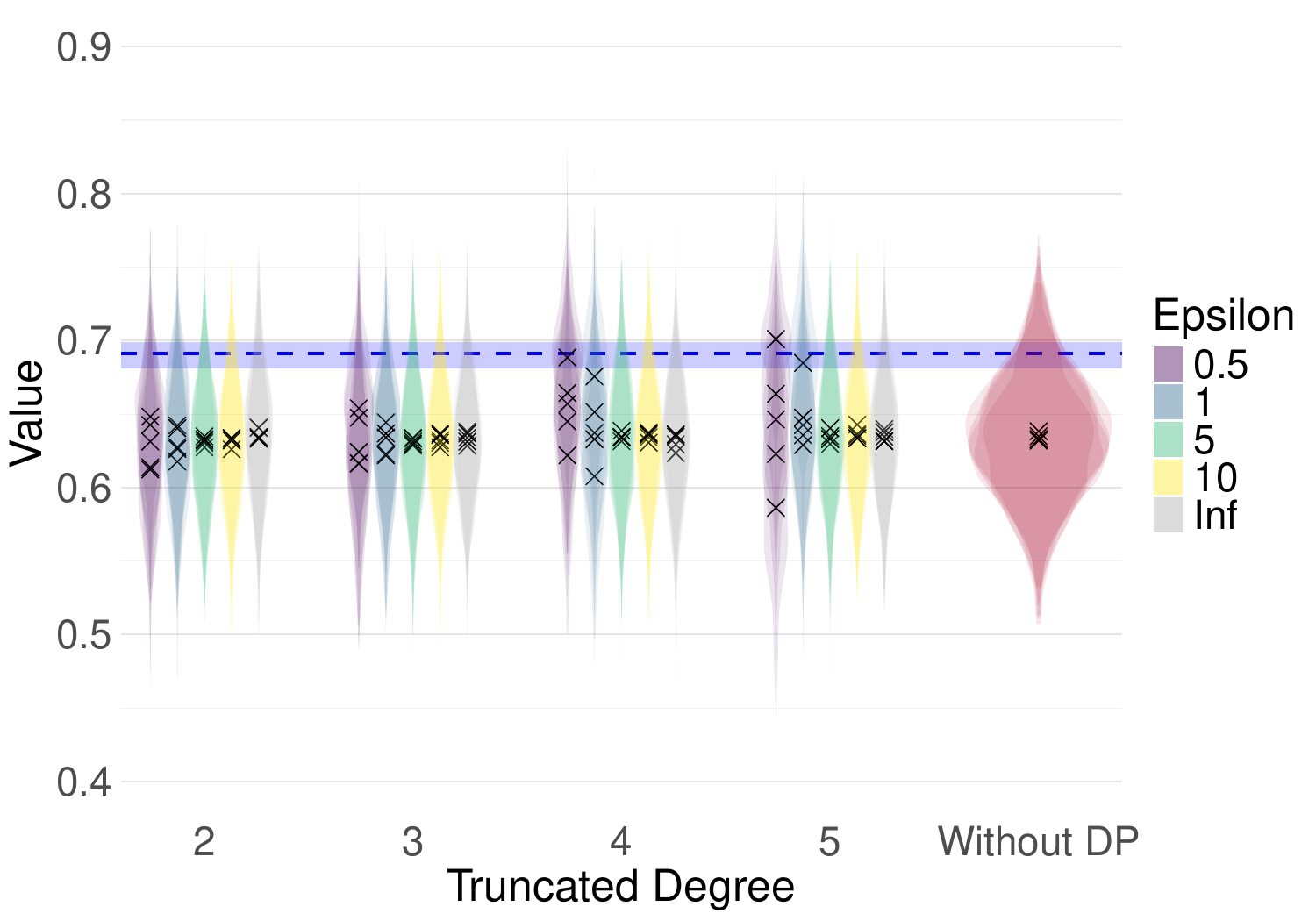}
    \caption{SBM - High Prevalence}
    \label{fig:prevalence_ratio_bm_high}
    \end{subfigure}
    \hfill
     \begin{subfigure}[b]{0.48\textwidth}
    \centering
    \includegraphics[width=\textwidth]{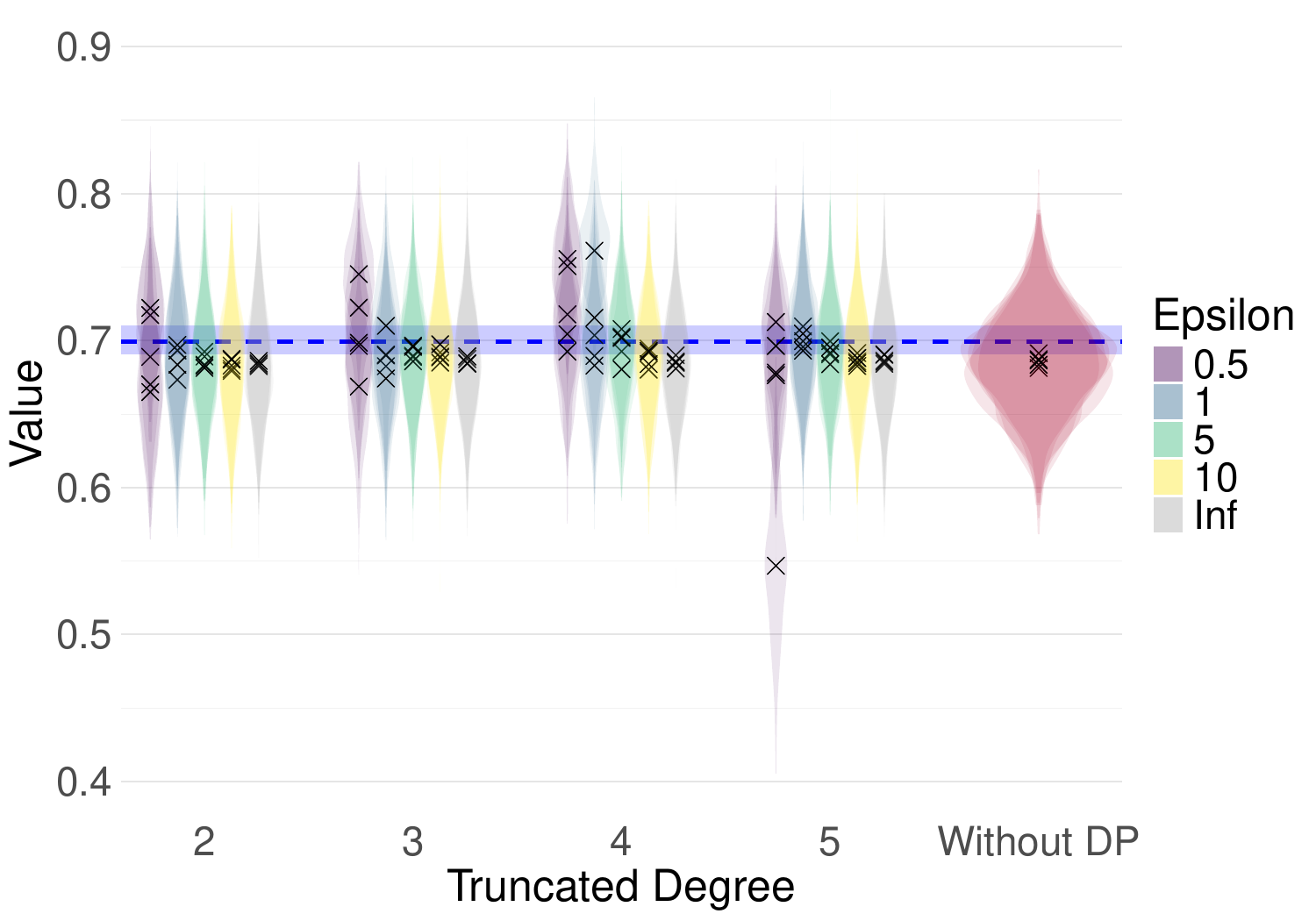}
    \caption{ERGM - High Prevalence}
    \label{fig:prevalence_ratio_ergm_high}
    \end{subfigure}

    \vspace{0.5cm}
    
    \begin{subfigure}[b]{0.48\textwidth}
        \centering
        \includegraphics[width=\textwidth]{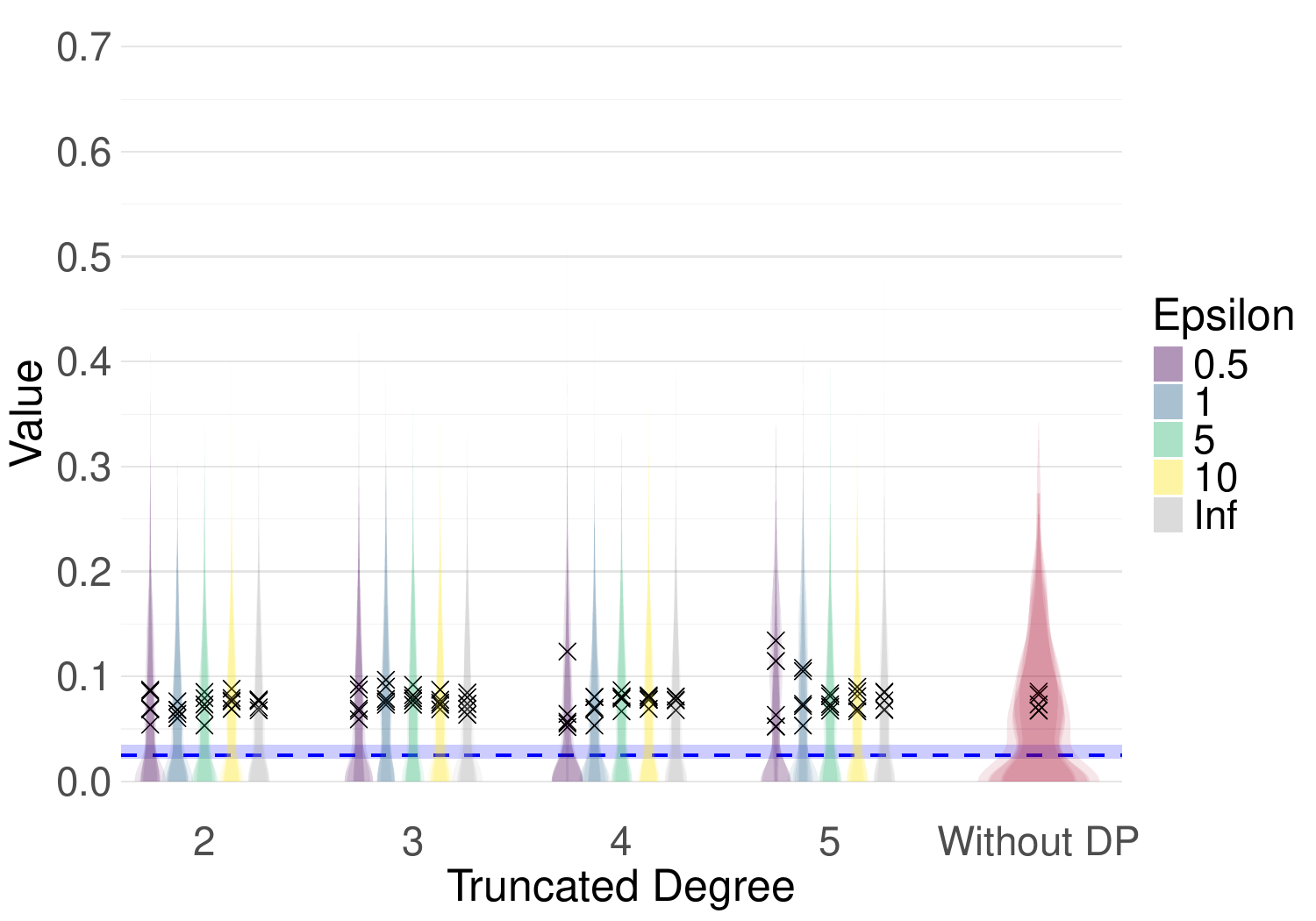}
        \caption{SBM - Low Prevalence}            
        \label{fig:prevalence_ratio_bm_low}
    \end{subfigure}
    \hfill
    \begin{subfigure}[b]{0.48\textwidth}
        \centering
        \includegraphics[width=\textwidth]{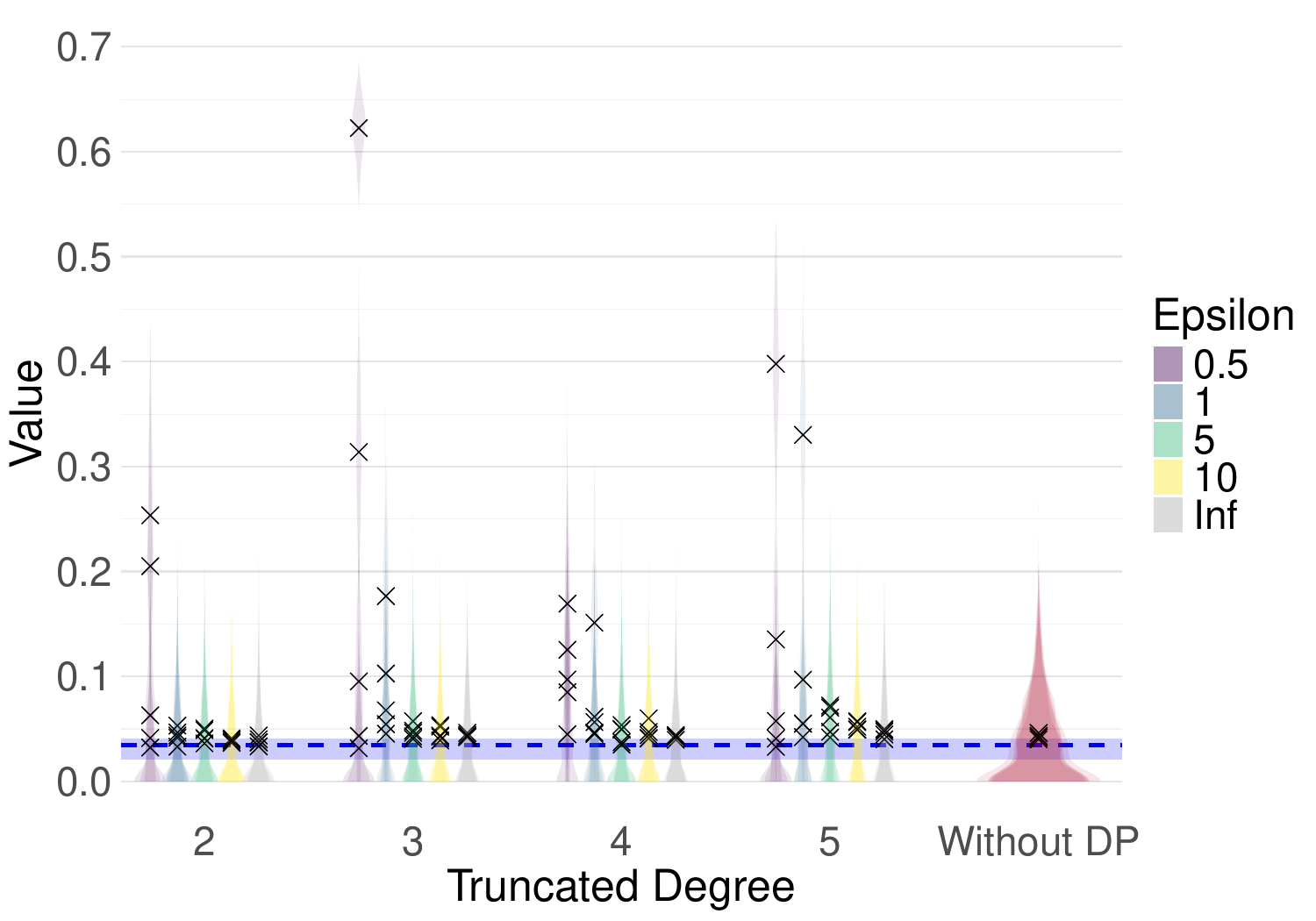}
        \caption{ERGM - Low Prevalence}
        \label{fig:prevalence_ratio_ergm_low}
    \end{subfigure}

    \caption{Prevalence ratios comparing intervention effects across different modeling approaches and prevalence conditions.}
    \label{fig:prevalence_ratio}
\end{figure}


      
    
  

Recall that our algorithm to compute network statistics with differential privacy: (1) preprocessing the the input network into a degree-capped network (controlled by truncated degree $\Delta$); and (2) calculating a network statistic with added calibrated noise (controlled by privacy budget $\eps$). To study the effect of additional error that is introduced due to differential privacy, either by degree truncation or noise addition (and not due to modeling choice), we compare synthetic networks trained with differential privacy and without differential privacy, with respect to the parameters $\Delta$ and $\eps$.

If the DP's parameter truncated degree $\Delta$ is larger than 2,
except for a few cases where the privacy budget $\eps$ is equal to or less than 1, there is no meaningful bias introduced in either ERGM or BM.
This observation holds across all of our metrics: network statistics (\Cref{fig:quality_metrics}), prevalence ratio (\Cref{fig:prevalence_ratio}), and per-scenario prevalence (\Cref{fig:prevalence_baseline,fig:prevalence_intervention}). This is apparent  when comparing the cluster of \emph{Without DP} against the various other \emph{with differential privacy} (per $\Delta$ and $\eps$).

There are some parameter settings where some bias in reported degree statistics is noticeable, but this does not translate into bias in epidemic simulation. 
When the truncated degree $\Delta$ is equal to 2, ERGMs show bias in network statistics, as the degree-capping preprocessing distorts the observed network which has a true maximum degree equal to 3, leading to loss of edges (see node degree distribution in \Cref{fig:node_metrics}). Nevertheless, the simulations are unaffected.
(In SBMs, neither graph nor epidemic statistics differ across values of truncated degree $\Delta$.)
In a few cases, at low values of $\eps$, we also observe a small bias due to clipping of noisy values that are out of range; see \Cref{app:releasing-scalar-stats}.

\subsection{Per-Group Epidemic Analysis }
\label{sec:find3}

\begin{findingbox}{3} 
Epidemic analysis is accurate also at the granular level per group, when the model is correctly specified.
\end{findingbox}

In \Cref{sec:find1} and \Cref{sec:find2}, our focus was on prevalence computed from epidemic simulations at the population level, not specific to any subgroup i.e. people belonging to a particular age or race group. However, accuracy at the population level does not imply accuracy at the level of demographic subgroups.

In our experiments, we observe that overall, \textbf{the sources of error at the granular sub-group level follow the same trend as that at the population level.} However, there are some additional observations that are worth discussing.

\begin{figure}[h!]
    \centering
    \begin{subfigure}[b]{0.48\textwidth}
        \centering
        \includegraphics[width=\textwidth]{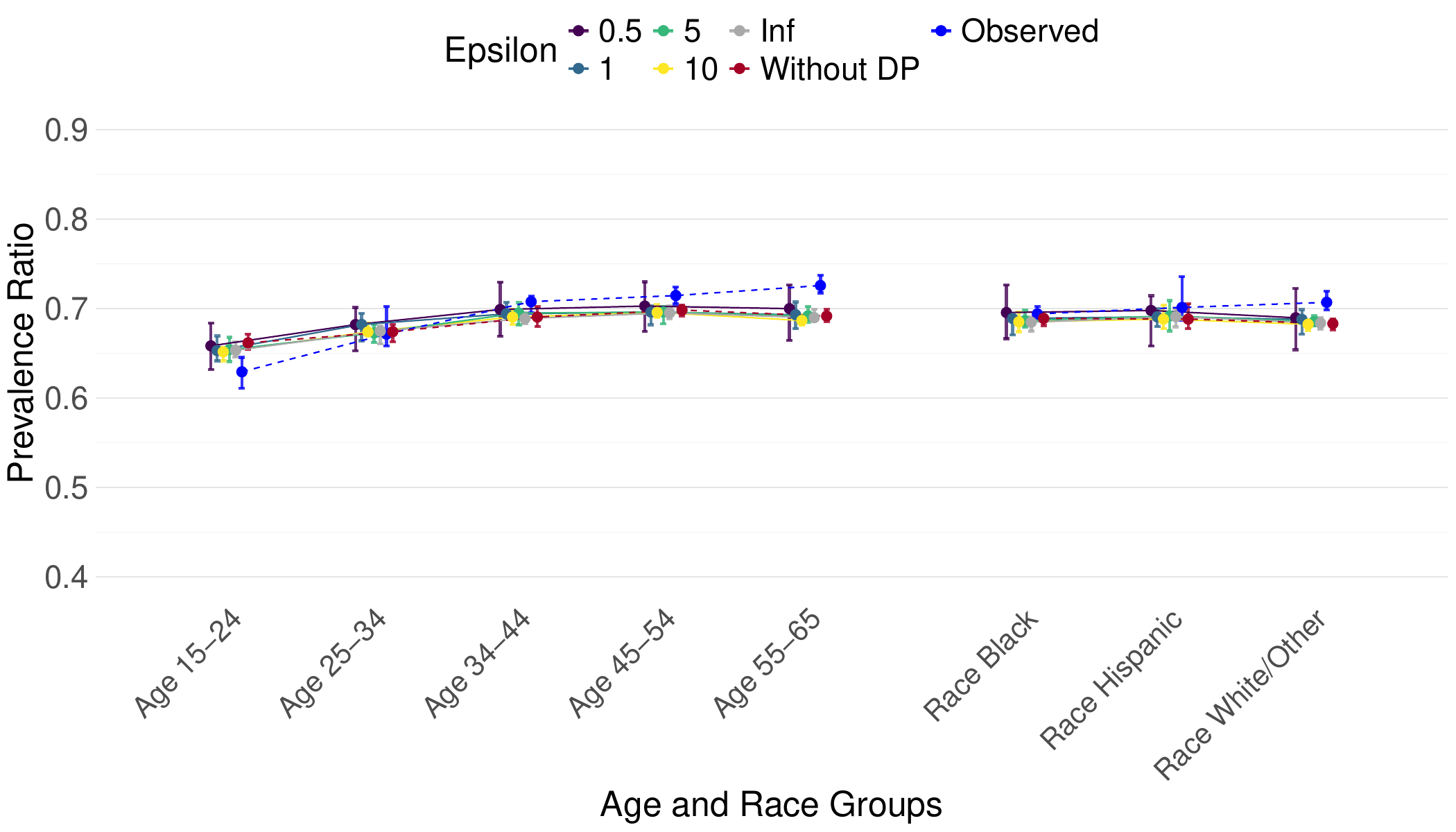}
        \caption{Truncated Degree 2}
        \label{fig:granular_ergm_high_truncated_2}
    \end{subfigure}
    \hfill
    \begin{subfigure}[b]{0.48\textwidth}
        \centering
        \includegraphics[width=\textwidth]{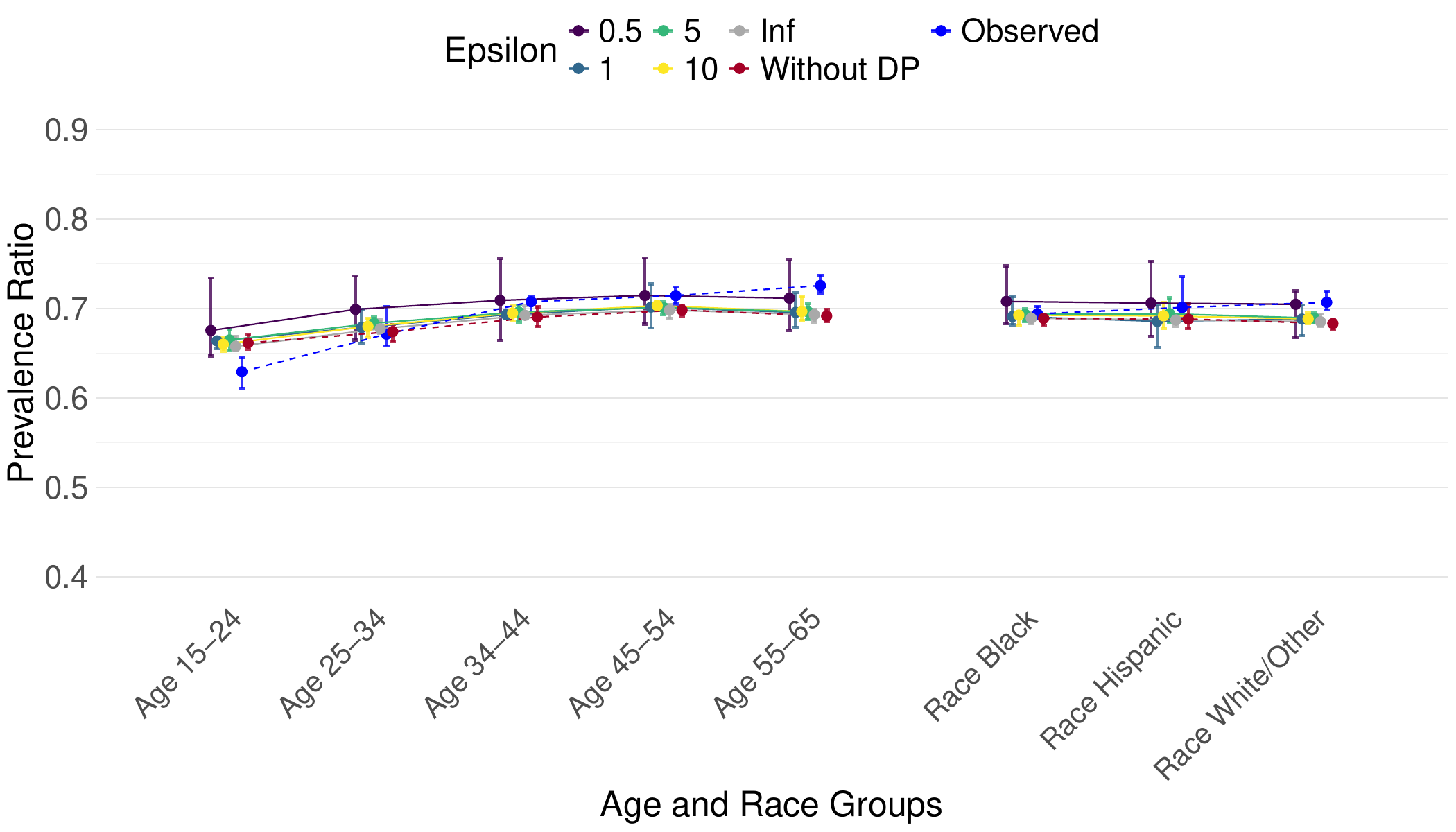}
        \caption{Truncated Degree 3}
    \end{subfigure}
    
    \vspace{0.5cm}
    
    \begin{subfigure}[b]{0.48\textwidth}
        \centering
        \includegraphics[width=\textwidth]{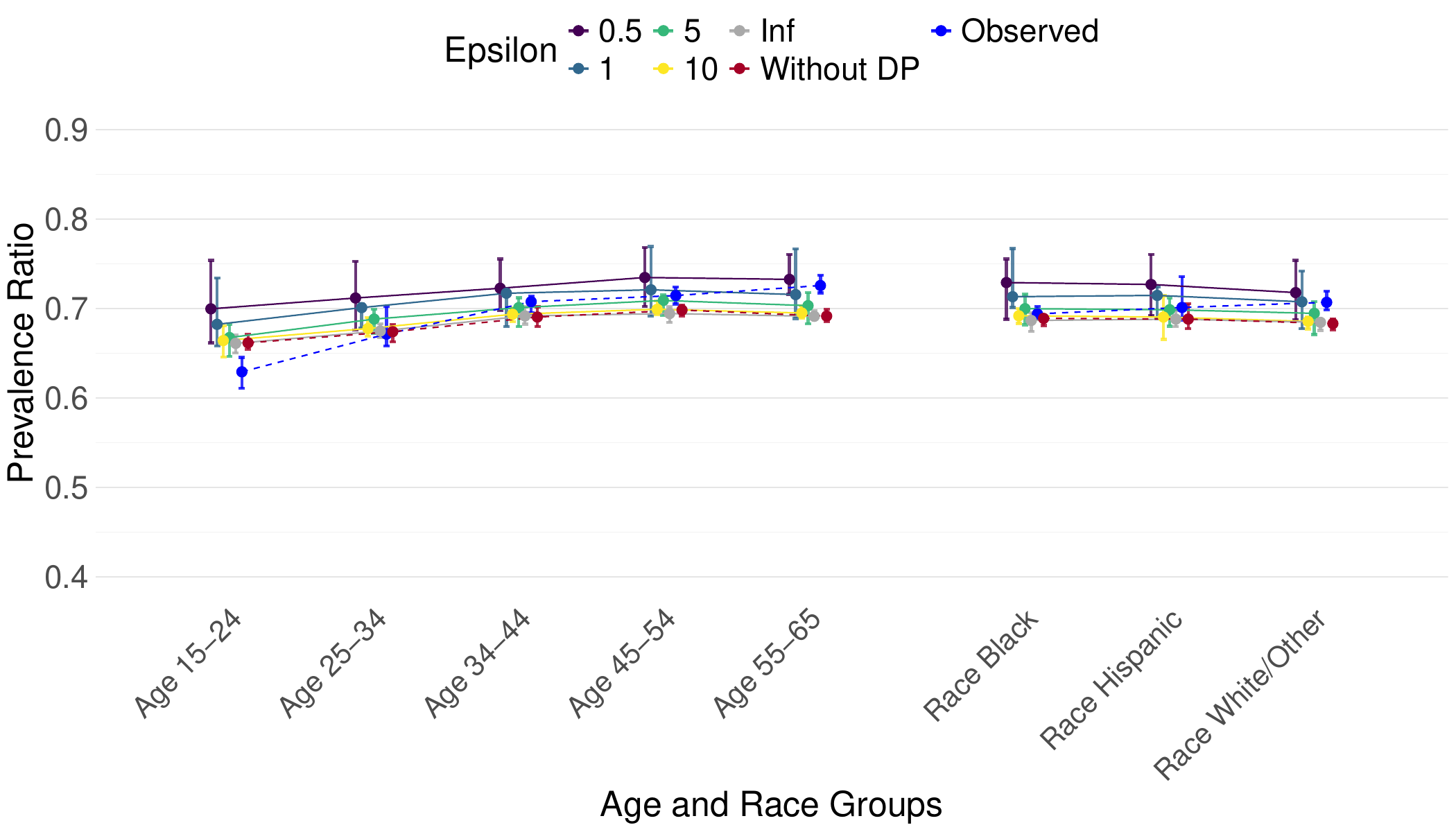}
        \caption{Truncated Degree 4}
    \end{subfigure}
    \hfill
    \begin{subfigure}[b]{0.48\textwidth}
        \centering
        \includegraphics[width=\textwidth]{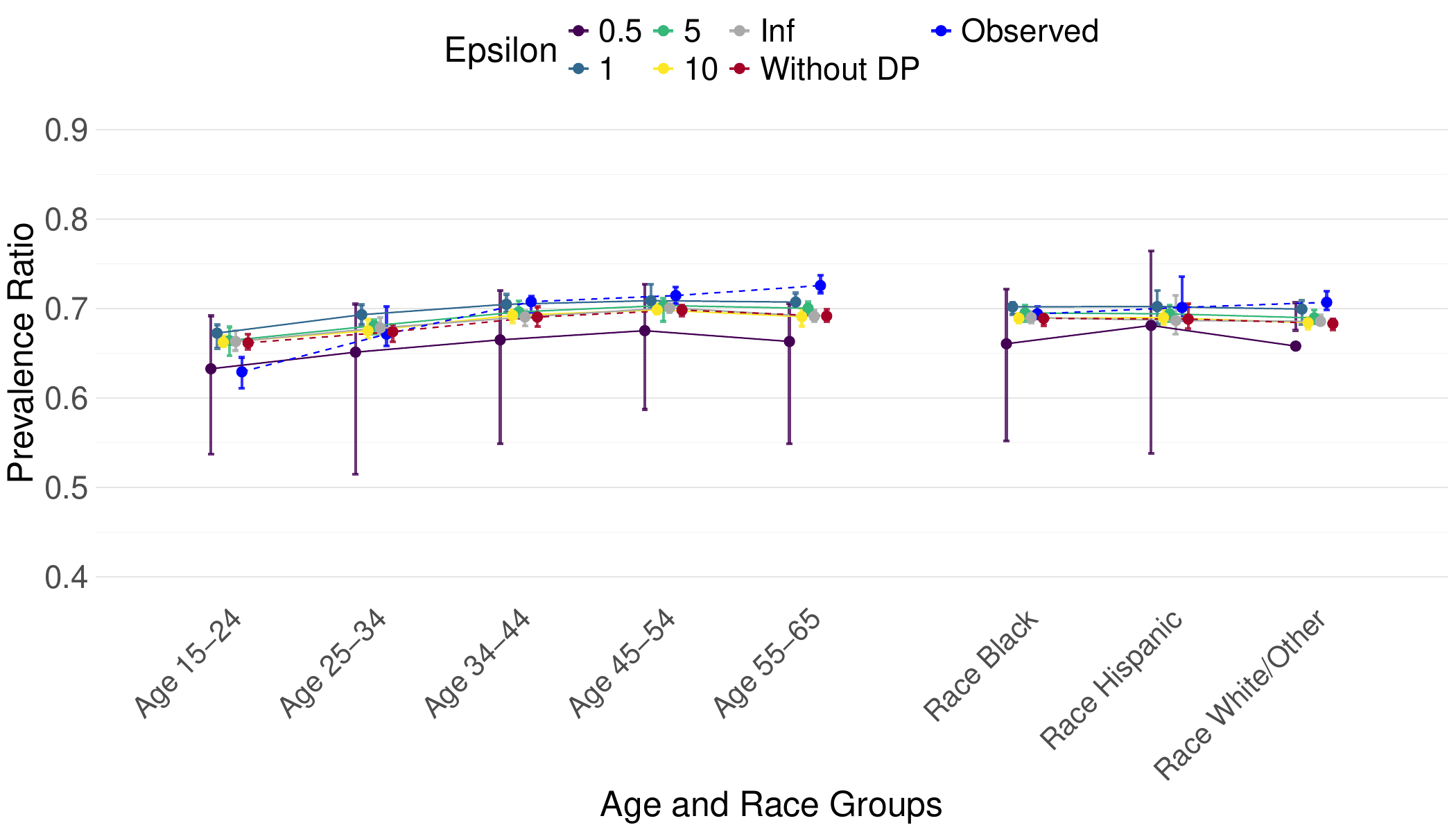}
        \caption{Truncated Degree 5}
    \end{subfigure}
    \caption{Granular analysis by age and race for ERGM high prevalence condition across different maximum degree constraints. }
    \label{fig:granular_ergm_high}
\end{figure}

\subsubsection{Effect of Modeling}
\label{sec:granular_model_misspecification}

To evaluate the effect of modeling at the subgroup level we compare the per-group epidemic statistics derived from observed network simulation against those from synthetic networks trained \emph{without differential privacy}. Similar to \Cref{sec:find1}, we look at two models (1) ERGMs that match the specification (terms) used to generate the observed network; and (2) SBMs that are misspecified and capture only the mixing matrix of the age attribute.

\paragraph{\textbf{Model Misspecification.}} 
To study the effect of model misspecification, one can focus on any one of the four subplots (here we use (a)) in Figures \ref{fig:granular_ergm_high}, \ref{fig:granular_ergm_low}, \ref{fig:granular_bm_high} and \ref{fig:granular_bm_low}. The four different subplots will become relevant when we discuss bias due to differential privacy in \Cref{sec:granular_dp_bias}. \Cref{fig:granular_ergm_high_truncated_2} and \Cref{fig:granular_ergm_low_truncated_2} show that the synthetic networks sampled from the ERGM trained \emph{without differential privacy} (\textcolor{red}{red dots}) accurately estimate the prevalence ratio of the \emph{observed network} (\textcolor{blue}{blue dots}) for each age and race group with relatively low bias for the high and low prevalence case respectively. The absolute difference is at most 0.05 and 0.025 units respectively for the high prevalence (\Cref{fig:granular_ergm_high}) and low prevalence (\Cref{fig:granular_ergm_low}) cases. On the other hand in \Cref{fig:granular_bm_high} and \Cref{fig:granular_bm_low}, we clearly see a large noticeable bias for all age and race groups. In the high prevalence case (\Cref{fig:granular_bm_high}), the red and blue dots differ by at least 0.05 units and in low prevalence case (\Cref{fig:granular_bm_low}) the red and blue dots differ by at least 0.03 units. This shows that \textbf{when model misspecification exists, it is the largest source of error at the subgroup level}. This is similar to what we observe at the population level in \Cref{sec:find1}.

\paragraph{\textbf{Correctly specified model can also suffer leads to avergaing effects.}}
However, one interesting difference from the population level trend is that even when one uses a correctly specified model (such as the ERGM in our experiments which matches the model used to create the observed network), some small error persists for the different age and race groups as observed in \Cref{fig:granular_ergm_high_truncated_2} and \Cref{fig:granular_ergm_low_truncated_2}. In \Cref{fig:epsilon_max_degree_3}, we plot the prevalence ratios of the various subgroups under different experimental settings. For this discussion, we focus on subplots (a) and (b) where in \Cref{fig:epsilon_max_degree_3} where the model used to train the networsk in ERGM. More specifically we focus on the ``Observed'' and ``No DP'' columns. In the observed network, there is a larger variation in prevalence ratio across the different age and race groups compared to the no DP case, where the prevalence ratios are much more tightly packed close to each other. Some of the subgroups are underestimated whereas others are overestimated when modelled without differential privacy. This suggests that \textbf{ modeling introduces an averaging effect that brings the prevalence ratio of all the different subgroups closer to an average value.} This explains why we see a small amount of bias for the various subgroups in the no model misspecification case. However, note that this is true only seen when model misspecification does not exist. The presence of bias due to model misspecfication can overshadow this effect as can be seeb in subplots (c) and (d) in \Cref{fig:epsilon_max_degree_3}.

\subsubsection{Bias due to Differential Privacy}
\label{sec:granular_dp_bias}

To study the additional bias due to differential privacy at the subgroup level, we compare the prevalence ratio captured from synthetic networks trained using a model (ERGM or SBM) with differential privacy using different values of the privacy parameter $\eps$ against the prevalence ratio captured when the same model is trained without differential privacy for each age and race group. 

In Figures \ref{fig:granular_ergm_high}, \ref{fig:granular_ergm_low}, \ref{fig:granular_bm_high} and \ref{fig:granular_bm_low} the sampled networks from models trained with differential privacy are marked with differently colored dots for each value of $\eps$ and the networks sampled from models trained without differential privacy marked by red dots.

Similar to \Cref{sec:find2} we see that across two different modeling choices ERGM or SBM, different prevalence levels, high or low prevalence, different algorithmic choices of truncated degree used in the differentially private algorithm and across age and race groups, there is very little additional bias introduced due to differential privacy, except in a few cases where $\eps$ is less than equal to 1.

To better understand the effect of the privacy budget $\eps$ on the additional bias due to DP, we focus on a specific case of truncated degree 3 in \Cref{fig:epsilon_max_degree_3}. We see that there is no meaningful change in the additional bias introduced by DP as $\eps$ decreases for each age and race group except in the case of ERGM with low prevalence. 

\begin{figure}[h!]
    \centering
    \begin{subfigure}[b]{0.48\textwidth}
        \centering
        \includegraphics[width=\textwidth]{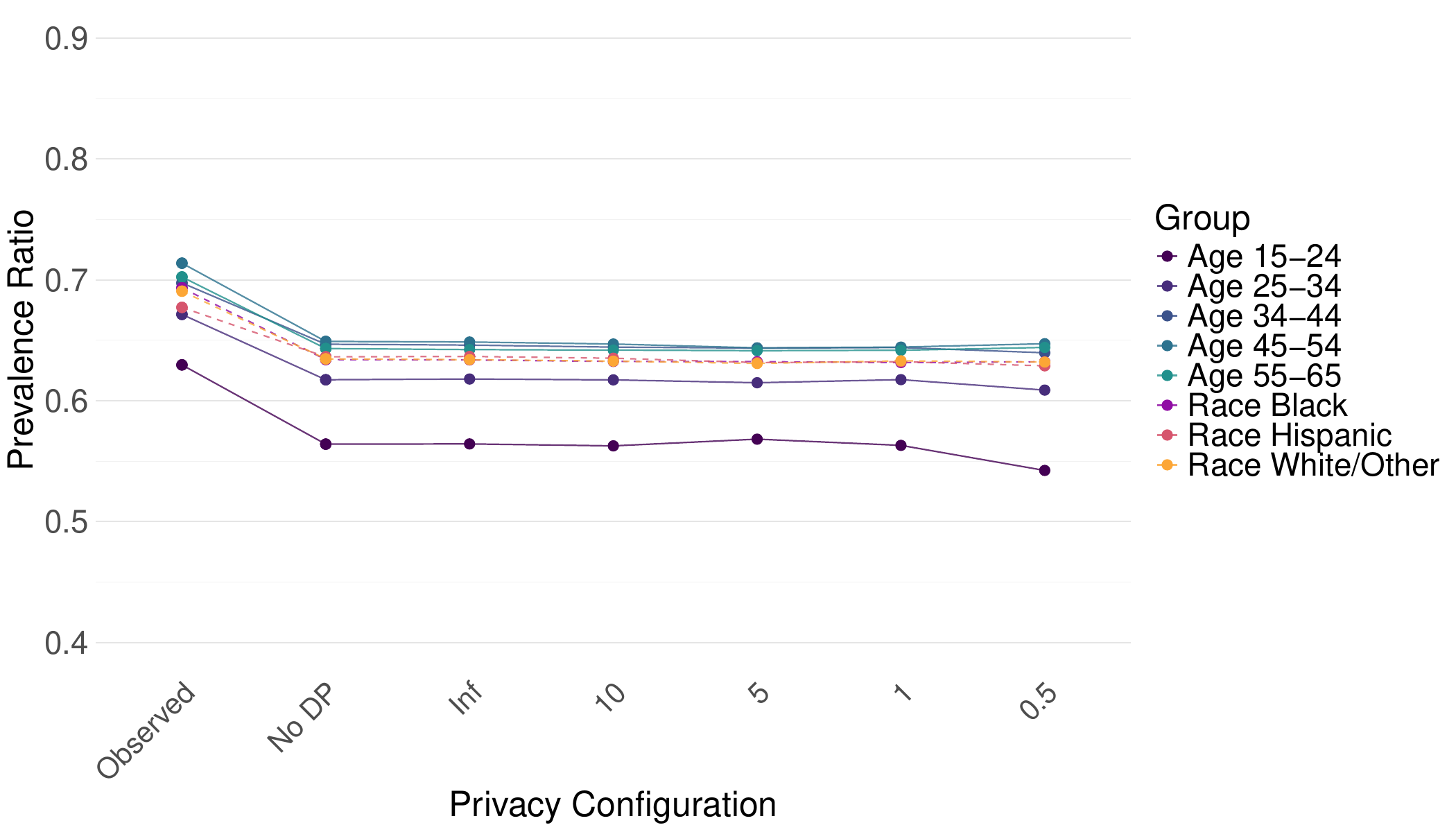}
        \caption{SBM - High Prevalence}
    \end{subfigure}
    \hfill
    \begin{subfigure}[b]{0.48\textwidth}
        \centering
        \includegraphics[width=\textwidth]{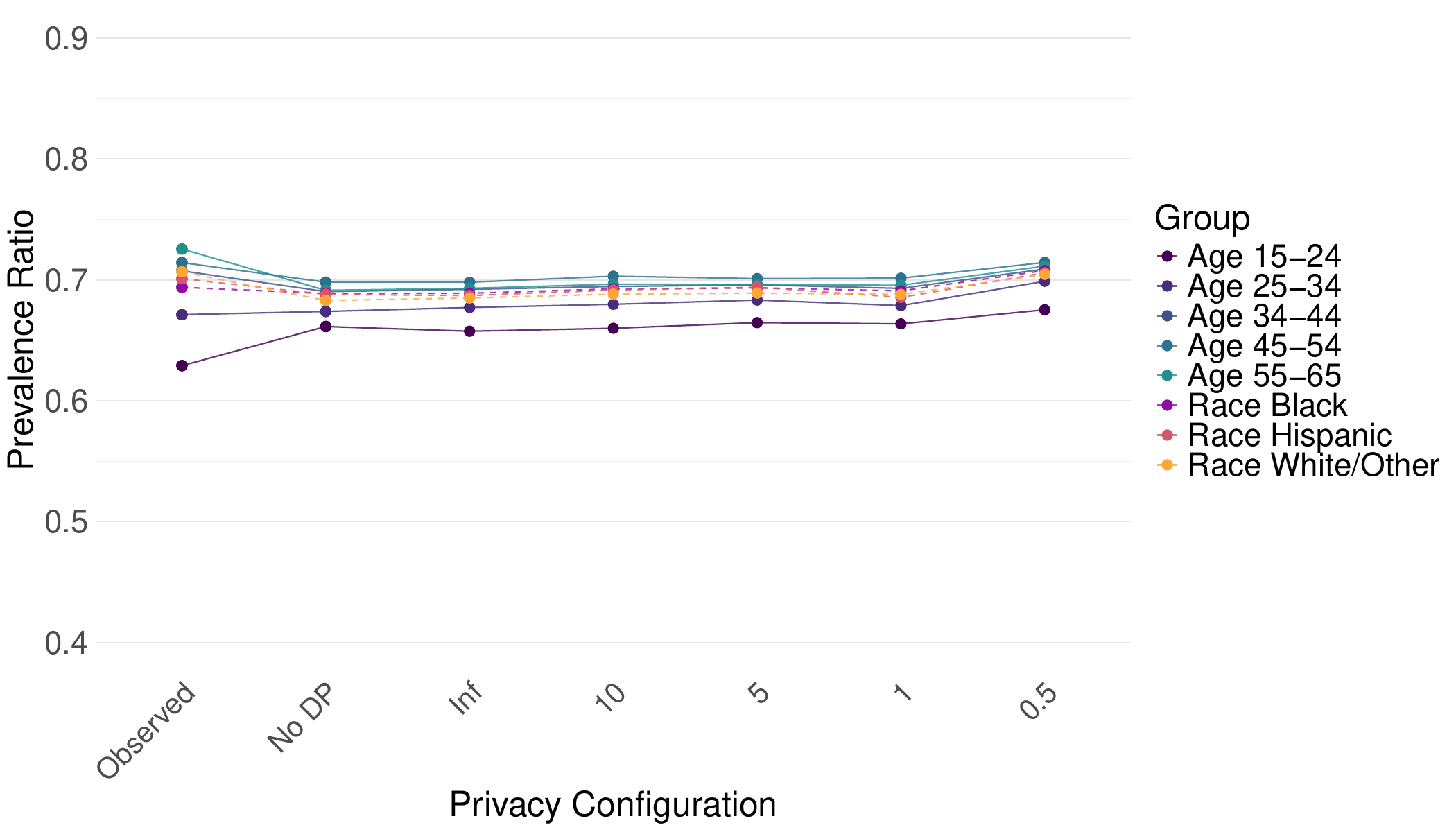}
        \caption{ERGM - High Prevalence}
    \end{subfigure}
    
    \vspace{0.5cm}
    
    \begin{subfigure}[b]{0.48\textwidth}
        \centering
        \includegraphics[width=\textwidth]{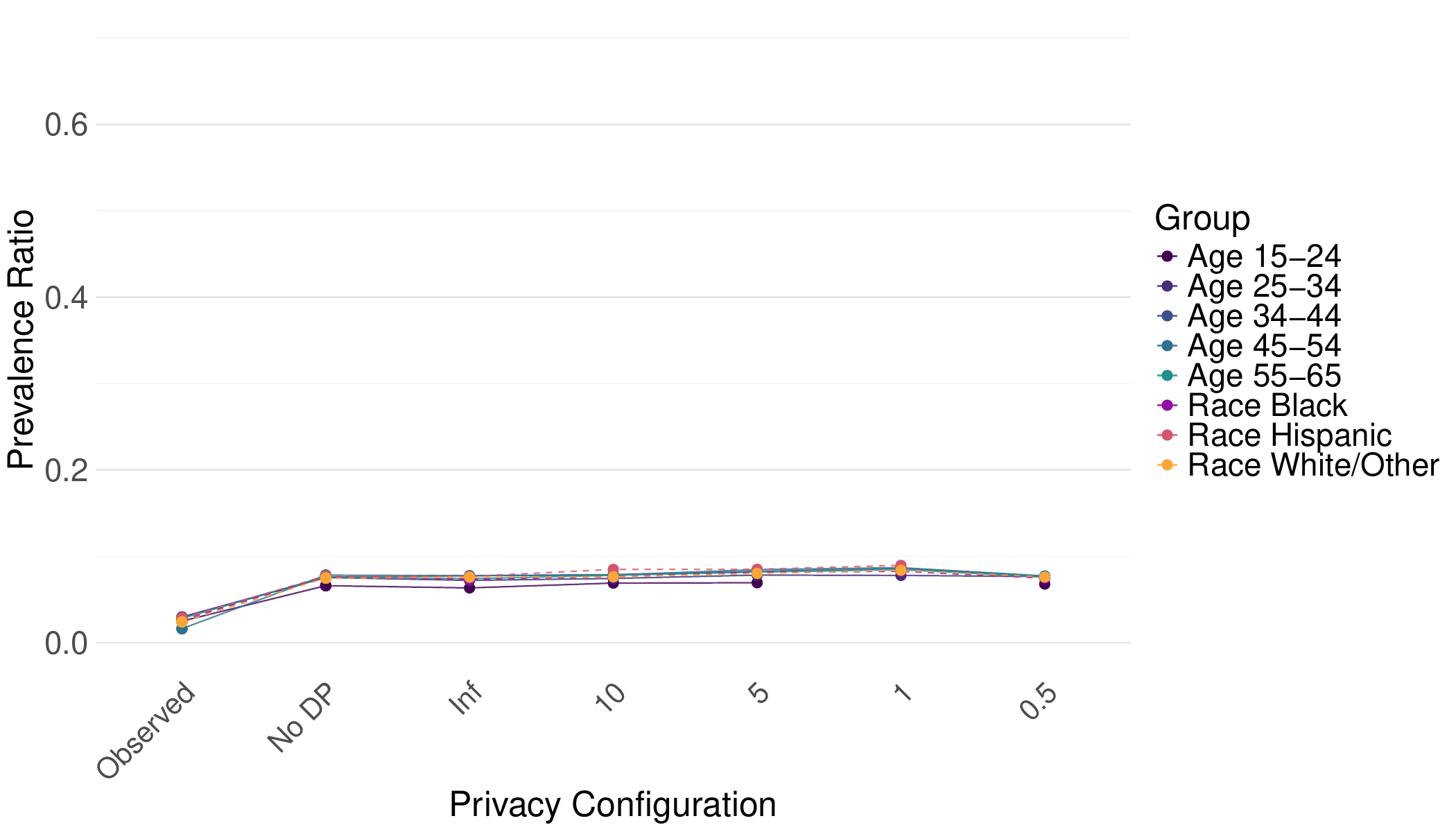}
        \caption{SBM - Low Prevalence}
    \end{subfigure}
    \hfill
    \begin{subfigure}[b]{0.48\textwidth}
        \centering
        \includegraphics[width=\textwidth]{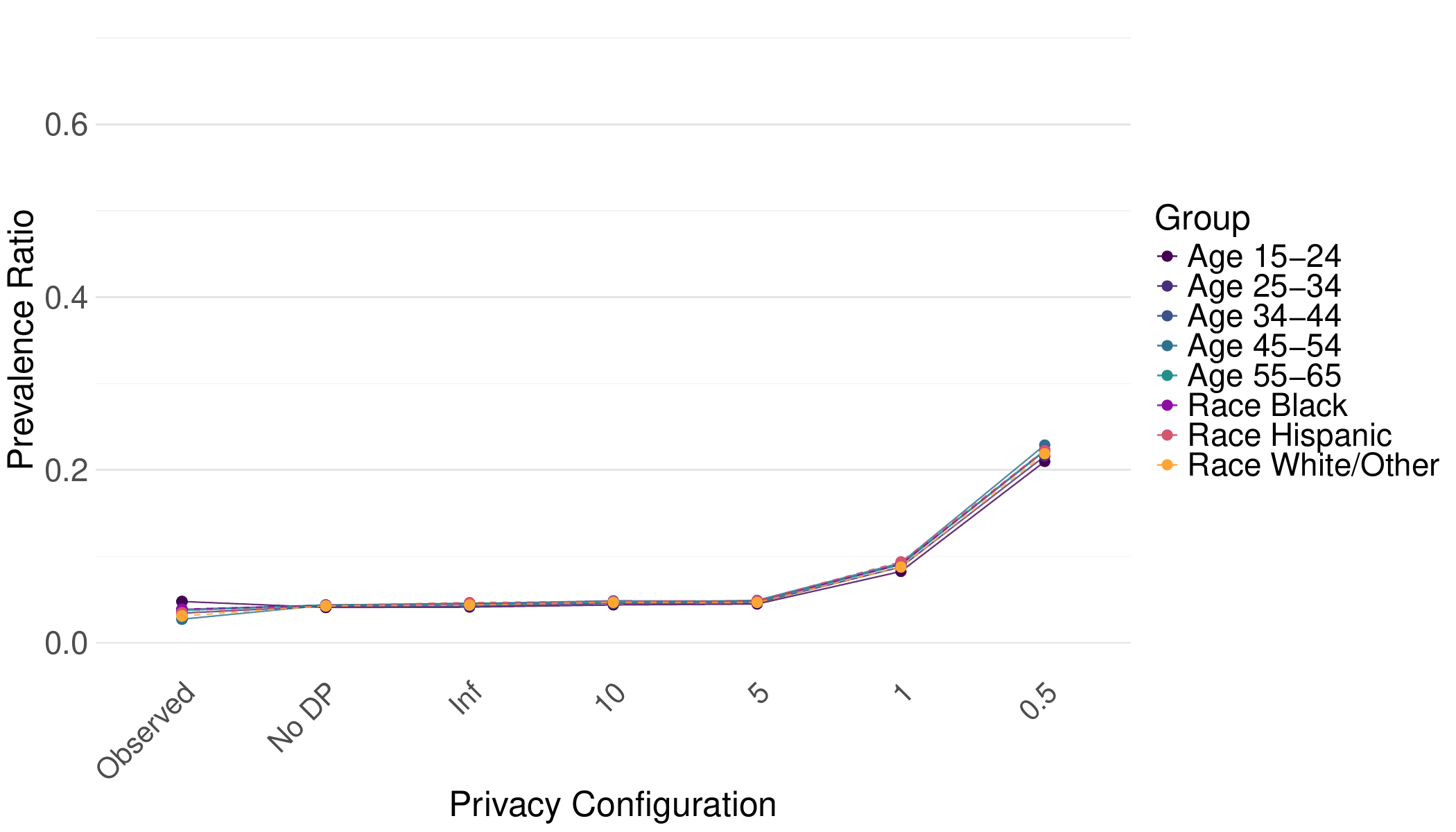}
        \caption{ERGM - Low Prevalence}
    \end{subfigure}
    \caption{Effect of privacy budget ($\eps$) on prevalence ratio with maximum degree truncated at 3 across all four modeling scenarios.}
    \label{fig:epsilon_max_degree_3}
\end{figure}

\subsection{Variance Sources in the Pipeline}
\label{sec:find4}

\begin{findingbox}{4}
Less variance is added by differential privacy than by network sampling and epidemic simulation.
\end{findingbox}

Variance of network and epidemic staticsts due to differential privacy releases remains low at privacy loss $\eps = 5$ or 10, but tends to increase when truncated degree is higher, such as $\Delta = 5$ or with lower privacy loss $\eps = 0.5$ and 1. This is evidenced by the spread of per-release averages for prevalence \Cref{fig:prevalence_baseline,fig:prevalence_intervention} and prevalence ratio \Cref{fig:prevalence_ratio}

As discussed in \Cref{sec:sources-of-error}, there are additional sources of variance introduced in various steps of the pipeline.
To study the relative contribution of each source of randomness, we examine what proportion of the total variation can be attributed to: (1) the differentially private release of summary statistics; (2) sampling a synthetic network from a model fitted on summary statistics; and (3) simulating an epidemic on a synthetic network.

These sources of randomness can be illustrated by the spread of each individual violin plot. To derive more quantifiable findings, we apply an ANOVA sum-of-squares decomposition technique to the results (see \Cref{app:variance-analysis} for a complete description).

We focus on the case $\eps = 1$ and $\Delta = 3$ for the mean prevalence metric in the baseline scenario. 
In this $\eps, \Delta$ configuration, we find that, in all four cases (model family $\times$ high/low prevalence condition) except the ERGM with low prevalence, the randomness from sampling a network is the most prominent source of randomness, accounting for close to 50\% of the total variance. \Cref{fig:variance_analysis_ergm_high} visualizes the sources of randomness for ERGM at the high prevalence condition, and \Cref{tab:anova_ergm_high} provides the variance decomposition (see \Cref{app:variance-analysis-results} for the results of the other combinations of model family and prevalence condition).

\begin{figure}[tb]
    \centering
    \includegraphics[width=0.75\textwidth]{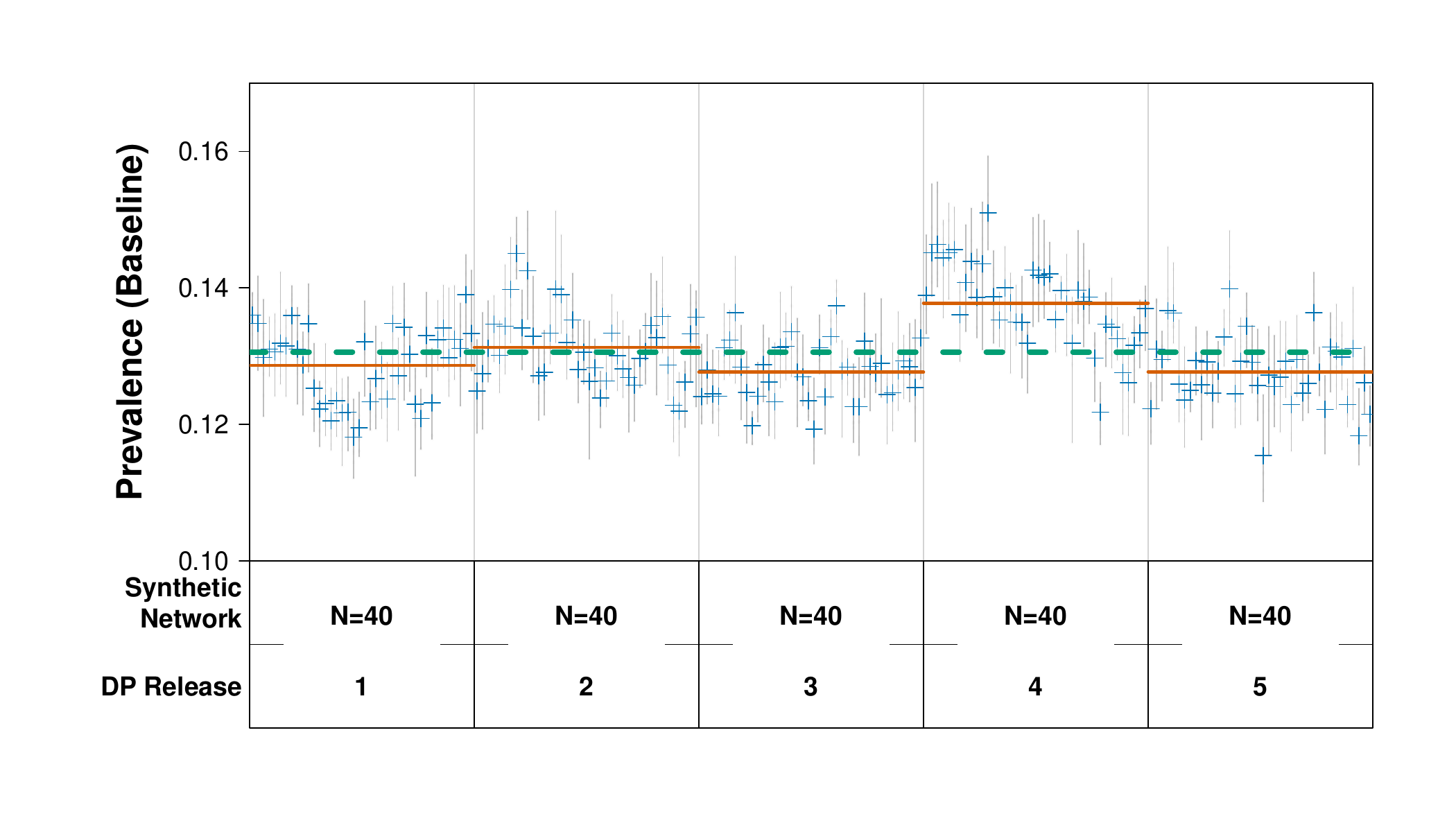}

    \caption{
    Variance analysis at privacy budget $\eps = 1$ and truncated degree $\Delta = 3$ showing the breakdown of different sources of randomness in the experimental pipeline for ERGM at High Prevalence condition. The plots show the average prevalence of baseline scenario over all synthetic networks and simulations (\textcolor{purple}{solid line}), per differential private release (\textcolor{orange}{dashed line}), per synthetic network (\textcolor{blue}{plus sign}), and the range per simulation (\textcolor{gray}{horizontal line}).}
    \label{fig:variance_analysis_ergm_high}
\end{figure}

\begin{table}[tb]
\centering
\begin{tabular}{lcccc}
\toprule
Source & df & SS & MS & Var (\%) \\
\midrule
Release & 4 & 0.0291 & 0.0073 & 26.36 \\
Network : Release & 195 & 0.0532 & 0.0003 & 48.13 \\
Simulation : Network : Release & 1800 & 0.0282 & 0.0000 & 25.51 \\
\bottomrule
\end{tabular}
\caption{Analysis of variance for ERGM high prevalence condition showing the breakdown of variance sources in the experimental pipeline. Columns show the source of variance, degrees of freedom (df), sum of squares (SS), mean squares (MS), and percentage of total variance explained (Var \%).}
\label{tab:anova_ergm_high}
\end{table}

We can conclude that even though the analyst cannot assess the variance introduce by differential privacy (\Cref{sec:analyst-prespective}), the additional variance is smaller and is comparable to the variance already present in the remaining pipeline. 


\section{Discussion}
\label{sec:discussion}

Human contact network data offers substantial value for infectious disease research but raises serious privacy concerns. We propose a pipeline that integrates node-level differential privacy with statistical network models and epidemic simulation. Our evaluation suggests that this approach can provide formal privacy protection while maintaining accuracy sufficient for epidemic simulation studies, particularly when statistical models are appropriately specified and privacy parameters are reasonably chosen. This also applies to the analysis of epidemic outcomes in a more granular level (e.g., per demographic group). The additional error and variance introduced by differential privacy remain modest relative to other sources of uncertainty in standard epidemiological practice. These results suggest that privacy-preserving analysis of sensitive network data through differential privacy is technically feasible, offering a path toward responsible sharing of network data for epidemiological research.

\subsection{Limitations and Directions for Future Work}
\label{sec:limitations-future-work}

Several limitations constrain the generalizability of our findings and suggest directions for future work.

\paragraph{Network dynamics.}
Our evaluation uses cross-sectional networks, but many epidemiological studies incorporate partnership formation, dissolution, and temporal dynamics during simulation. Extending our pipeline to support such dynamic epidemic models on static network snapshots represents a natural next step. A further challenge involves extending differential privacy to longitudinal network data with multiple snapshots over time, which requires protecting an individual's entire trajectory of connections rather than their presence in a single cross-section.

\paragraph{Model complexity.}
Our network model specification (BMs and ERGMs) and disease model (SIS) are 
widely used. Nevertheless, the literature includes many more complex models.
Whether our results extend to more sophisticated network and disease models---incorporating, for example, exposed periods, recovered state, and behavioral interventions---requires further investigation.

\paragraph{Privacy Budget and Algorithm Selection} The differential privacy guarantee is specified by a parameter $\eps$, often called the privacy budget, which bounds how much can be leaked about any individual. The range $\eps\leq 1$ is generally considered most meaningful; in our experiments, this is also the range where distortion due to privacy starts to play a noticeable role.

These results highlight the importance of exploring a wider range of algorithm design techniques.
Our differentially private algorithms follow the design of \citet{DayLL16}: we apply a truncated degree transformation, compute network statistics on the transformed graph, and inject calibrated noise to all statistics. 
The degree truncation parameter $\Delta$ determines a trade-off between bias, which increases as $\Delta$ gets smaller, and noise variance. In practice, one would need to select $\Delta$ carefully, either by spending a small part of the privacy budget or by leveraging prior knowledge about related data sets. 

That said, many other types of differentially private algorithms exist. Exploring other techniques---for example, adapting the so-called \textit{matrix mechanism} \cite{LiMHMR2015matrix} to the setting of releasing multiple network statistics---could yield better privacy-utility trade-offs. We consider this an exciting direction for future work.


\section{Related Work}
\label{sec:related-work}

Differential privacy in epidemiology remains largely unexplored despite its potential significance, although recently it starting to gain some interest. The few existing works include \citet{Chen2024acc} weight-DP mechanism for releasing basic reproduction numbers using ODE models on community-level graphs, and their extension to distributed reproduction numbers for community subsets \cite{Chen2025arxiv}. \citet{Nguyen2024Computing} provides a differentially private algorithm to compute the size of the outbreak with the SIR model. Another work focuses on optimizing public health interventions such as vaccination strategies \citet{Nguyen2025Controlling}. Other applications encompass privacy-preserving mobility matrices for epidemiological simulations \cite{Savi2023plos}, a \texttt{data.org} competition on designing DP mechanisms for epidemiological tasks using merchant-level financial transaction data \cite{dataorg2023pets}, and Google's differentially private COVID-19 Community Mobility Reports \cite{aktay2022google}.

Our work applies differential privacy to another standard epidemiological practice of directly simulating disease spread on individual-level human contact networks rather than using ODE models or community-level representations. The design of our pipeline is motivated by the an existing effort to collect detailed individual contact data over a period of time \cite{mapps2023,mappsWorkshop2023}. For a more comprehensive review, refer to Appendix~\ref{app:related-work}.

\jnote{Any other acknowledgements?}

\addcontentsline{toc}{section}{Acknowledgments}
\section*{Acknowledgments}

We thank Dr. Samuel M. Jenness for the use of ARTNet data, the egocentric network dataset upon which our generative network models rely.

\jnote{If you don't like "initiated", feel free to tweak the language to your satisfaction. I know you (Shlomi and Debanuj) started this project in a class.}

\subsection*{Funding}

This project was initiated in part at the ``Workshop on Privacy and Ethics in Pandemic Data Collection and Processing'' (January 17--20, 2023) organized by the MAPPS (Mobility Analysis for Pandemic Prevention Strategies) Project. The MAPPS project was funded under a Predictive Intelligence for Pandemic Prevention Phase I award from the National Science Foundation (Grant No. 2154941). The workshop was held in Providence, Rhode Island, and hosted by ICERM (Institute for Computational and Experimental Research in Mathematics). ICERM is supported by the National Science Foundation under Grant No. DMS-1929284.

S.H. was supported in part by DARPA under Agreement No. HR00112020021. A.S. was partly supported by US National Science Foundation awards CNS-2232694 and CNS-2120667. D.N. was partially supported by US National Science Foundation award 2022446.  Any opinions, findings, and conclusions or recommendations expressed in this material are those of the authors and do not necessarily reflect the views of the United States Government.

\subsection*{Prior Presentation}

Preliminary results from this work were presented at the Society for Epidemiologic Research's Annual Meeting (June 10--13, 2025) in Boston, Massachusetts.


\bibliographystyle{abbrvnat}
\bibliography{references}

\clearpage

\appendix

\section{Additional Related Work}
\label{app:related-work}

\subsection{Privacy of Network Data}
\shnote{I'd love to get feedback on this section, there is much more to say based on the survey, but I'm not sure whether it is needed.}

A substantial body of literature demonstrates that network data concerning individuals poses significant privacy risks. Network anonymization techniques, despite their intended protective function, have been repeatedly shown to remain vulnerable to de-anonymization attacks. This section highlights several foundational works; readers seeking comprehensive coverage are directed to the survey by \cite{Ji2017Graph}, which concludes that existing anonymization schemes fail to adequately protect individuals, even with edge-level differential privacy guarantees. This highlights the need for node-level differential privacy protections studied in this work.

The seminal work of \citet{Backstrom2027Wherefore} designed two distinct attacks against ``anonymized'' networks (those stripped of personal identifiers such as names and addresses). These attacks aim to determine whether an edge exists between two targeted nodes. Under the active threat model, attackers can introduce new nodes and edges into the network prior to its anonymized release. By strategically designing these additions, they can subsequently recover a subgraph containing both their planted nodes and the targeted nodes, thereby learning about the targets' connectivity. In contrast, passive attacks enable ordinary network participants to ascertain their positions through analysis of local structural patterns. When a small coalition of such attackers collaborates to pinpoint their respective locations, this process inadvertently compromises the privacy of adjacent network participants.

\citet{Narayanan2009De} traded assumptions in the threat model: attackers possess only the anonymized network itself (without being part of it) along with some external auxiliary knowledge about individuals. Despite these constraints, they successfully de-anonymized several thousand Twitter users from an anonymous graph by leveraging a completely different social network (Flickr) as their source of auxiliary information.

\subsection{Differential Privacy and ERGMs}

Prior works have considered the estimation of exponential random graphs subject to differential privacy. \citet{KarwaS12} show an algorithm for obtaining a maximum likelihood estimator of the degree sequence of the graph and use its output to obtain  synthetic graphs via the $\beta$-model, a simplified version of ERGMs that uses only the degree sequence as input.  Lu and Miklau obtain edge-private algorithms for computing three graph statistics, the number of alternating $k$-stars, alternating $k$-twopaths, and alternating $k$-triangles, and use these statistics to fit ERGMs~\cite{LuM14}. They also propose a new Bayesian method for the ERGM fitting that takes into account the noise distribution of the private statistics and offers better accuracy. In two consecutive works, \citet{KarwaSK14} study the problem of sharing synthetic social network data and propose a randomized response algorithm for generating a synthetic graph and an improved MCMC approach for ERGM inference from the differentially private graph \cite{KarwaKS15}. \citet{LiuEJB20} propose an algorithm for node-privately generating synthetic graphs via ERGMs that does not depend on a choice of statistics. Their approach depends on obtaining a private estimate of the posterior distribution of the parameters of the ERGM given the input graph. Our approach is most similar to that of \citet{LuM14} where we use privately estimated statistics of the observed network to generate synthetic graphs via ERGMs. However, our approach satisfies the stronger guarantee of node-privacy. Differently from the aforementioned works, we do not modify the fitting and generation procedures for ERGMs to account for the privacy noise, with the goal of maintaining simplicity and the ability to use common tools in the epidemiology community for disease simulation. Finally, we evaluate our method for private ERGMs on the downstream tasks of disease spread simulation, whereas previous works evaluate their methods in terms of similarity of the generated graph to the observed graph. 

\subsection{ERGMs in Epidemiology}

\shnote{Jason, could you review this please?}

Network models and specifically ERGMs are commonly used to model the spread of infection in a population. For instance, ERGMs have been used to model the dynamics of sexually transmitted diseases such as HIV, gonorrhea and chlamydia with real world observed data and typical transmission parameters. This has allowed epidemiologists to study and infer insights about HIV epidemics in different contexts \cite{jenness2016effectiveness, jenness2016impact, goodreau2012concurrent, goodreau2012drives, jenness2017incidence, goodreau2017sources}. ERGMs have also been used to investigate the the effect of missing cases on a transmission network for drug-resistant tuberculosis (TB) \cite{nelson2020modeling}.  More recently, ERGMs have been used to study the efficacy of disease outbreak prevention measures for COVID-19 \cite{jenness2021dynamic}.

\subsection{Differential Privacy and Epidemiology on Network Data}

\citet{Chen2024acc} provides a DP mechanism to release the (network-level) basic reproduction number $R_0$--—the average number of new infections caused by a single infected individual in a completely susceptible population---given the weighted adjacency graph between communities. The differential privacy variant is neither node-DP nor edge-DP, but a weight-DP, where two graphs are weight adjacent (neighbors) if they have the same edge and node sets, and the distance between their weight matrices is bounded by a prespecified value in the Frobenius norm. The disease dynamics is modeled by ordinary differential equations, and not via direct simulation on a network as in our work. A follow-up work design a shuffle-model DP mechanism to release a more granular version of $R_0$, namely distributed reproduction numbers, quantifying the spread in diseases for subsets of nodes (community) \cite{Chen2025arxiv}.

\citet{Nguyen2024Computing} provide an edge-DP algorithm to compute the size of the outbreak with an SIR model, adding a \emph{recovered} state which ``excludes'' a node from the simulation (we used an SIS model). Their results include utility guarantees for the family of expander graphs, and a lower bound on the additive error that DP algorithms induce for estimating the expected number of infections.

Rather than computing an epidemic metric, another type of goal seeks to find an optimal public health intervention. \citet{Nguyen2025Controlling} transposes the problem of finding a vaccine strategy---selecting which individual to vaccinate, meaning prevent them from being either susceptible or infected---into a minimization problem of network statistics. Specifically, they provide two edge-DP algorithms for choosing nodes to be removed in order to minimize the maximum degree and spectral distance.

Differential privacy was applied also to matrices of mobility data, obtained from a data broker, documenting movement of individuals from one county to another, with privacy unit of individual per time window \cite{Savi2023plos}. These matrices are used to feed ODE epidemiological simulations with various starting points and dynamics. Recently, a \texttt{data.org} competition was conducted on designing new DP mechanisms to solve various epidemiological tasks based on merchant-level weekly financial transaction data at postal code granularity \cite{dataorg2023pets}.

During the COVID-19 pandemic, Google published Community Mobility Reports that quantified changes in mobility movement trends over time by geography, across different categories of places, releasing metrics with differential privacy using user-day as the privacy unit \cite{aktay2022google}.

\subsection{Differential Privacy for Network Data}

An extensive literature has studied various aspects of network analysis and modelling via edge-privacy, see, for example, \cite{NissimRS07,HayLMJ09,BlockiBDS12,GuptaRU12,KarwaS12,Upadhyay13, WangWW13, WangWZX13,KarwaRSY14,ProserpioGM14,LuM14, ZhangCPSX15, MulleCB15, NguyenIR16, RoohiRT19,ZhangN19,AhmedLJ20,BlockiGM22,DLRSSY22}. 

Node privacy is a better fit for social network data, where one individual's data consists of all the relationships (edges) that the individual contributes to the graph. Its much stronger privacy protection makes it harder to attain.

Existing work addresses subgraph counts \cite{BlockiBDS13, KasiviswanathanNRS13, ChenZ13, 
DingZB018,  LiuML20}; degree and triangle distributions \cite{RaskhodnikovaS16,DayLL16,LiuML20}; number of connected components \cite{KalemajRST23}, parameter estimation in stochastic block models \cite{BorgsCS15, BorgsCSZ18, SealfonU21}; and training of graph neural networks \cite{DaigavaneMSTAJ21}. See \cite{RaskhodnikovaS16-E, MuellerUPRK22} for surveys on node-private algorithms and \cite{XiaCKHT021} for implementations of some of the algorithms. 

Multiple works focus on generating synthetic networks with edge-DP using statistics based on the degree distribution \cite{Mir2012Differentially,Proserpio2012Workflow,Zhu2019DPFT}, particularly with the dK-graph model \cite{Sala2011Sharing,Wang2013Preserving,Gao2019Sharing,Gao2020Protecting,Iftikhar2020dK}. Another edge-DP approach clusters the network into communities, which tend to have stronger intra-connectivity, and counts edges within and between these clusters of nodes \cite{Chen2014Correlated,Gao2019PHDP,Yuan2023Privgraph}. The cut function has also been used as a statistic to be released with edge-DP for generating synthetic networks \cite{Gupta2012Iterative,Eliavs2020Differentially}.
Other edge-DP works have leveraged neural networks to learn GANs with DP-SGD \cite{Yang2020Secure,Zheng2021Network} or to learn node embeddings that capture connectivity \cite{Zhao2023Differentially} using a method inspired by Word2Vec. Notably, there are a few works on node-DP. \citet{Xiao2014Differentially} utilize a hierarchical random graph model. \citet{Zhang2020Community} released the Katz index to generate synthetic networks, and \citet{Zhang2025PrivDPR} used an approximate PageRank algorithm to capture network structure. A recent work by \citet{Yuan2024PSGGrap} provides an edge-DP algorithm for continual release, i.e., capturing temporal changes in a network, using the clustering approach discussed above.

\section{Differentially Private Analysis of Network Data}
\label{sec:prelim}

\subsection{Node and Edge-Differential Privacy}

We start by formally defining the notion of neighbors. 

\begin{definition}[Node neighbors]  
Two graphs $G$ and $G'$ are node neighbors if $G'$ can be obtained from $G$ by removing a single node along with all its incident edges or adding a node with arbitrary edges incident on it.
\end{definition}

\begin{definition}[Edge neighbors]  
Two graphs $G$ and $G'$ are edge neighbors if $G'$ can be obtained from $G$ by adding or removing a single edge while keeping the node set unchanged.  
\end{definition}

The notion of privacy defined below depends on type of neighbors we use.

\begin{definition}[Node-differential privacy~\cite{HayLMJ09}]
\label{def:node-dp}
A randomized algorithm $\cA$ is $\eps$-node-differentially-private if for all node-neighboring graphs $G, G'$ and all events $S$ in the output space of $\cA$,  
\begin{align*}
    \Pr[\cA(G) \in S] \leq e^\eps \Pr[\cA(G') \in S].
\end{align*}  
\end{definition}

\begin{definition}[Edge-differential privacy~\cite{HayLMJ09}]  
\label{def:edge-dp}
A randomized algorithm $\cA$ is $\eps$-edge-differentially-private if for all edge-neighboring graphs $G, G'$ and all events $S$ in the output space of $\cA$,  
\begin{align*}
    \Pr[\cA(G) \in S] \leq e^\eps \Pr[\cA(G') \in S].
\end{align*}  
\end{definition}

The most basic private mechanism for releasing a statistic $f$ is called the \emph{Laplace Mechanism}. It returns the value of $f$ with additive noise scaled according to the global sensitivity of $f$. The noise follows a Laplace distribution. 

\begin{definition}[Global sensitivity \cite{DworkMNS16}]\label{def:GS} 
Given a function $f \colon \cG \to \R$, its global sensitivity, $GS(f)$, is defined as 
\begin{align*}
    GS(f) = \max_{\text{neighbors $G,G'$}} |f(G) - f(G')|.
\end{align*}
\end{definition}
\noindent Unless specified otherwise, we use $GS(f)$ w.r.t.\ node-neighbors.

\begin{definition}[Laplace distribution]
\label{def:laplace-distribution}
The Laplace distribution with mean $0$ and standard deviation $\sqrt{2}b$, denoted by $\mathrm{Lap}(b)$, is defined for all $z \in \mathbb{R}$ and has probability density
$$h(z) = \frac{1}{2b}\exp\big(-\frac{|z|}{b}\big).$$
\end{definition}

\begin{theorem}[Laplace Mechanism \cite{DworkMNS16}] 
\label{thm:laplace}
Given a function $f \colon \cG \to \R$. The algorithm $\cA$ that, given a graph $G$, outputs $\cA(G) = f(G) + Z $ where $Z \sim \mathrm{Lap}(GS(f)/\eps)$ is $\eps$-node-differentially-private.
\end{theorem}


Differential privacy is preserved under post-processing. Additionally, the outputs of multiple private algorithms can be combined to obtain an algorithm that has privacy protection linear in
the number of composed algorithms. 

\begin{lemma}[Composition and post-processing \cite{DworkMNS16,DworkKMMN06}]
\label{lem:composition}
If algorithm $\cA$ runs 
$\eps_1, \eps_2, \ldots, \eps_t$-node-private algorithms $\cA_1,...,\cA_t$
and applies a randomized algorithm $g$ to the outputs, then the algorithm $\cA(G) = g(\cA_1(G), \dots, \cA_t(G))$ is $(\sum_{i=1}^{t} \eps_i)$-node-private.
\end{lemma}

\subsection{Network Statistics Used in SBMs and ERGMs}
\label{app:net-stats-defs}

In this section, we assume that each vertex is labeled with one of $k$ classes. This is done to model certain attributes of the population such as age or race. Next we define graph statistics that will be used as sufficient statistics for the statistical models we use in our experiments.

Stochastic block models (SBMs) have one collections of sufficient statistics: the mixing matrix.

\begin{definition}[Mixing Matrix]
    Given a graph $G$ where each node is labeled with one of $k$ classes, $c_1,\dots , c_k$ the mixing matrix of the graph $G$ denoted as $f_{\mathsf{mm}}(G)$ is a $k \times k$ matrix where the $(i,j)$-th entry is equal to the number of edges that have one end point labeled $c_i$ and other endpoint labeled $c_j$. The $(i,j)$-th entry of this matrix can also be referred as $f_{\mathsf{mm}}[i,j](G)$.
\end{definition}

Exponential family random graph models (ERGMs) are much more expressive, and can be fitted to a large family of sufficient statistics. We present those used in this work.

\begin{definition}[Nodematch]
\label{def:nodematch}
    There are two versions of this statistic.
    \begin{itemize}
        \item \textbf{Nodematch per class}: Given a graph $G$ where each node is labeled with one of $k$ classess, $c_1,\dots , c_k$ the nodematch per class of the graph $G$ denoted as  $f_{\mathsf{nodematch-per-class}}(G)$ is a vector of size $k$ where the $i$-th entry is the number of edges with end-points labeled $c_i$. This is also equal to the k length vector formed from the diagonal of the mixing matrix $f_{\mathsf{mm}}(G)$. The $i$th entry of this vector can independently be referred to as \textbf{Nodematch for class $c_i$}, denoted as $f_{\mathsf{nodematch-per-class}}[i](G)$.

        \item \textbf{Total Nodematch}: Given a graph $G$ where each node is labeled with one of $k$ classess, $c_1,\dots , c_k$ the nodematch of the graph $G$ denoted as $f_{\mathsf{total-nodematch}}(G)$ is equal to the sum of the entries of the vector $f_{\mathsf{nodematch-per-class}}(G)$.  This is also equal to the trace of the mixing matrix $f_{\mathsf{mm}}(G)$.
    \end{itemize}
\end{definition}

\begin{definition}[Nodefactor]
\label{def:nodefactor}
     Given a graph $G$ where each node is labeled with one of $k$ classess, $c_1,\dots , c_k$ the nodefactor of the graph $G$ denoted as $f_{\mathsf{nodefactor}}(G)$ is a vector of length $k$ where the $i$-th entry is the number of edges that have at least one endpoint labeled $c_i$. The $i$th entry of this vector can independently be referred to as \textbf{Nodefactor for class $c_i$}, denoted as $f_{\mathsf{nodefactor}}[i](G)$.
\end{definition}

In the Block Model specification, we use a mixing matrix based on the \emph{age} attribute. In the ERGM specification of our models, we use the following statistics for a graph $G$,
\begin{itemize}
    \item Number of edges,
    \item Number of nodes with degree at least two and degree at least four,
    \item Nodematch per class for \emph{age} and total nodematch for \emph{race}, and
    \item Nodefactor for \emph{age} and \emph{race}.
\end{itemize}

\subsection{Sensitivity Analysis for Relevant Network Statistics}
\label{app:sensitivity-analsysis}

\begin{theorem}[Global Sensitivity of Graph Statistics]
Let $\mathcal{G}_\Delta$ be the family of graphs with maximum degree $\Delta$. The global sensitivity (GS) of the following statistics are:
\begin{enumerate}
    \item Number of Edges: $GS(f_{\mathsf{edges}}) = \Delta$.
    \item Nodes with degree $\ge d$: $GS(f_{\ge d}) = \Delta + 1$.
    \item Mixing Matrix entry for classes $c_i, c_j$: $GS(f_{\mathsf{mm}}[i,j]) = \Delta$
    \item Nodematch for class $c_i$: $GS(f_{\mathsf{nodematch-per-class}}[i]) = \Delta$.
    \item Total Nodematch: $GS(f_{\mathsf{total-nodematch}}) = \Delta$.
    \item Nodefactor for class $c_i$: $GS(f_{\mathsf{nodefactor}}) = 2\Delta$.
\end{enumerate}
\end{theorem}

\begin{proof}
Let $G'$ be obtained from $G$ by adding a single node $u$ with incident edges $E_u$.

\paragraph{1. Number of Edges}
Adding node $u$ adds $|E_u|$ edges to the graph. Since the maximum degree is $\Delta$, the maximum number of added edges is $\Delta$.
$$ GS(f_{\mathsf{edges}}) = \Delta $$

\paragraph{2. Nodes with degree $\ge d$}
The statistic counts nodes $v$ where $\deg(v) \ge d$.
\begin{itemize}
    \item \textbf{Impact on $u$:} The new node $u$ contributes $+1$ to the count if $\deg(u) \ge d$.
    \item \textbf{Impact on neighbors:} The neighbors of $u$ (denoted by $N(u)$) each satisfy $\deg'(v) = \deg(v) + 1$. A neighbor contributes to the change only if its degree increases from $d-1$ to $d$. There are at most $\Delta$ such neighbors.
\end{itemize}
The total maximum change is the sum of the change for $u$ and its neighbors: $1 + \Delta$.
$$ GS(f_{\ge d}) = \Delta + 1 $$

\paragraph{3. Mixing Matrix entry for class $c_i, c_j$}
The Mixing Matrix is a matrix $M$ where $M[i,j]$ counts edges between classes $c_i$ and $c_j$. Let $u$ be the new node added with label $c(u)$. 
\begin{itemize}
    \item \textbf{Case 1:  ($c(u) \neq c_i \text{ and } c(u) \neq c_j$).} In this case, all incident edges to $u$ have one endpoint with a label that is neither $c_i$ nor $c_j$. Thus, no edge incident to $u$ can connect class $c_i$ to class $c_j$. There is no change to $f_{\mathsf{mm}}[i,j](G)$.
    
    \item \textbf{Case 2: $(c(u) = c_i \text{ or } c(u) = c_j)$.} Without loss of generality, let $c(u) = c_i$. Any new edge $(u, v)$ contributes to the statistic if and only if the neighbor $v$ has label $c_j$. In the worst case, all the $\Delta$ neighbors are in class $c_j$. Thus $f_{\text{mm}}[i,j](G)$ increases by $\Delta$.
\end{itemize}
$$GS(f_{\mathsf{mm}}[i,j]) = \Delta$$


\paragraph{4. Nodematch for class $c_i$ (Scalar).} 
$f_{\mathsf{nodematch-per-class}}[i](G)$ counts edges with both endpoints in class $c_i$. Let $u$ be the new node added whose label is $c(u)$.
\begin{itemize}
    \item \textbf{Case 1: ($c(u) \neq c_i$).} In this case, all incident edges to $u$ have one endpoint with label not equal to $c_i$. Thus there is no change to $f_{\mathsf{nodematch-per-class}}[i](G)$.
    \item \textbf{Case 2: ($c(u) = c_i$).} Any new edge $(u, v)$ contributes to the statistic if and only if $v$ also has label $c_i$. In the worst case, all the $\Delta$ neighbors are also in $c_i$. Thus $f_{\mathsf{nodematch-per-class}}[i](G)$ increases by $\Delta$.
\end{itemize}
$$ GS(f_{\mathsf{nodematch-per-class}}[i]) = \Delta $$

\paragraph{5. Total Nodematch (Scalar).}
This statistic sums the entries of the vector $f_{\mathsf{nodematch-per-class}}(G)$. It represents the total number of intra-class edges.
Adding node $u$ increases this count by the number of neighbors $v$ such that $\mathrm{label}(u) = \mathrm{label}(v)$. The maximum number of such neighbors is $\Delta$.
$$ GS(f_{\mathsf{total-nodematch}}) = \Delta $$

\paragraph{6. Nodefactor for class $c_i$ (Scalar).}
 $f_{\text{nodefactor}}[i](G)$ counts edges having at least one endpoint in class $c_i$. Let $u$ be the new node added whose label is $c(u)$.
 \begin{itemize}
     \item \textbf{Case 1: ($c(u) \neq c_i$).} In this case, the new node $u$ is not in class $c_i$. An incident edge $(u, v)$ contributes to the count if and only if the other endpoint $v$ is in class $c_i$ (so that the edge touches $c_i$). In the worst case, all $\Delta$ neighbors are in class $c_i$. Thus $f_{\text{nodefactor}}[i](G)$ increases by $\Delta$.
     
     \item \textbf{Case 2: ($c(u) = c_i$).} In this case, the new node $u$ is in class $c_i$. Therefore, every incident edge $(u, v)$ has at least one endpoint ($u$) in class $c_i$. Each of the $\Delta$ new edges contributes exactly +1 to the statistic, regardless of the label of $v$. Thus $f_{\text{nodefactor}}[i](G)$ increases by $\Delta$.
 \end{itemize}

 $$GS(f_{\text{nodefactor}}[i]) = \Delta$$

%
\end{proof}

\subsection{Differentially Private Algorithms to Release Network Statistics}
\label{app:dp-algs}




In this section, we present the blueprint of the differentially private algorithms used to release the network statistics for SBMs and ERGMS in \Cref{app:net-stats-defs} required to parameterize our generative models.

\subsubsection{Releasing a Scalar Statistic}
\label{app:releasing-scalar-stats}
To safely release a single, real-valued network statistic, our algorithm (described in \Cref{alg:dp_stat_release}) takes as input a graph $G$ the privacy parameter $\eps$, and the truncation degree bound $\Delta$.
\begin{itemize}
    \item \textbf{Filtering and Projection}. We first apply a filtering step to isolate the nodes and edges relevant to the specific statistic being queried. Because the global sensitivity of many network statistics scales with the number of nodes, we subsequently apply a projection step based on \cite{DayLL16} that caps the maximum degree of any node in the filtered graph to a predefined truncated degree parameter $\Delta$. This preprocessing limits the influence of any single node, thereby guaranteeing that the global sensitivity of the queried statistic is strictly bounded as a function of $\Delta$.

    \item \textbf{Compute True Statistic}. We compute the exact, unperturbed value of the statistic on this degree-bounded, projected graph.

    \item \textbf{Noise Addition for Privacy}. To achieve differential privacy, we perturb the true statistic using the Laplace mechanism (\autoref{thm:laplace}). We add random noise drawn from a Laplace distribution, where the scale of the noise is proportional to the statistic's global sensitivity divided by $\epsilon$.

    \item \textbf{Clipping}. Because structural graph statistics (such as edge counts or degree distributions) cannot realistically take negative values, we apply a deterministic clipping step. For any noisy output $x$, we release $\max(0, x)$ to ensure the final scalar value does not fall below zero. By the post processing property of differential privacy (\autoref{lem:composition}) the output remains private.  
\end{itemize}

\textbf{Note: Bias Introduced by Clipping (Boundary Effects)}
Clipping via $\max(0, x)$ ensures valid, non-negative counts but introduces a  positive bias for small true values. Let $v$ be the value of the true graph statistic and $z \sim \text{Lap}(b)$ the zero-mean noise needed for privacy. While the unclipped noisy value $v+z$ is an unbiased estimator i.e. ($\mathbb{E}[v+z] = v$), a small $v$ relative to the noise scale $b$ creates a high probability that $v+z < 0$. By clipping the negative left tail of this distribution to exactly zero while leaving the positive right tail unchanged, the expected value shifts upward; that is, $\mathbb{E}[\max(0, v+z)] > v$. This boundary effect can artificially inflate near-zero counts—such as sparse network features or rare interactions—propagating bias into downstream experiments.\asnote{Add a sentence indicating how often (or seldom!) this actually comes up, and where.}

\subsubsection{Releasing Multiple Scalar (and Vector) Statistics}

Fitting our statistical network models requires querying a collection of different statistics simultaneously. To ensure the entire procedure maintains $\epsilon$-node-differential privacy, we rely on the composition property of differential privacy (\autoref{lem:composition}). We divide the total available privacy budget $\epsilon$ among the queried statistics proportionally to their respective global sensitivities (as analyzed in \Cref{app:sensitivity-analsysis}). By allocating the budget in this manner, statistics with naturally higher sensitivities receive a larger share of $\epsilon$, which helps balance the scale of the Laplace noise added across all released metrics.

In some cases, the statistics we release are  vectors or matrices 
(for example, the mixing matrix, nodematch per class, and nodefactor). 
We handle these vector statistics by unpacking them into a series of distinct scalar components. Each constituent element (e.g., a specific cell within a mixing matrix) is treated as an independent scalar statistic and individually passed through the filtering, projection, and noise addition steps described above. Treating these statistics jointly might allow for reduced levels of noise—see our discussion of future work (\Cref{sec:limitations-future-work}).

\begin{algorithm}[H]
\caption{Differentially Private Scalar Graph Statistic Release}
\label{alg:dp_stat_release}
\begin{algorithmic}[1]
\Require Graph $G = (V, E)$, privacy budget $\epsilon > 0$, truncation degree bound $\Delta \in \mathbb{N}$,
\Statex filtering function $f_{\text{filter}}$, statistic function $f_{\text{stat}}$, sensitivity function $f_{\text{sens}}$
\Ensure Differentially private scalar statistic $\tilde{y}_{\text{final}}$
\Statex
\Statex \textbf{Step 1: Graph Filtering}
\State Initialize $E_{\text{filter}} \gets \emptyset$
\For{each edge $(u, v) \in E$}
    \If{$f_{\text{filter}}(u, v)$ is True}
        \State $E_{\text{filter}} \gets E_{\text{filter}} \cup \{(u, v)\}$
    \EndIf
\EndFor
\State $G_{\text{filter}} \gets (V, E_{\text{filter}})$

\Statex
\Statex \textbf{Step 2: Graph Projection (Enforcing Bounded Sensitivity)}
\State Initialize $E_{\text{proj}} \gets \emptyset$
\State Consider a stable ordering $\Lambda$ of the edges in $E_{\text{filter}}$ \Comment{For more details refer to \cite{DayLL16}}
\For{each edge $(u, v) \in E_{\text{filter}}$ chosen according to $\Lambda$}
    \If{$\text{degree}_{G_{\text{filter}}}(u) \leq \Delta$ \textbf{and} $\text{degree}_{G_{\text{filter}}}(v) \leq \Delta$}
        \State $E_{\text{proj}} \gets E_{\text{proj}} \cup \{(u, v)\}$
    \EndIf
\EndFor
\State $G_{\text{proj}} \gets (V, E_{\text{proj}})$

\Statex
\Statex \textbf{Step 3: Compute True Statistic}
\State $y \gets f_{\text{stat}}(G_{\text{proj}})$

\Statex
\Statex \textbf{Step 4: Sensitivity Calculation and Noise Addition}
\State $S \gets f_{\text{sens}}(\Delta)$
\State Draw noise $z \sim \text{Laplace}\left(0, \frac{S}{\epsilon}\right)$
\State $\tilde{y} \gets y + z$

\Statex
\Statex \textbf{Step 5: Post-Processing}
\State $\tilde{y}_{\text{final}} \gets \max(0, \tilde{y})$

\Statex
\State \Return $\tilde{y}_{\text{final}}$
\end{algorithmic}
\end{algorithm}

\clearpage

\section{Additional Results}
\label{app:results}

This section presents detailed results from our differentially private epidemiological modeling pipeline across four experimental conditions: ERGM with high prevalence, ERGM with low prevalence, baseline model (BM) with high prevalence, and BM with low prevalence.

\subsection{Prevalence}

\begin{figure}[h!]
    \centering
    \begin{subfigure}[b]{0.48\textwidth}
        \centering
        \includegraphics[width=\textwidth]{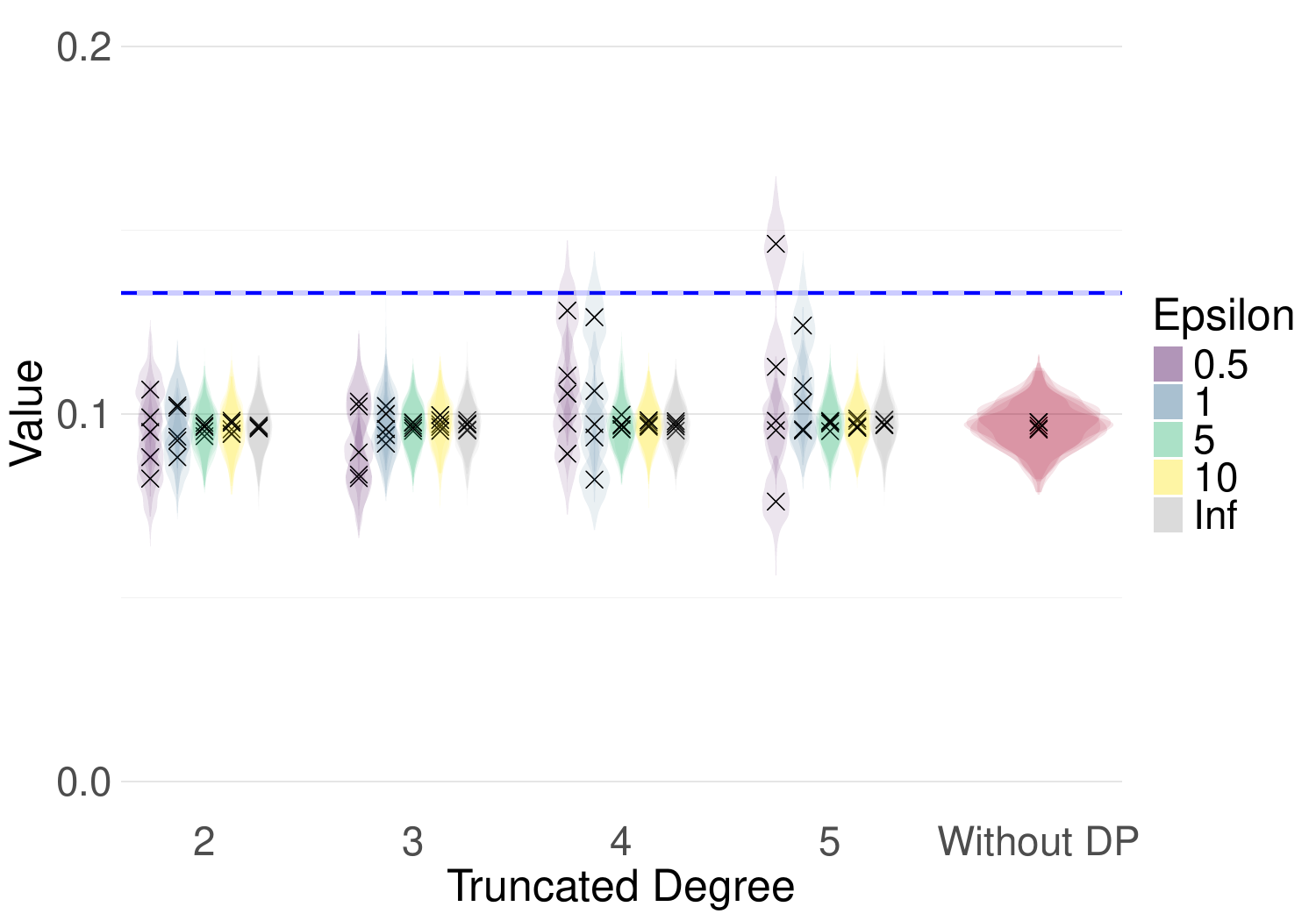}
        \caption{SBM - High Prevalence}
        \label{fig:prevalence_baseline_bm_high}
    \end{subfigure}
    \hfill
    \begin{subfigure}[b]{0.48\textwidth}
        \centering
        \includegraphics[width=\textwidth]{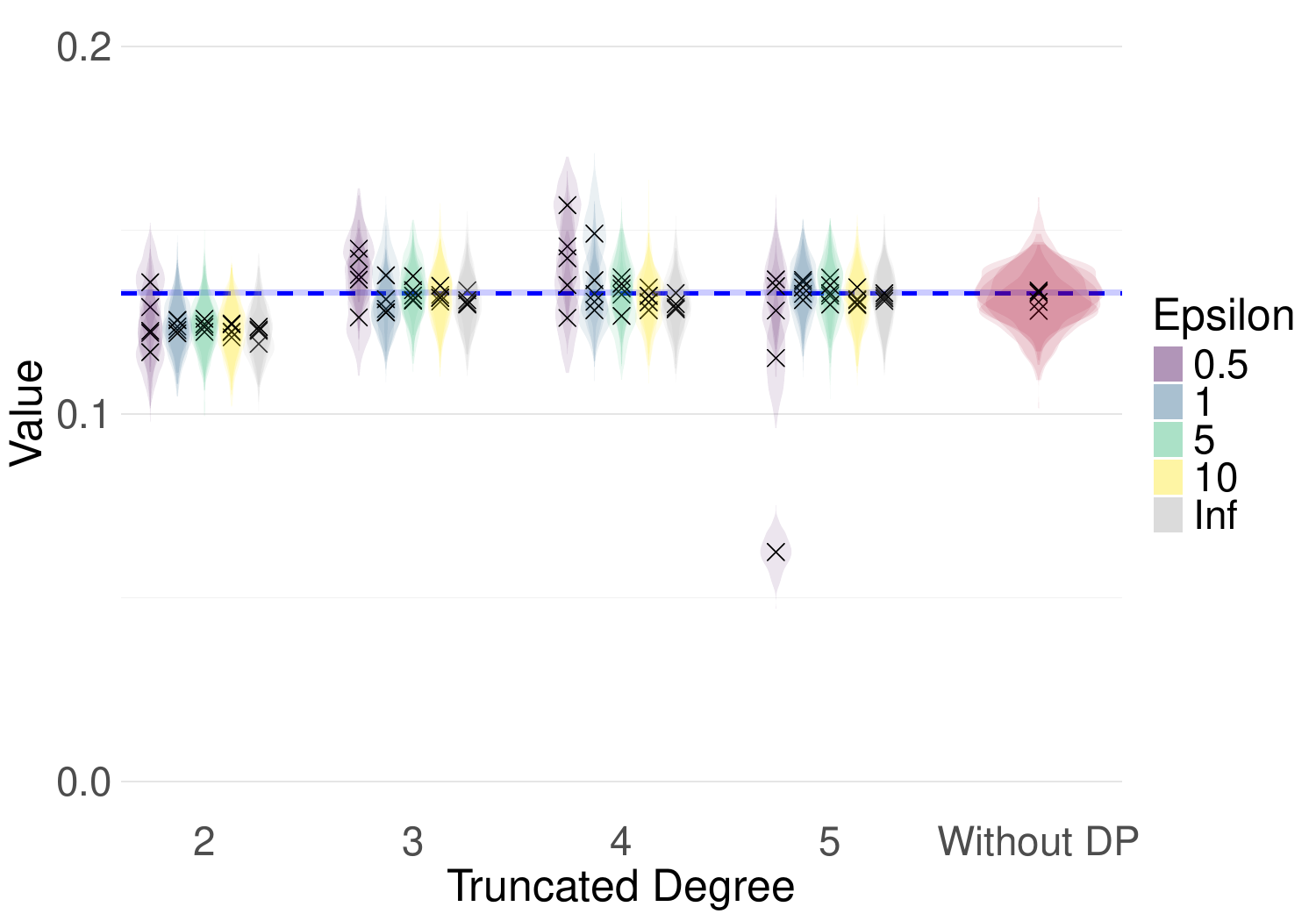}
        \caption{ERGM - High Prevalence}
        \label{fig:prevalence_baseline_ergm_high}
    \end{subfigure}
    
    \vspace{0.5cm}

    \begin{subfigure}[b]{0.48\textwidth}
        \centering
        \includegraphics[width=\textwidth]{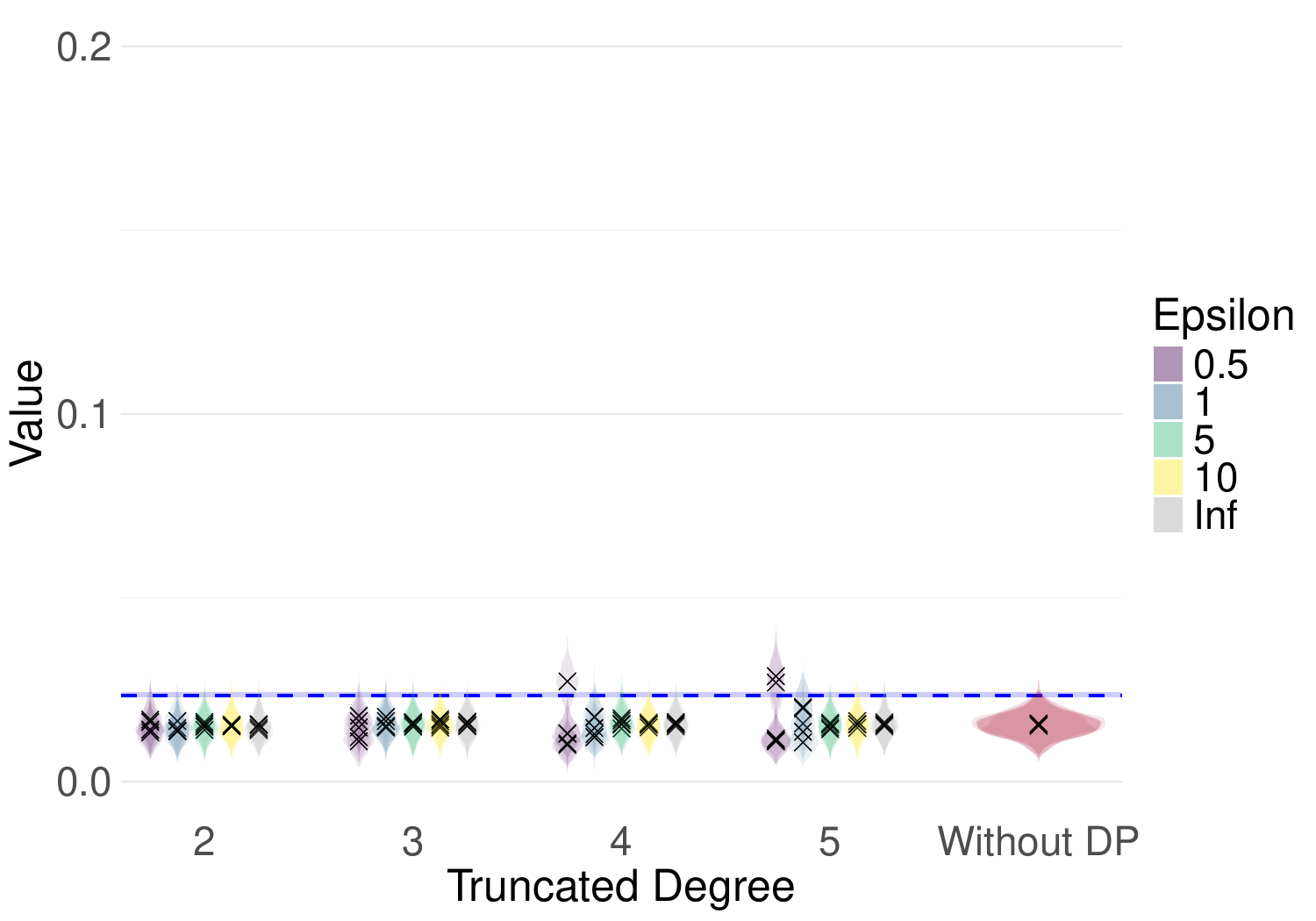}
        \caption{SBM - Low Prevalence}
        \label{fig:prevalence_baseline_bm_low}
    \end{subfigure}
    \hfill
    \begin{subfigure}[b]{0.48\textwidth}
    \centering
    \includegraphics[width=\textwidth]{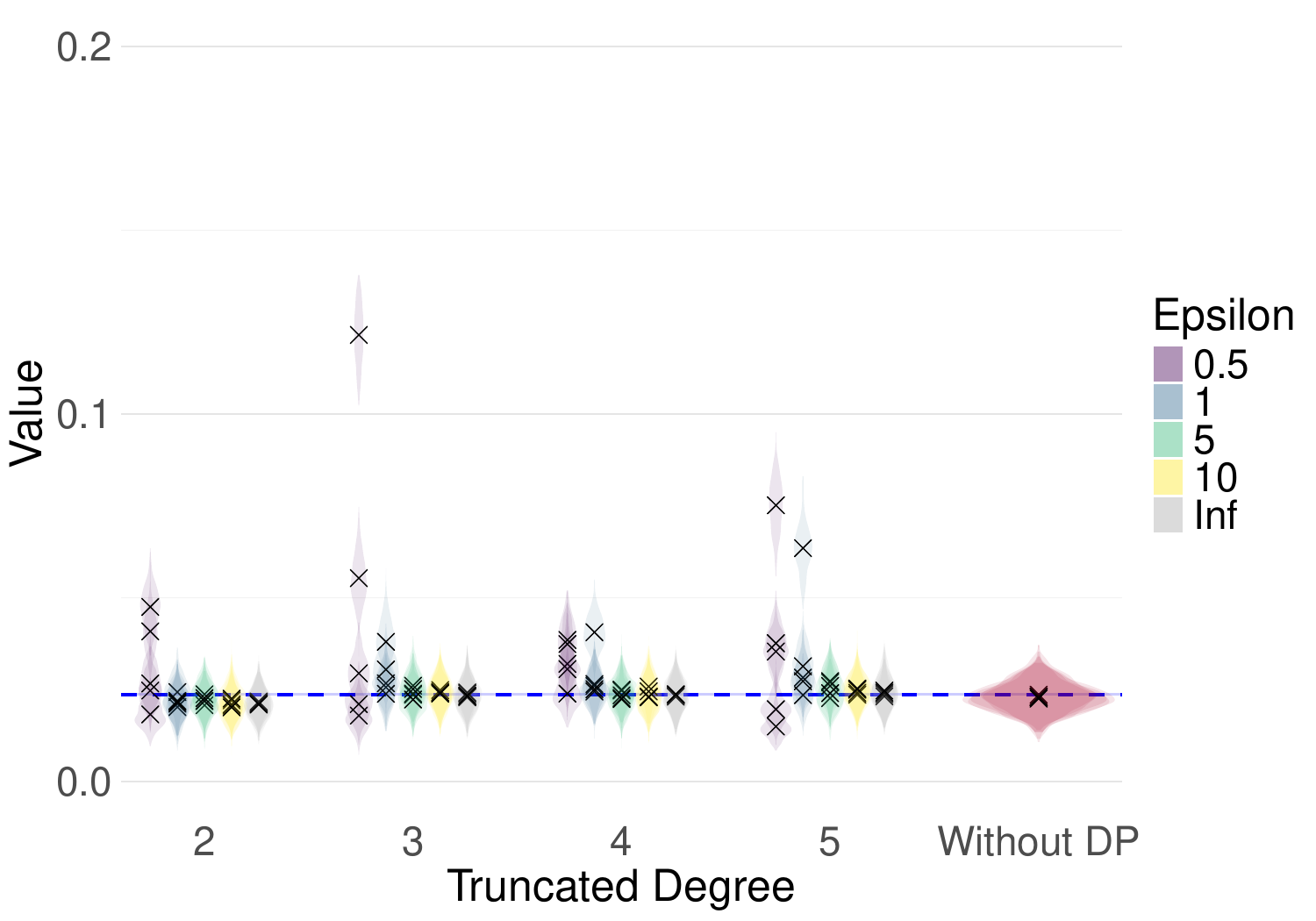}
    \caption{ERGM - Low Prevalence}
    \label{fig:prevalence_baseline_ergm_low}
    \end{subfigure}
    \caption{Prevalence for the baseline scenario across different modeling approaches and prevalence conditions.}
    \label{fig:prevalence_baseline}
\end{figure}

\begin{figure}[h!]
    \centering
    \begin{subfigure}[b]{0.48\textwidth}
        \centering
        \includegraphics[width=\textwidth]{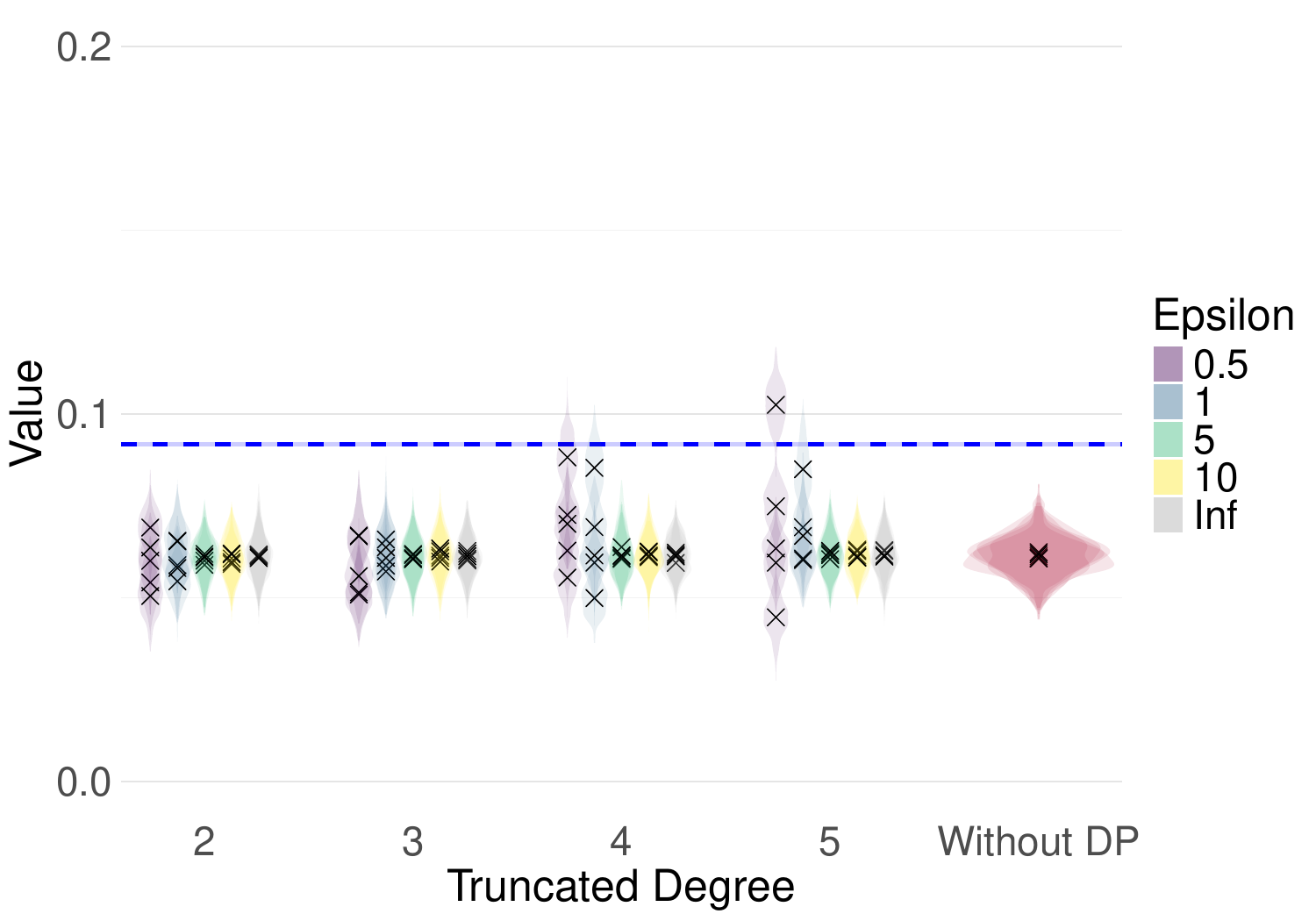}
        \caption{SBM - High Prevalence}
        \label{fig:prevalence_intervention_bm_high}
    \end{subfigure}
    \hfill
    \begin{subfigure}[b]{0.48\textwidth}
        \centering
        \includegraphics[width=\textwidth]{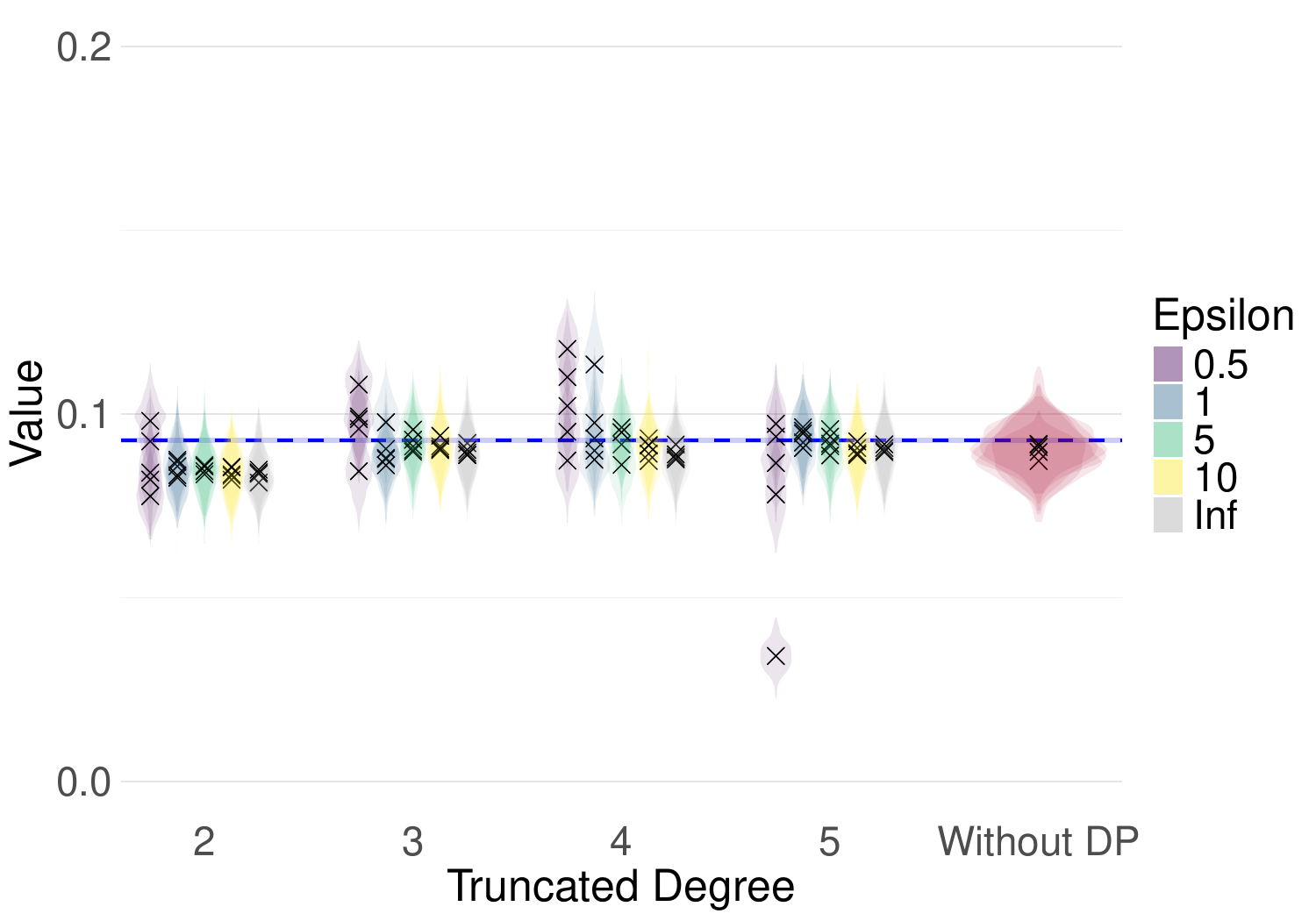}
        \caption{ERGM - High Prevalence}
        \label{fig:prevalence_intervention_ergm_high}
    \end{subfigure}
    
    \vspace{0.5cm}
    
    \hfill
    \begin{subfigure}[b]{0.48\textwidth}
        \centering
        \includegraphics[width=\textwidth]{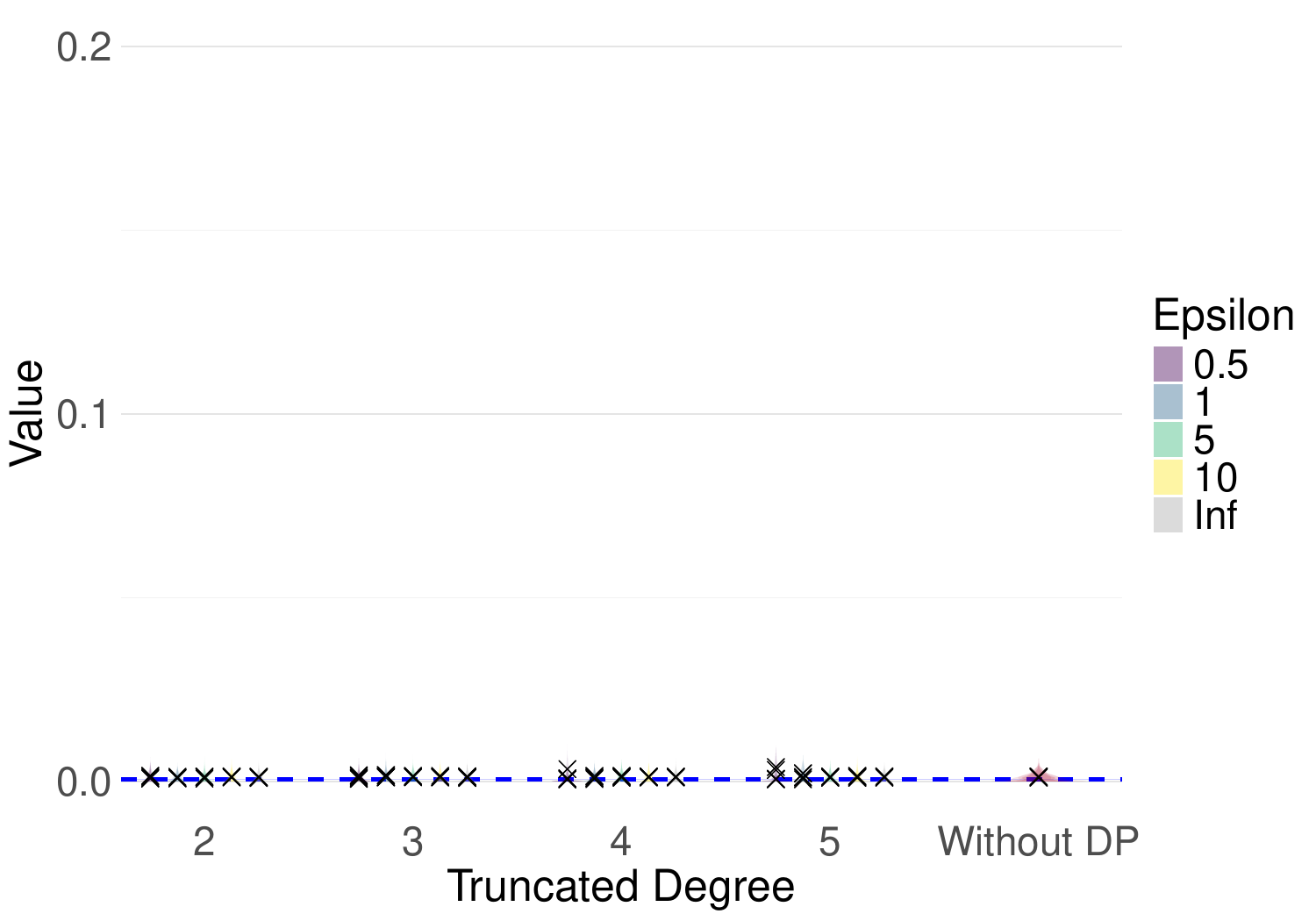}
        \caption{SBM - Low Prevalence}
        \label{fig:prevalence_intervention_bm_low}
    \end{subfigure}
    \hfill
    \begin{subfigure}[b]{0.48\textwidth}
        \centering
        \includegraphics[width=\textwidth]{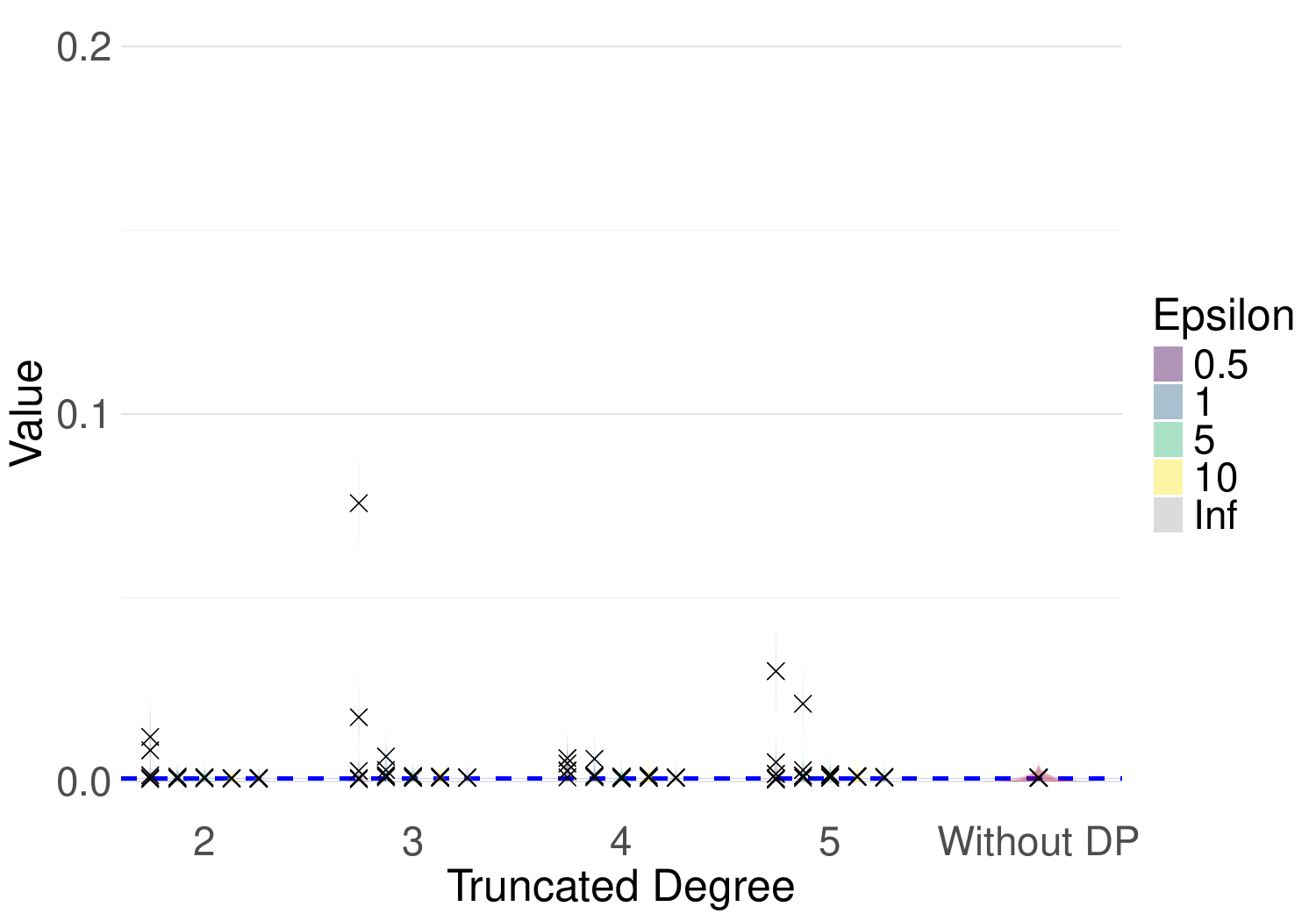}
        \caption{ERGM - Low Prevalence}
        \label{fig:prevalence_intervention_ergm_low}
    \end{subfigure}
    \caption{Prevalence for the intervention scenario across different modeling approaches and prevalence conditions.}
    \label{fig:prevalence_intervention}
\end{figure}

\clearpage

\subsection{Incidence Rate Ratio}

\begin{figure}[h!]
    \centering
    
    \begin{subfigure}[b]{0.48\textwidth}
        \centering
        \includegraphics[width=\textwidth]{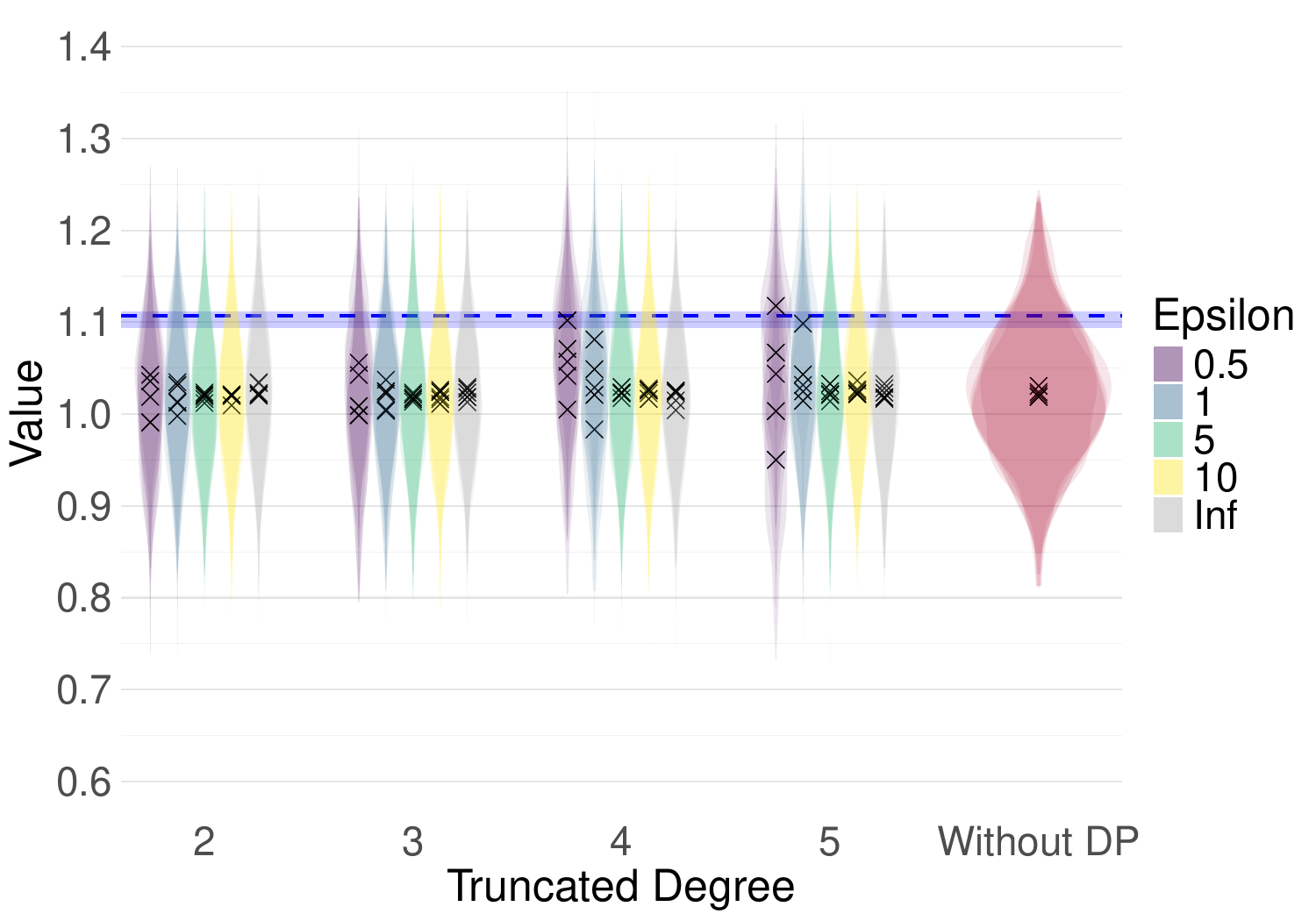}
        \caption{SBM - High Prevalence}
    \end{subfigure}
    \hfill
    \begin{subfigure}[b]{0.48\textwidth}
        \centering
        \includegraphics[width=\textwidth]{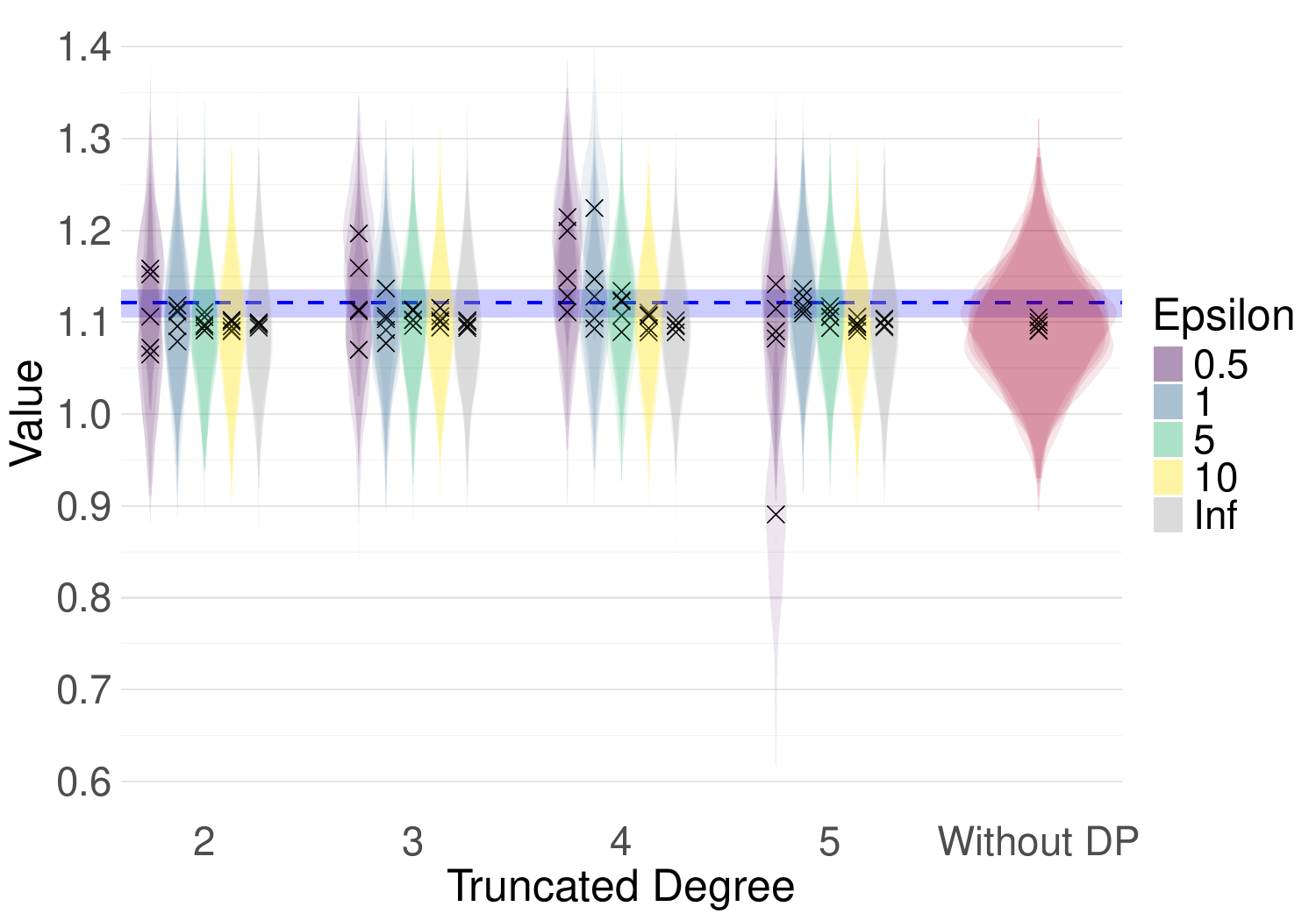}
        \caption{ERGM - High Prevalence}
    \end{subfigure}
    
    \vspace{0.5cm}
    
    \begin{subfigure}[b]{0.48\textwidth}
        \centering
        \includegraphics[width=\textwidth]{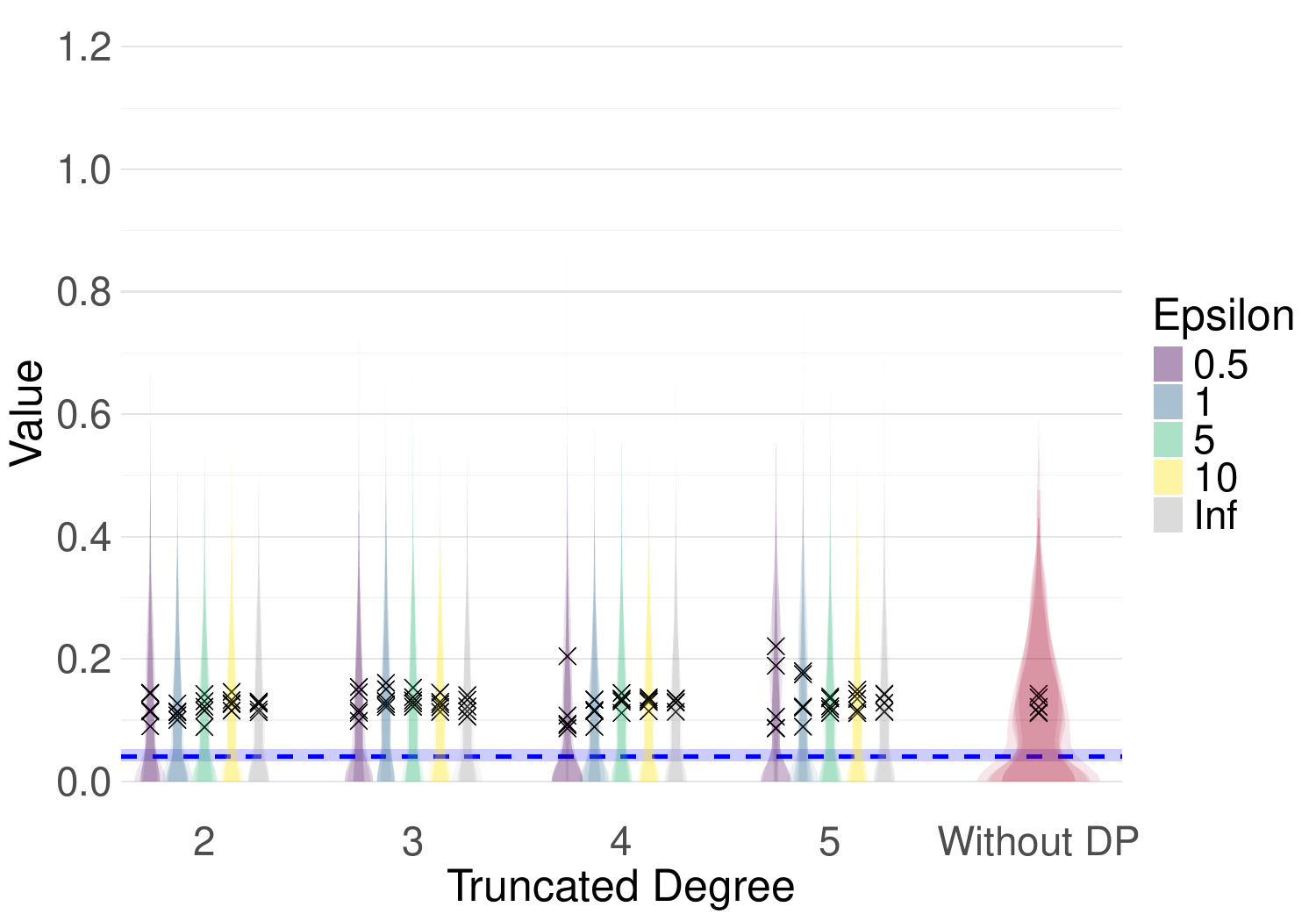}
        \caption{SBM - Low Prevalence}
    \end{subfigure}
    \hfill
    \begin{subfigure}[b]{0.48\textwidth}
        \centering
        \includegraphics[width=\textwidth]{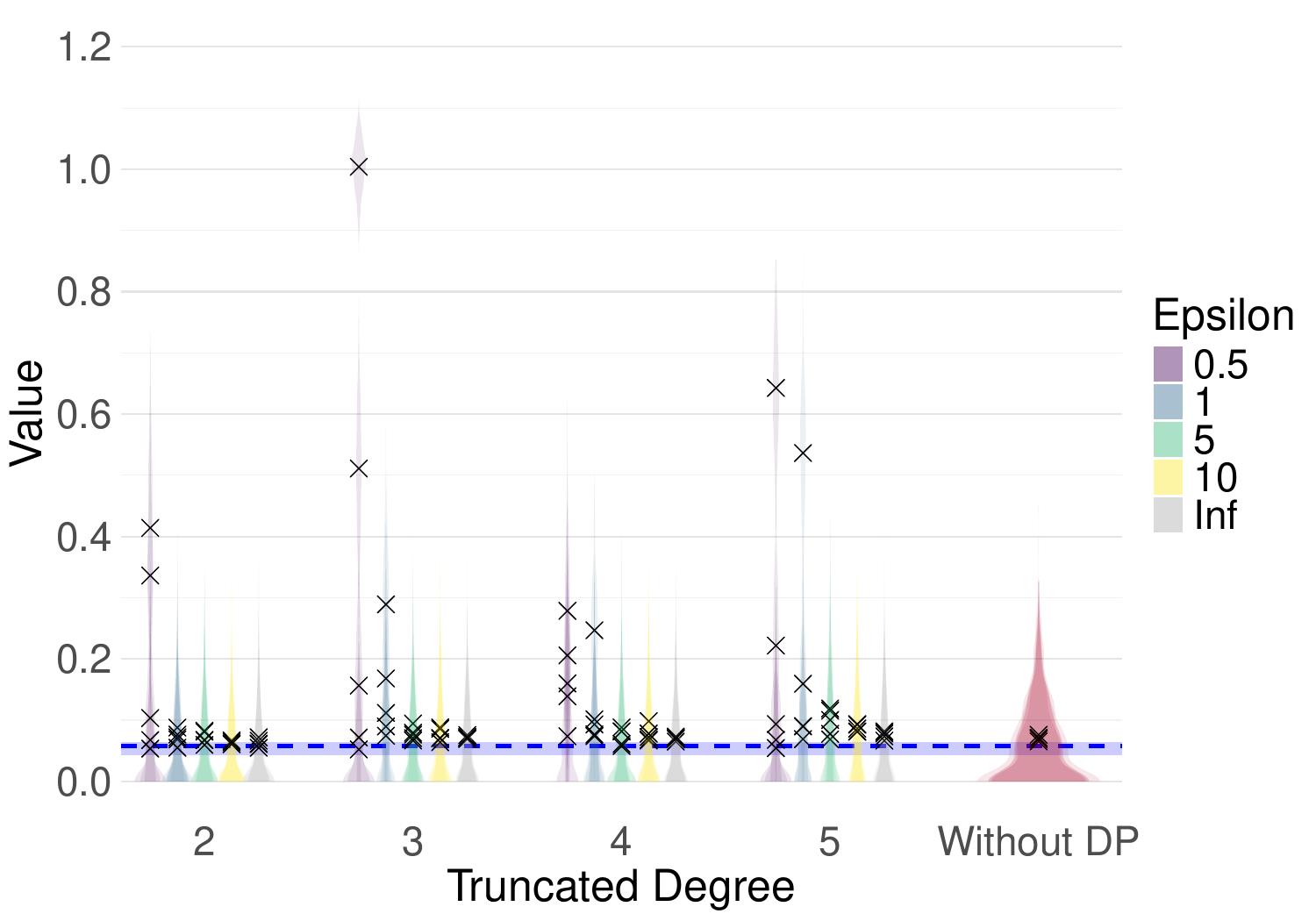}
        \caption{ERGM - Low Prevalence}
    \end{subfigure}

    \caption{Incidence rate ratios comparing intervention effects across different modeling approaches and prevalence conditions.}
    \label{fig:incidence_ratio}
\end{figure}

\clearpage

\subsection{Granular Analysis by Age and Race}

\begin{figure}[h!]
    \centering
    \begin{subfigure}[b]{0.48\textwidth}
        \centering
        \includegraphics[width=\textwidth]{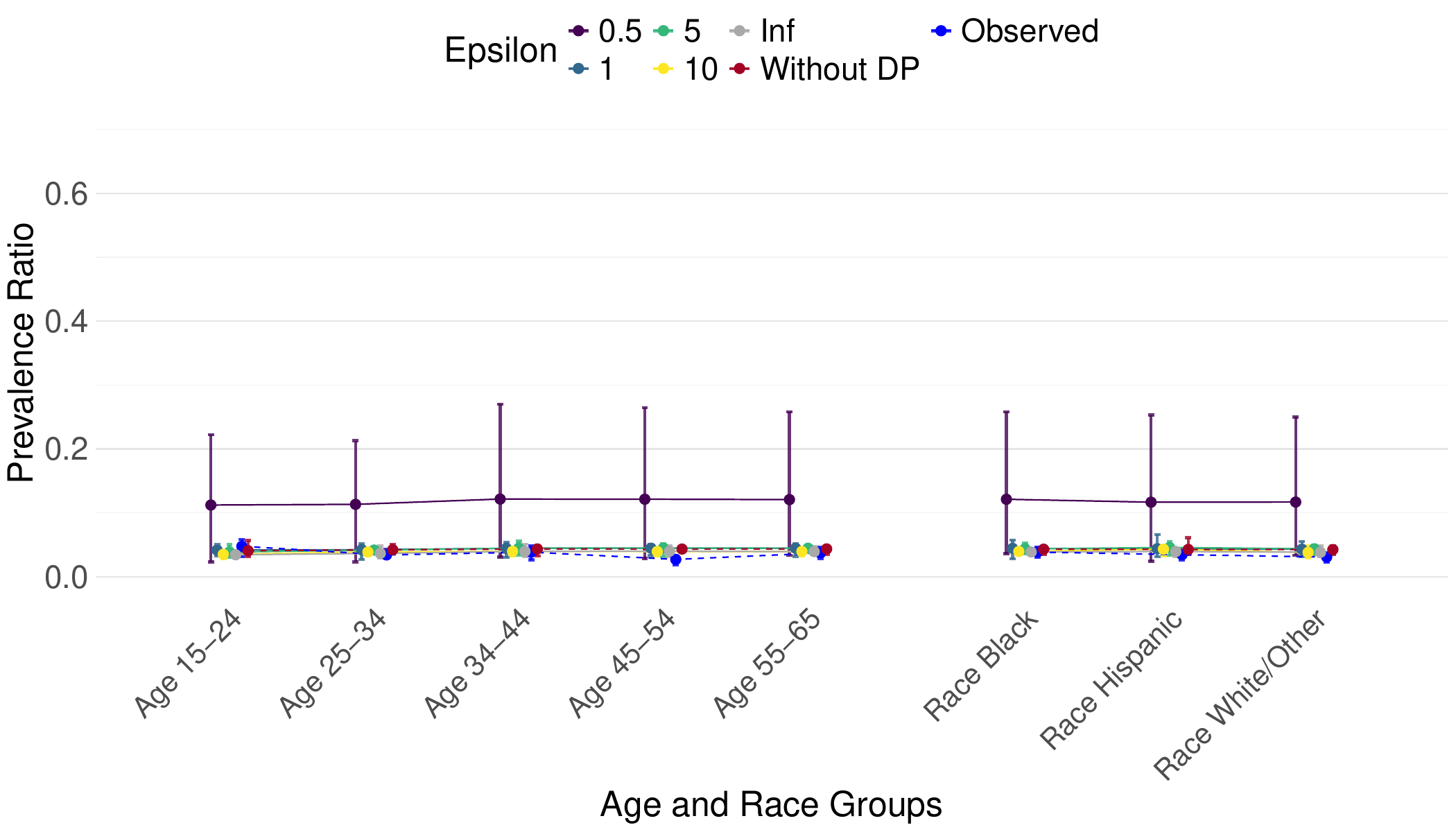}
        \caption{Truncated Degree 2}
        \label{fig:granular_ergm_low_truncated_2}
    \end{subfigure}
    \hfill
    \begin{subfigure}[b]{0.48\textwidth}
        \centering
        \includegraphics[width=\textwidth]{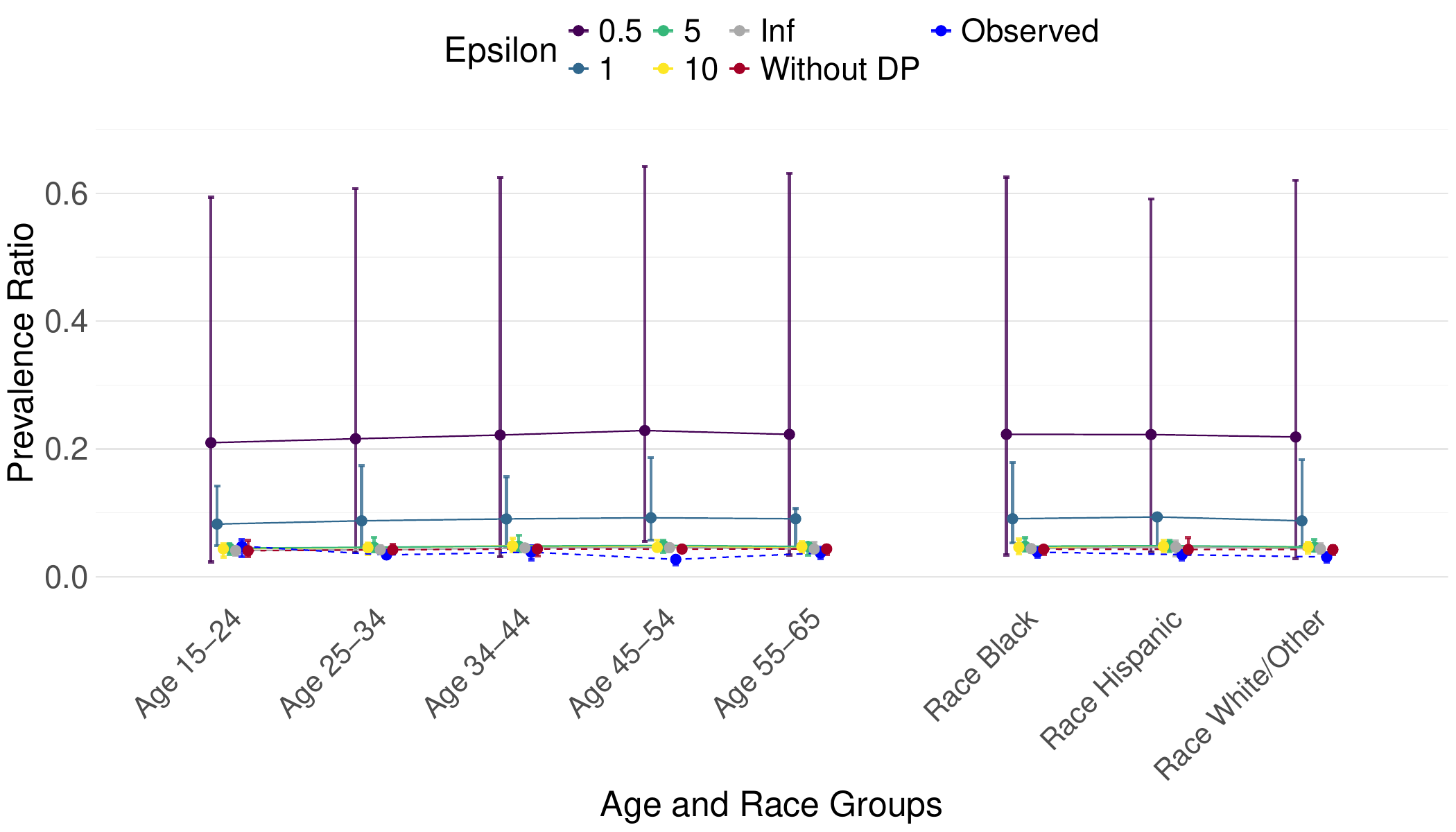}
        \caption{Truncated Degree 3}
    \end{subfigure}
    
    \vspace{0.5cm}
    
    \begin{subfigure}[b]{0.48\textwidth}
        \centering
        \includegraphics[width=\textwidth]{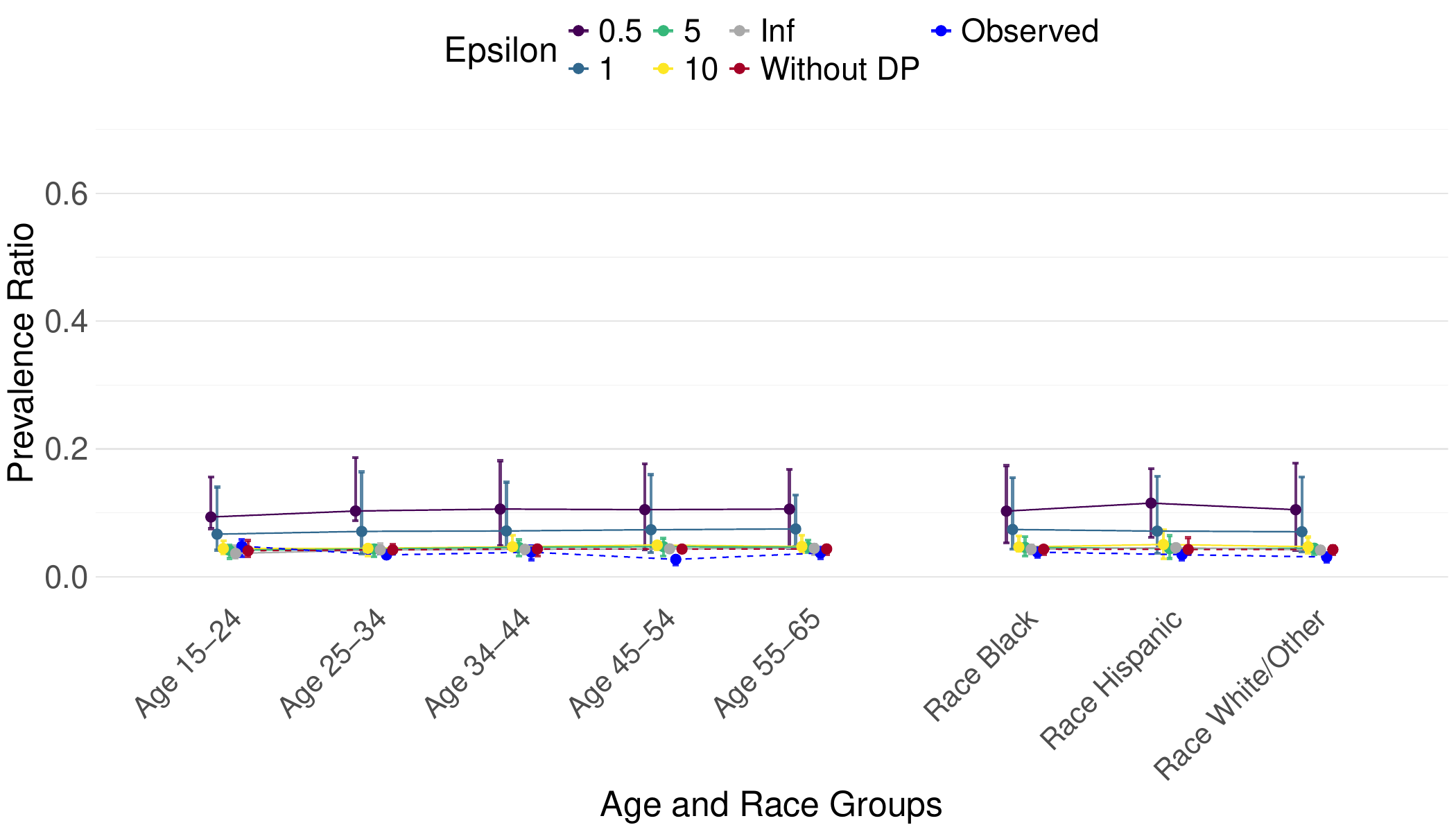}
        \caption{Truncated Degree 4}
    \end{subfigure}
    \hfill
    \begin{subfigure}[b]{0.48\textwidth}
        \centering
        \includegraphics[width=\textwidth]{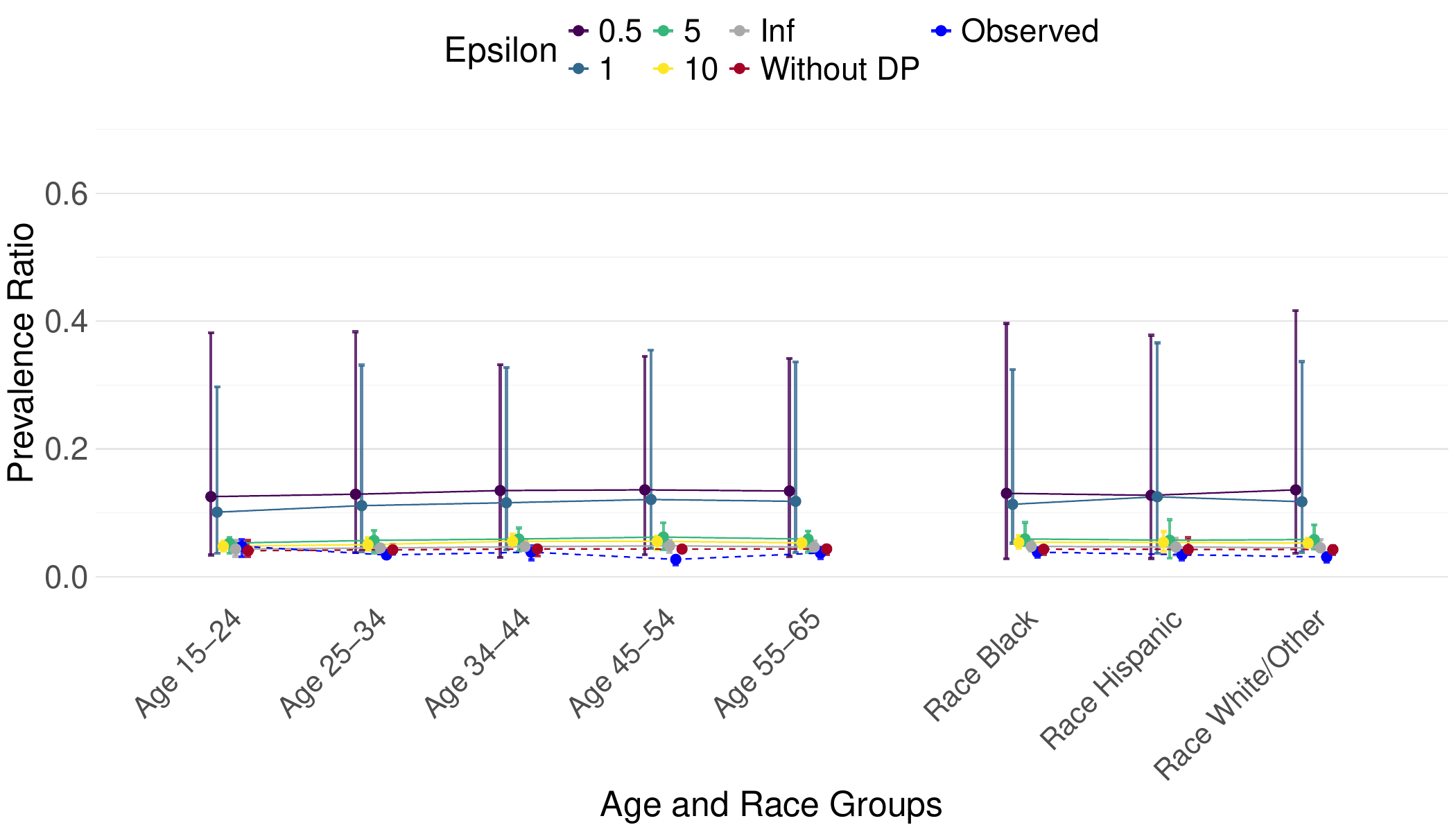}
        \caption{Truncated Degree 5}
    \end{subfigure}
    \caption{Granular analysis by age and race for ERGM low prevalence condition across different maximum degree constraints.}
    \label{fig:granular_ergm_low}
\end{figure}

\begin{figure}[h!]
    \centering
    \begin{subfigure}[b]{0.48\textwidth}
        \centering
        \includegraphics[width=\textwidth]{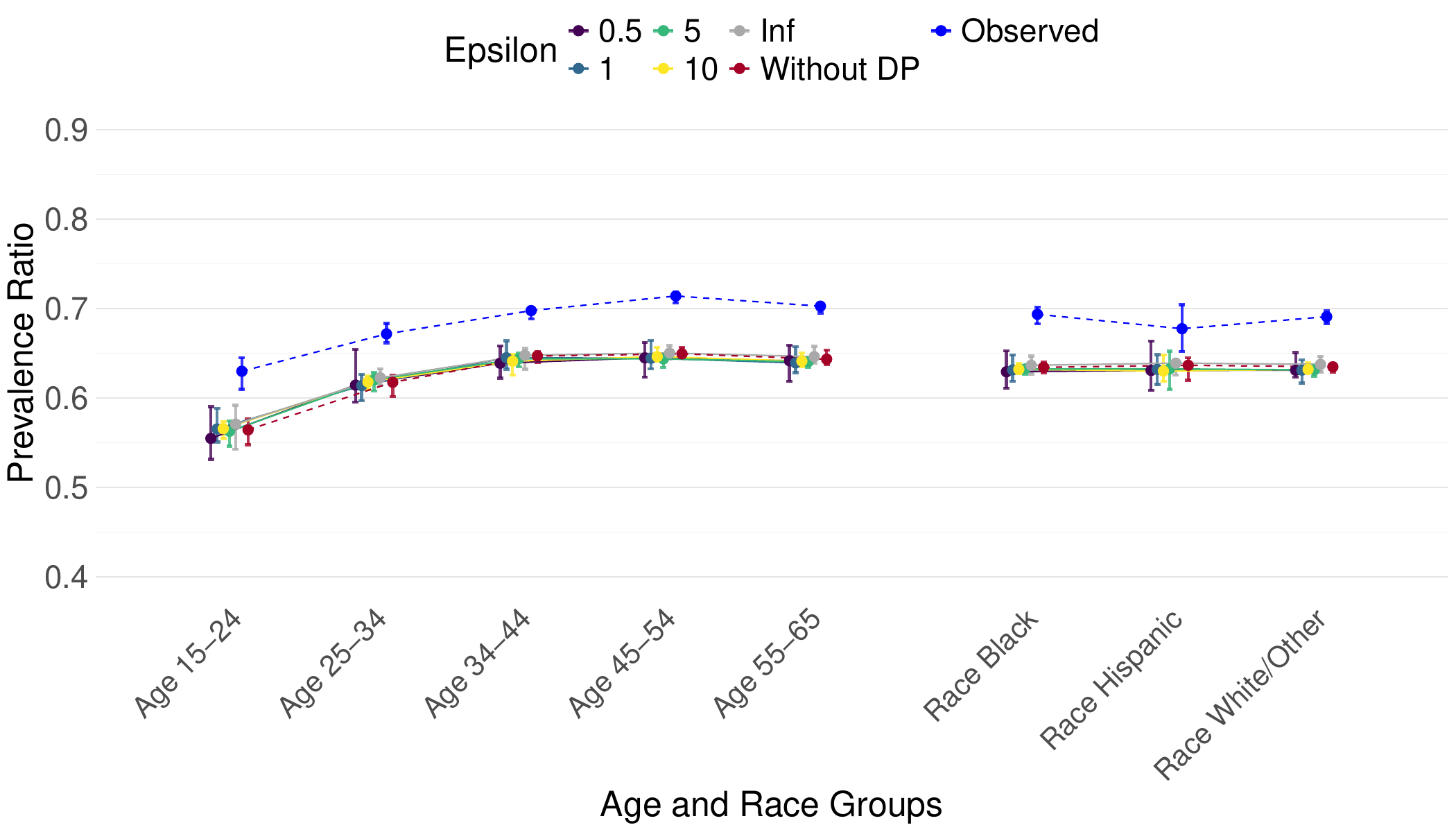}
        \caption{Truncated Degree 2}
        \label{fig:granular_bm_low_truncated_2}
    \end{subfigure}
    \hfill
    \begin{subfigure}[b]{0.48\textwidth}
        \centering
        \includegraphics[width=\textwidth]{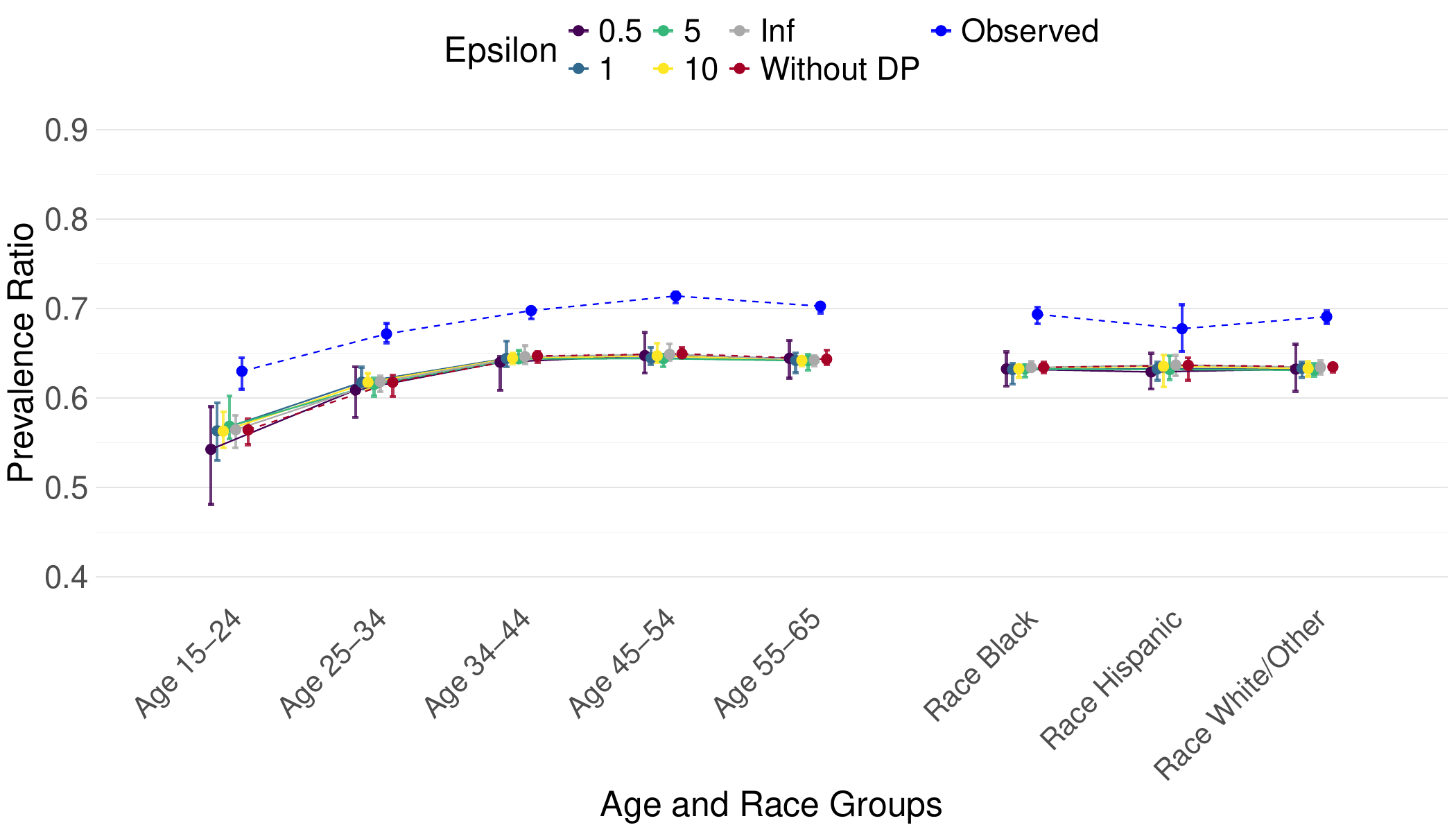}
        \caption{Truncated Degree 3}
    \end{subfigure}
    
    \vspace{0.5cm}
    
    \begin{subfigure}[b]{0.48\textwidth}
        \centering
        \includegraphics[width=\textwidth]{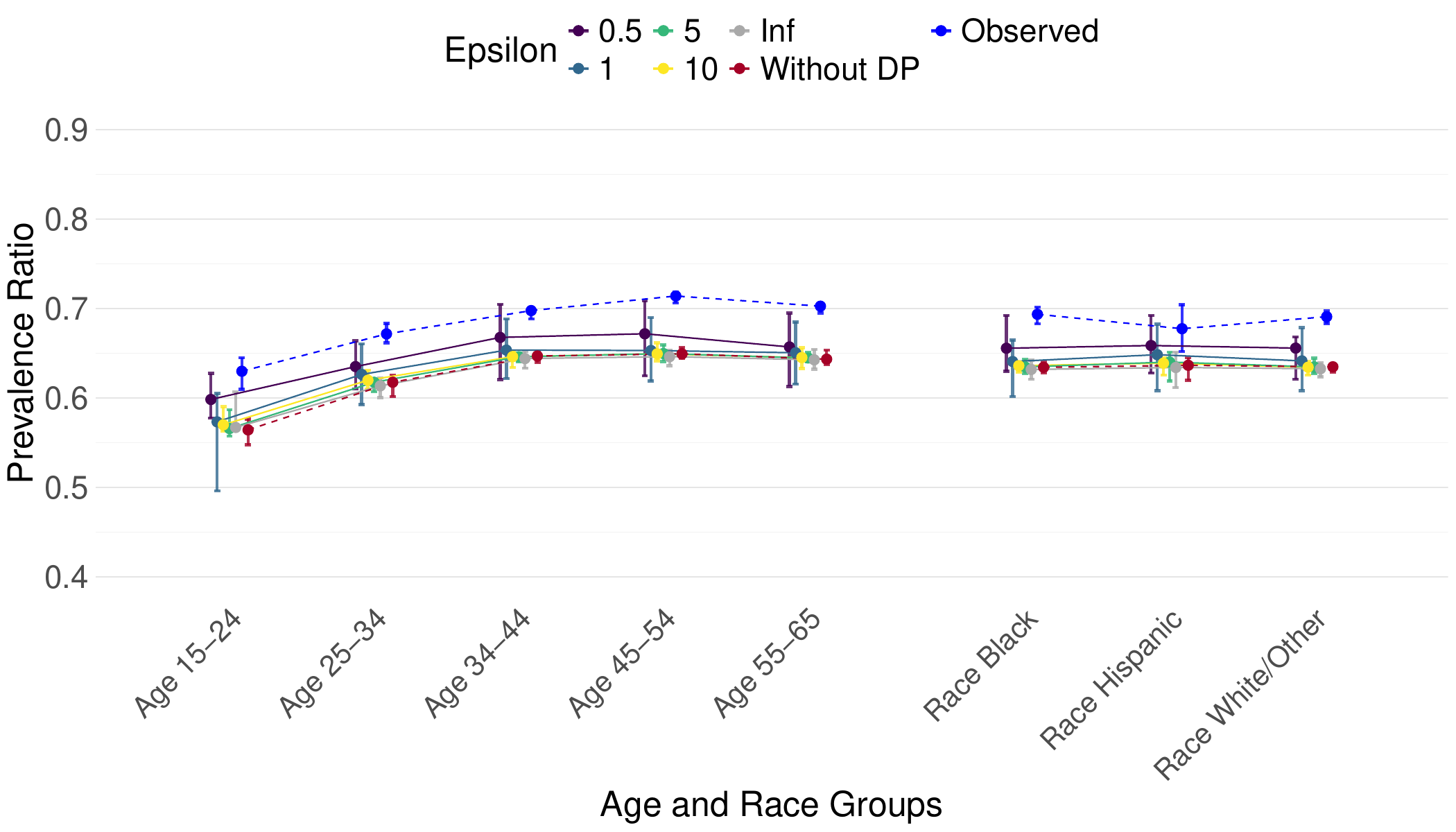}
        \caption{Truncated Degree 4}
    \end{subfigure}
    \hfill
    \begin{subfigure}[b]{0.48\textwidth}
        \centering
        \includegraphics[width=\textwidth]{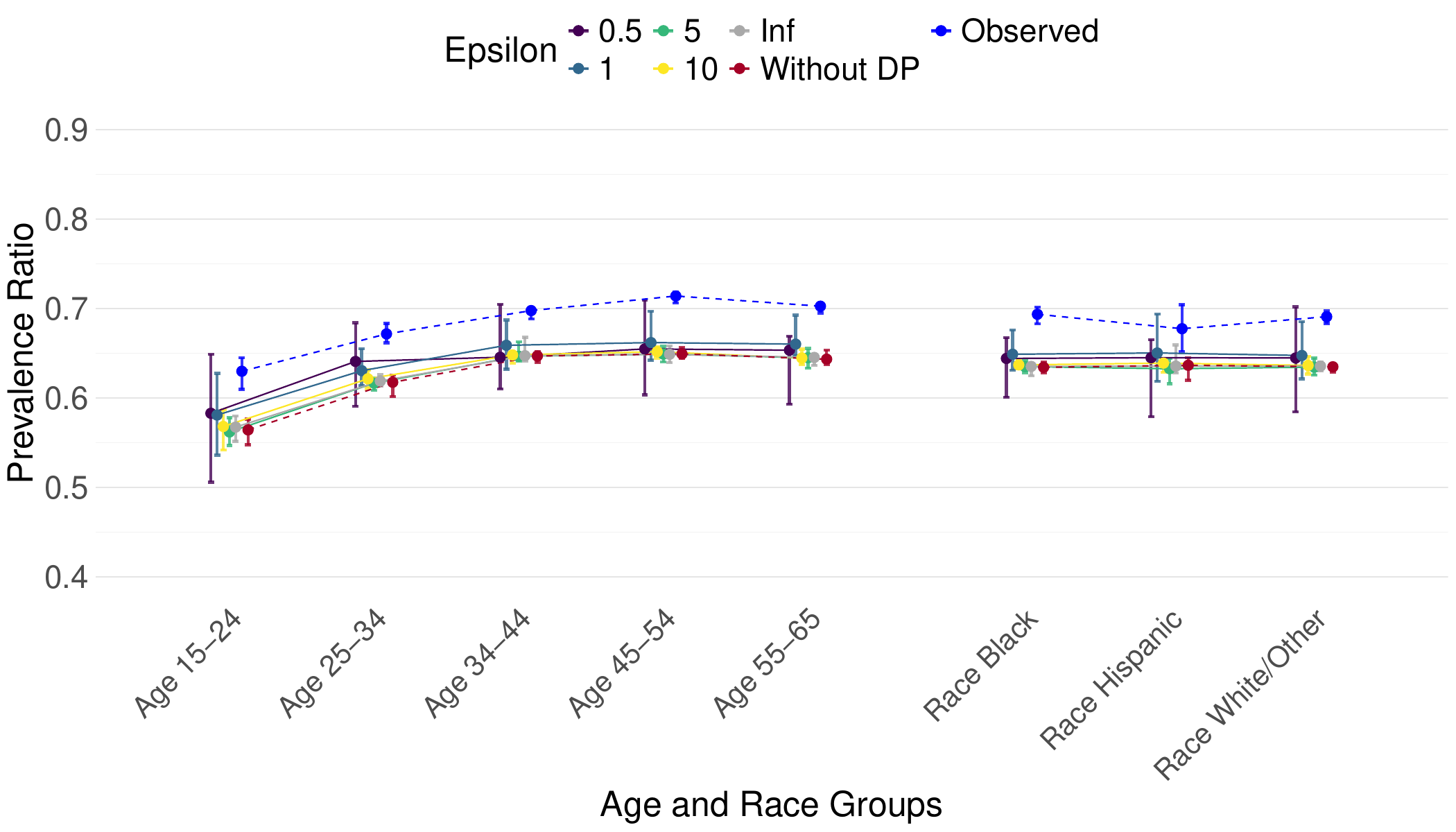}
        \caption{Truncated Degree 5}
    \end{subfigure}
    \caption{Granular analysis by age and race for BM high prevalence condition across different maximum degree constraints.}
    \label{fig:granular_bm_high}
\end{figure}

\begin{figure}[h!]
    \centering
    \begin{subfigure}[b]{0.48\textwidth}
        \centering
        \includegraphics[width=\textwidth]{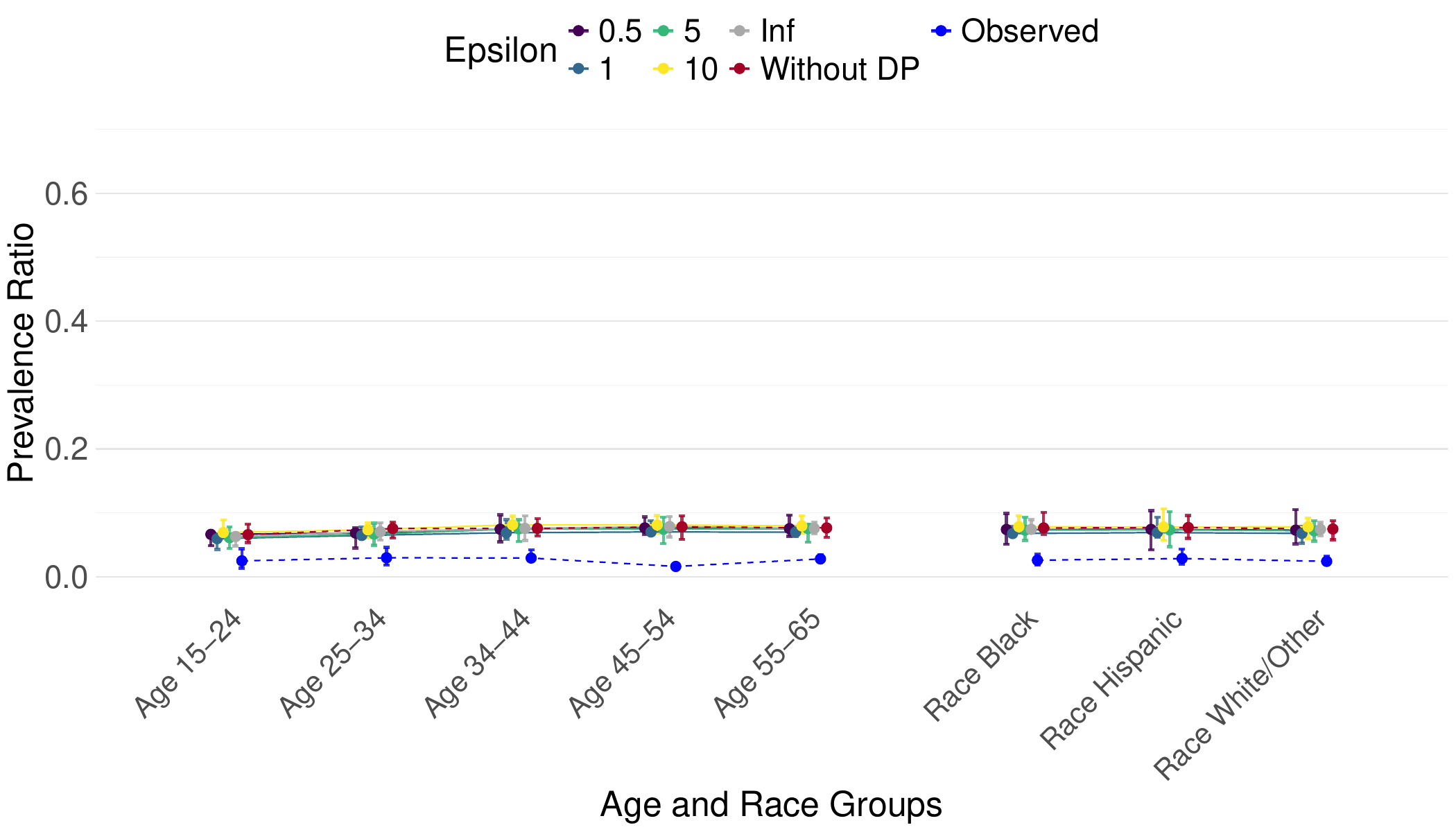}
        \caption{Truncated Degree 2}
        \label{fig:granular_bm_low_truncated_2}
    \end{subfigure}
    \hfill
    \begin{subfigure}[b]{0.48\textwidth}
        \centering
        \includegraphics[width=\textwidth]{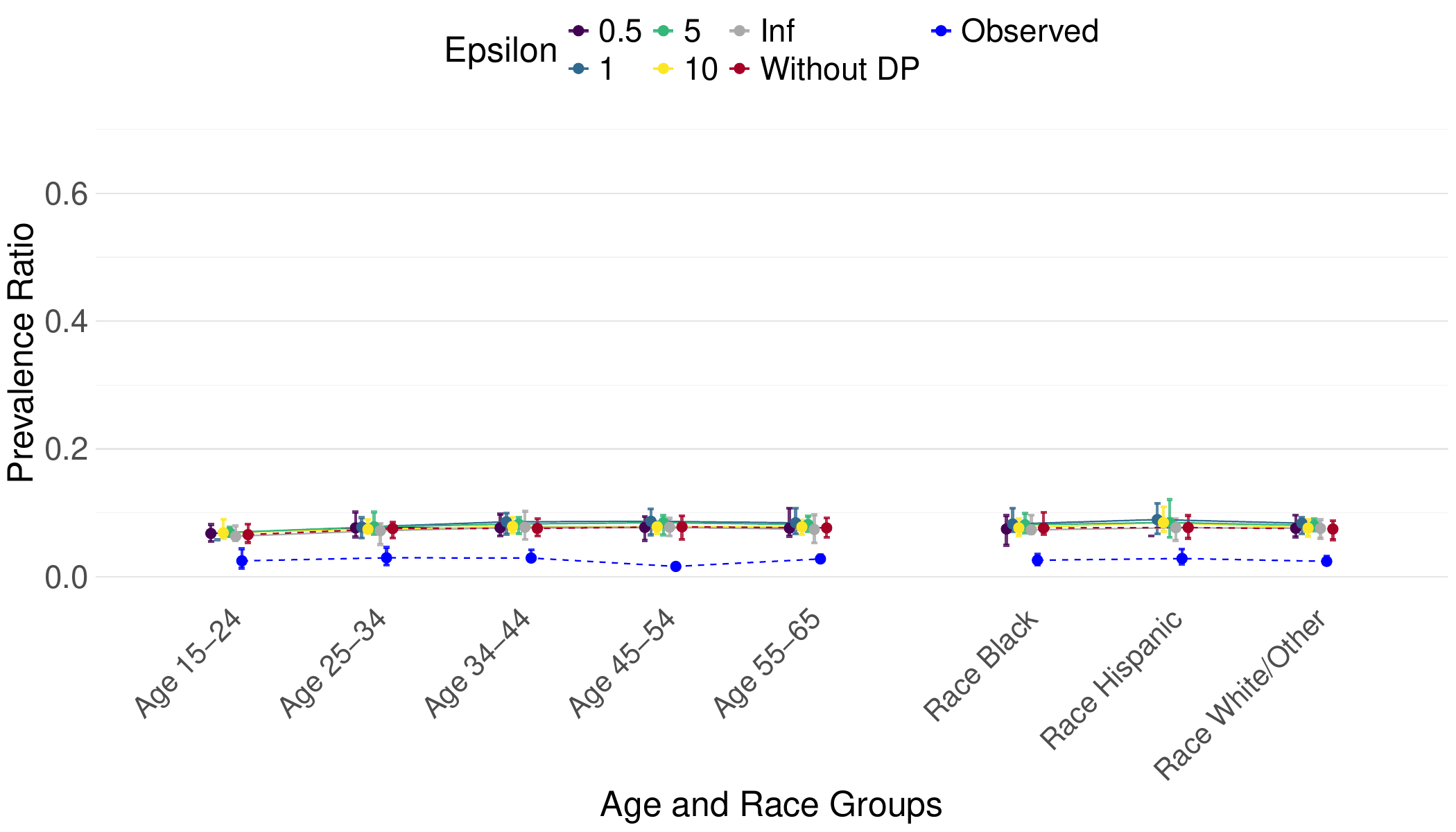}
        \caption{Truncated Degree 3}
    \end{subfigure}
    
    \vspace{0.5cm}
    
    \begin{subfigure}[b]{0.48\textwidth}
        \centering
        \includegraphics[width=\textwidth]{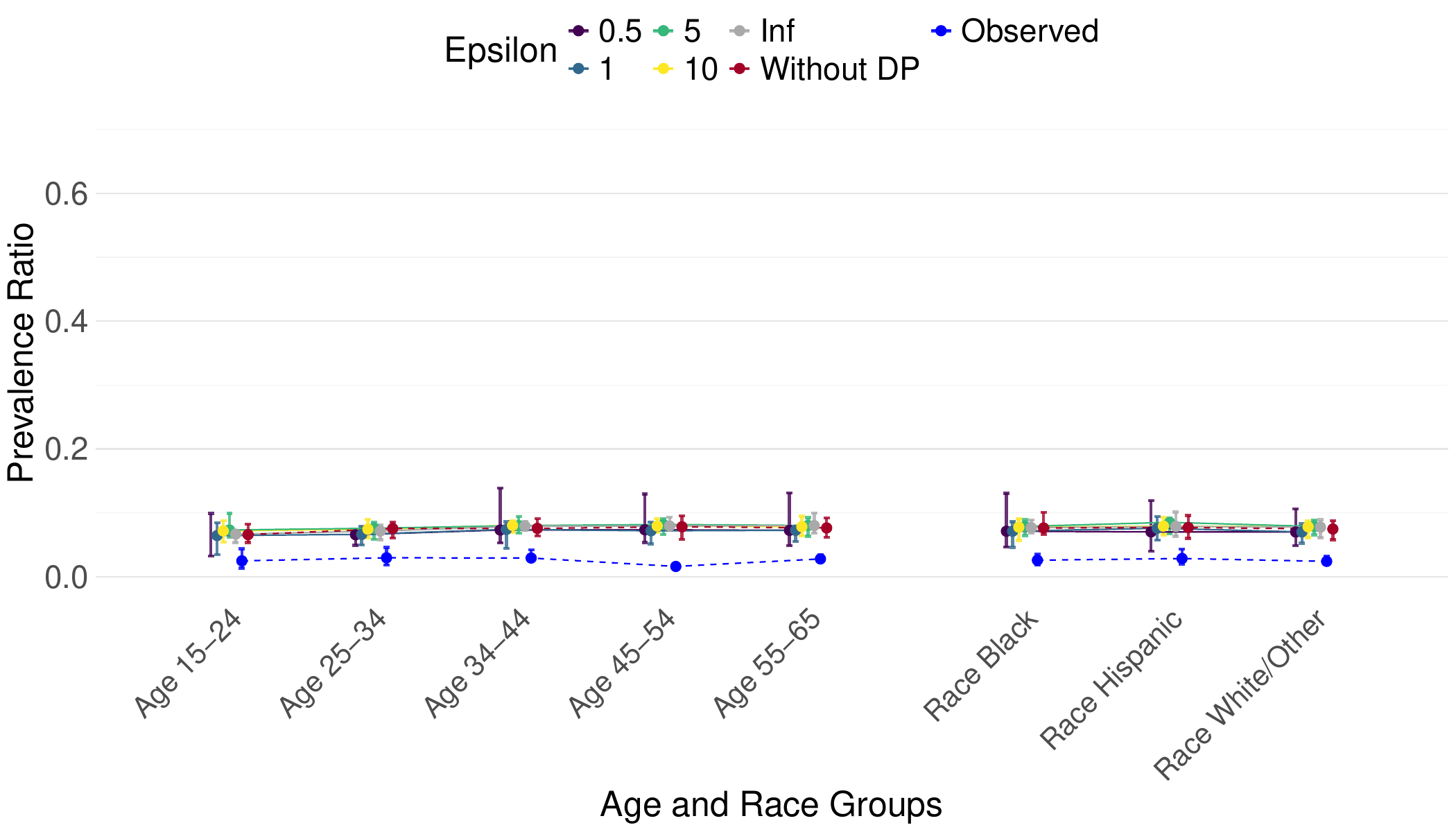}
        \caption{Truncated Degree 4}
    \end{subfigure}
    \hfill
    \begin{subfigure}[b]{0.48\textwidth}
        \centering
        \includegraphics[width=\textwidth]{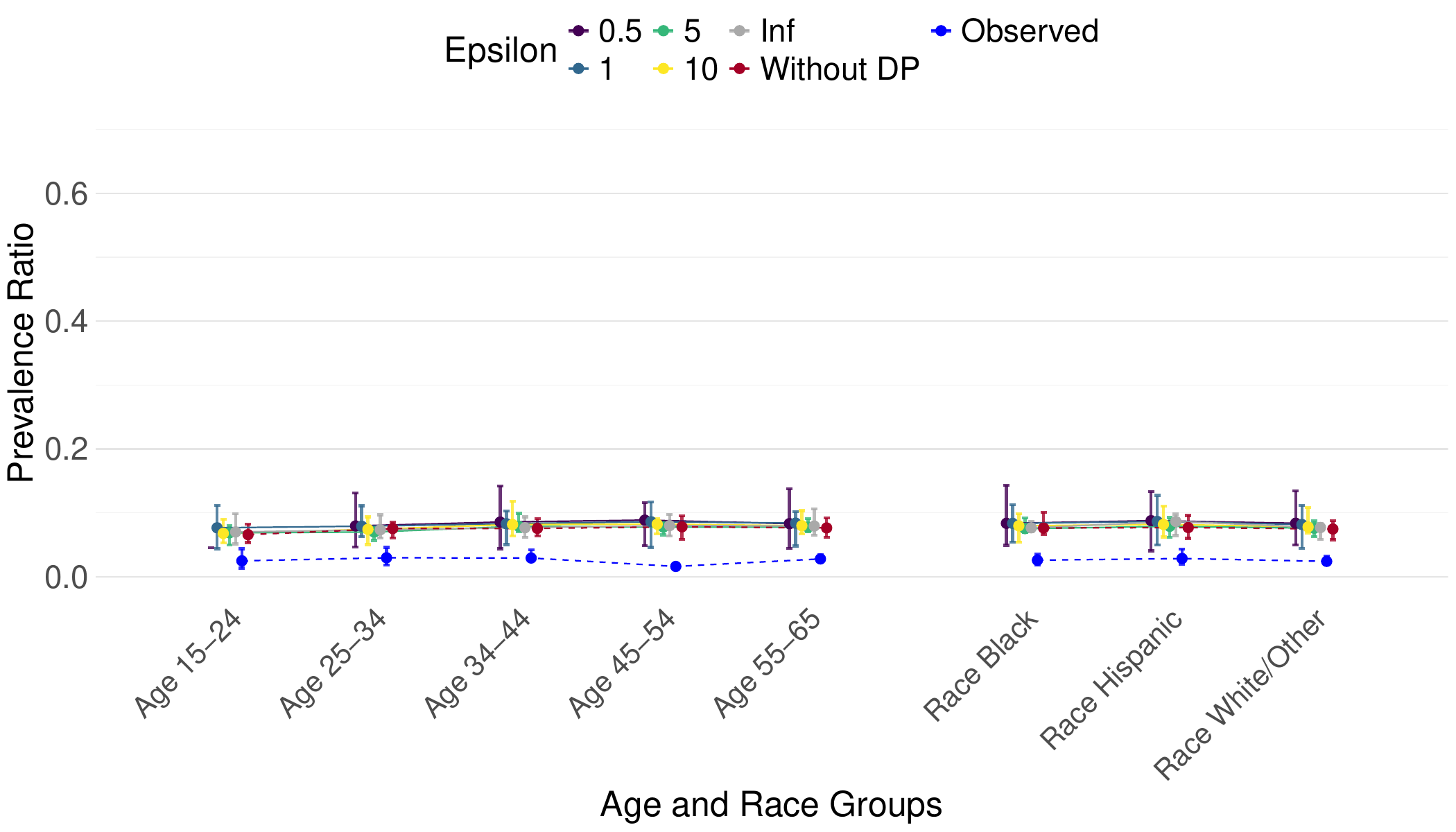}
        \caption{Truncated Degree 5}
    \end{subfigure}
    \caption{Granular analysis by age and race for BM low prevalence condition across different maximum degree constraints.}
    \label{fig:granular_bm_low}
\end{figure}

\clearpage

\subsection{Variance Analysis}
\label{app:variance-analysis}

Our experimental pipeline contains three nested levels of stochasticity: (i) $R$ independent differential privacy releases, (ii) $N$ synthetic networks per release, and (iii) $M$ epidemic simulations per network. This yields $RNM$ total observations of the response variable $y_{ijk}$, representing mean baseline prevalence for simulation $k$ on network $j$ from release $i$.

Under this balanced nested design, we partition the total sum of squares $SS_{\text{total}} = \sum_{i,j,k}(y_{ijk} - \bar{y}_{\cdot\cdot\cdot})^2$ into three orthogonal components, where $\bar{y}_{\cdot\cdot\cdot}$ denotes the grand mean across all observations:

\begin{align}
SS_R &= NM\sum_i(\bar{y}_{i\cdot\cdot} - \bar{y}_{\cdot\cdot\cdot})^2 & \text{(between-release)} \\
SS_{N|R} &= M\sum_{i,j}(\bar{y}_{ij\cdot} - \bar{y}_{i\cdot\cdot})^2 & \text{(between-network within-release)} \\
SS_E &= SS_{\text{total}} - SS_R - SS_{N|R} & \text{(within-network simulation)}
\end{align}

where $\bar{y}_{i\cdot\cdot}$ is the mean for release $i$ (averaged over all networks and simulations within that release), and $\bar{y}_{ij\cdot}$ is the mean for network $j$ in release $i$ (averaged over all simulations on that network).

The proportions $SS_R/SS_{\text{total}}$, $SS_{N|R}/SS_{\text{total}}$, and $SS_E/SS_{\text{total}}$ quantify the relative contribution of each source to total variation.

This decomposition requires no distributional assumptions, relying only on the balanced structure to exactly partition variation across the three hierarchical sources of randomness in our privacy-preserving epidemic modeling pipeline.

\subsubsection{Additional Results of Variance Analysis}
\label{app:variance-analysis-results}

The following tables provide quantitative breakdowns of the variance sources visualized in the plots above, showing the degrees of freedom (df), sum of squares (SS), mean squares (MS), and percentage of total variance explained by each source.

\begin{table}[h!]
\centering
\begin{tabular}{lcccc}
\toprule
Source & df & SS & MS & Var (\%) \\
\midrule
Release & 4 & 0.0506 & 0.0127 & 54.14 \\
Network : Release & 195 & 0.0257 & 0.0001 & 27.47 \\
Simulation : Network : Release & 1800 & 0.0172 & 0.0000 & 18.39 \\
\bottomrule
\end{tabular}
\caption{Analysis of variance for ERGM low prevalence condition showing the breakdown of variance sources in the experimental pipeline.}
\label{tab:anova_ergm_low}
\end{table}

\begin{table}[h!]
\centering
\begin{tabular}{lcccc}
\toprule
Source & df & SS & MS & Var (\%) \\
\midrule
Release & 4 & 0.0282 & 0.0070 & 24.86 \\
Network : Release & 195 & 0.0642 & 0.0003 & 56.66 \\
Simulation : Network : Release & 1800 & 0.0209 & 0.0000 & 18.47 \\
\bottomrule
\end{tabular}
\caption{Analysis of variance for SBM high prevalence condition showing the breakdown of variance sources in the experimental pipeline.}
\label{tab:anova_bm_high}
\end{table}

\begin{table}[h!]
\centering
\begin{tabular}{lcccc}
\toprule
Source & df & SS & MS & Var (\%) \\
\midrule
Release & 4 & 0.0024 & 0.0006 & 9.46 \\
Network : Release & 195 & 0.0129 & 0.0001 & 51.13 \\
Simulation : Network : Release & 1800 & 0.0099 & 0.0000 & 39.41 \\
\bottomrule
\end{tabular}
\caption{Analysis of variance for SBM low prevalence condition showing the breakdown of variance sources in the experimental pipeline.}
\label{tab:anova_bm_low}
\end{table}

\begin{figure}[h!]
    \centering
    
    \begin{subfigure}[b]{0.5\textwidth}
        \centering
        \includegraphics[width=\textwidth]{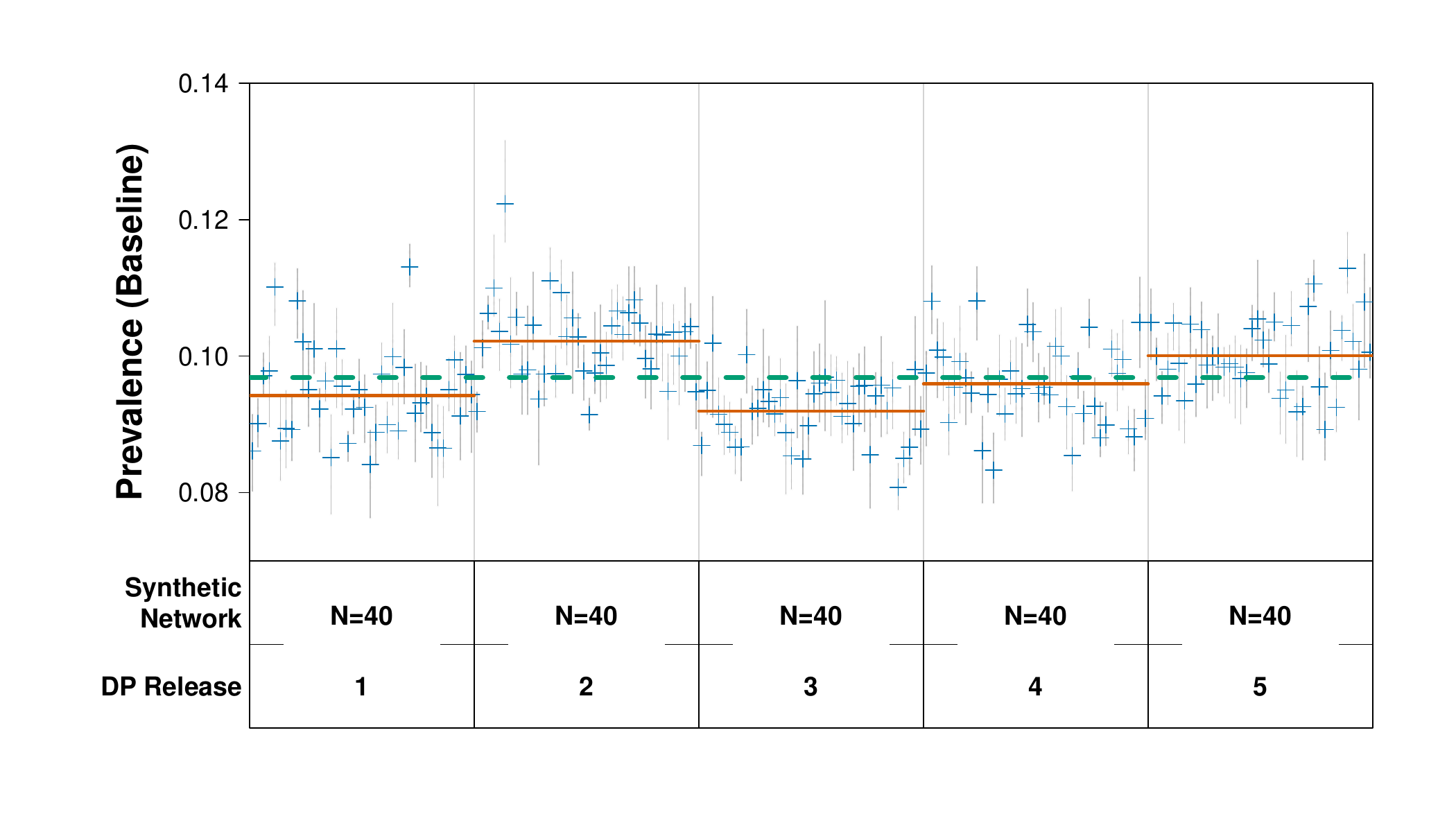}
        \caption{SBM - High Prevalence}
    \end{subfigure}

    \begin{subfigure}[b]{0.5\textwidth}
        \centering
        \includegraphics[width=\textwidth]{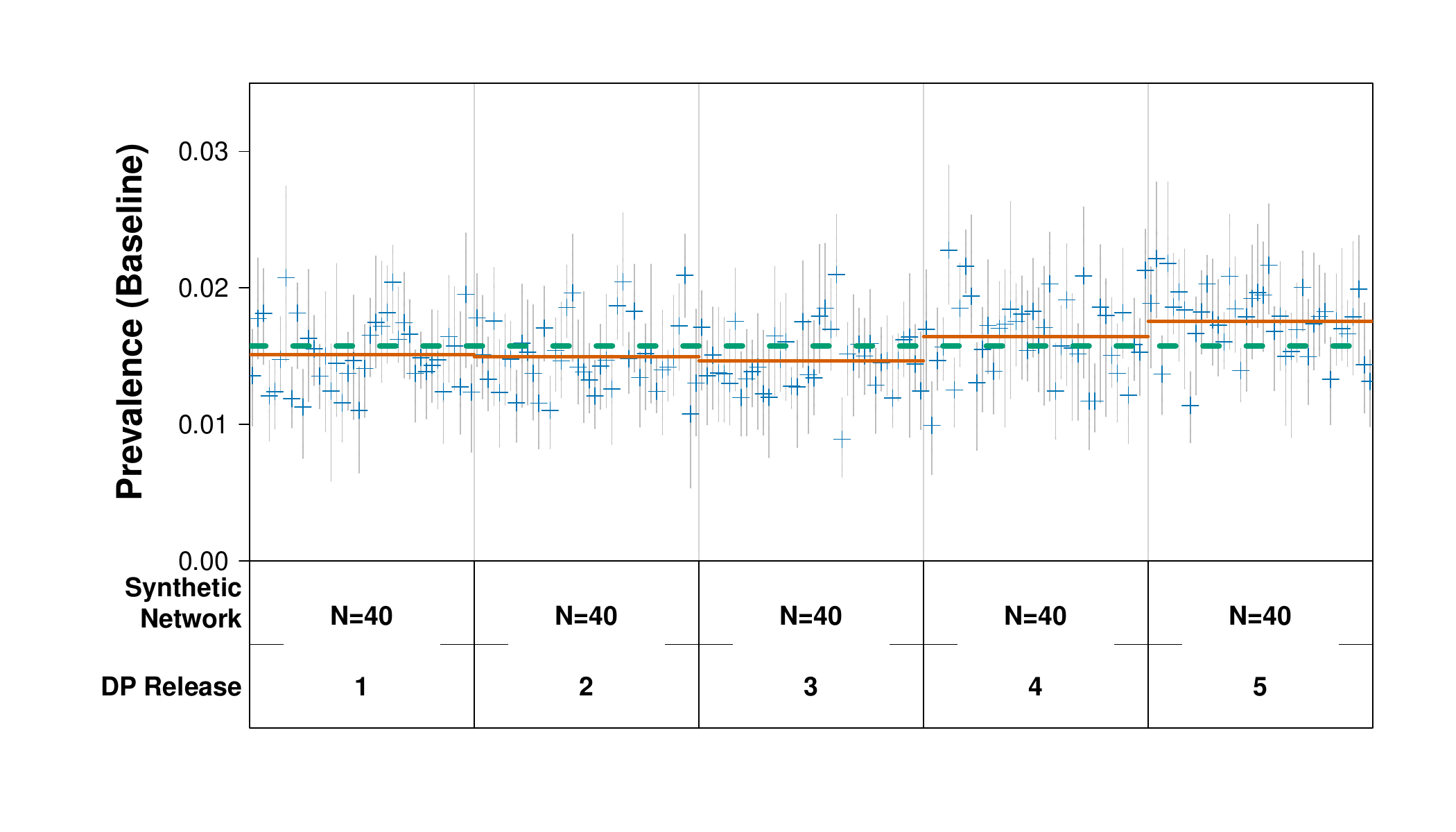}
        \caption{SBM - Low Prevalence}
    \end{subfigure}
    
    \begin{subfigure}[b]{0.5\textwidth}
        \centering
        \includegraphics[width=\textwidth]{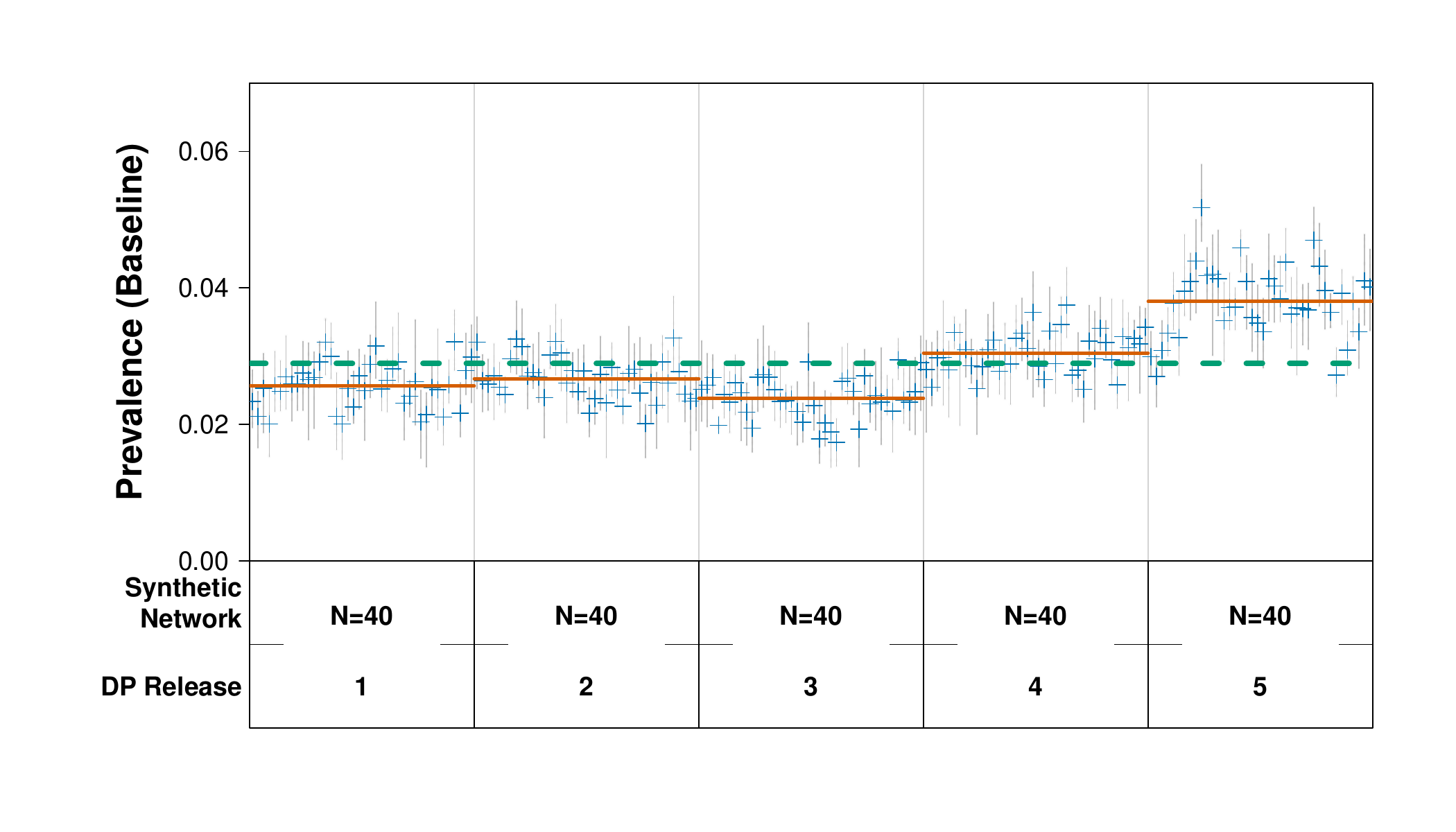}
        \caption{ERGM - Low Prevalence}
    \end{subfigure}
    \caption{
    Variance analysis at privacy budget $\eps = 1$ and truncated degree $\Delta = 3$ showing the breakdown of different sources of randomness in the experimental pipeline across different modeling approaches and prevalence conditions. The plots show the average prevalence of baseline scenario over all synthetic networks and simulations (\textcolor{purple}{solid line}), per differential private release (\textcolor{orange}{dashed line}), per synthetic network (\textcolor{blue}{plus sign}), and the range per simulation (\textcolor{gray}{horizontal line}).}
    \label{fig:variance_analysis}
\end{figure}

\clearpage

\subsection{Degree-Level Network Statistics}
\label{app:quality-metrics-degree}

\begin{figure}[h!]
    \centering
    \begin{subfigure}[b]{\textwidth}
        \centering
        \includegraphics[width=\textwidth]{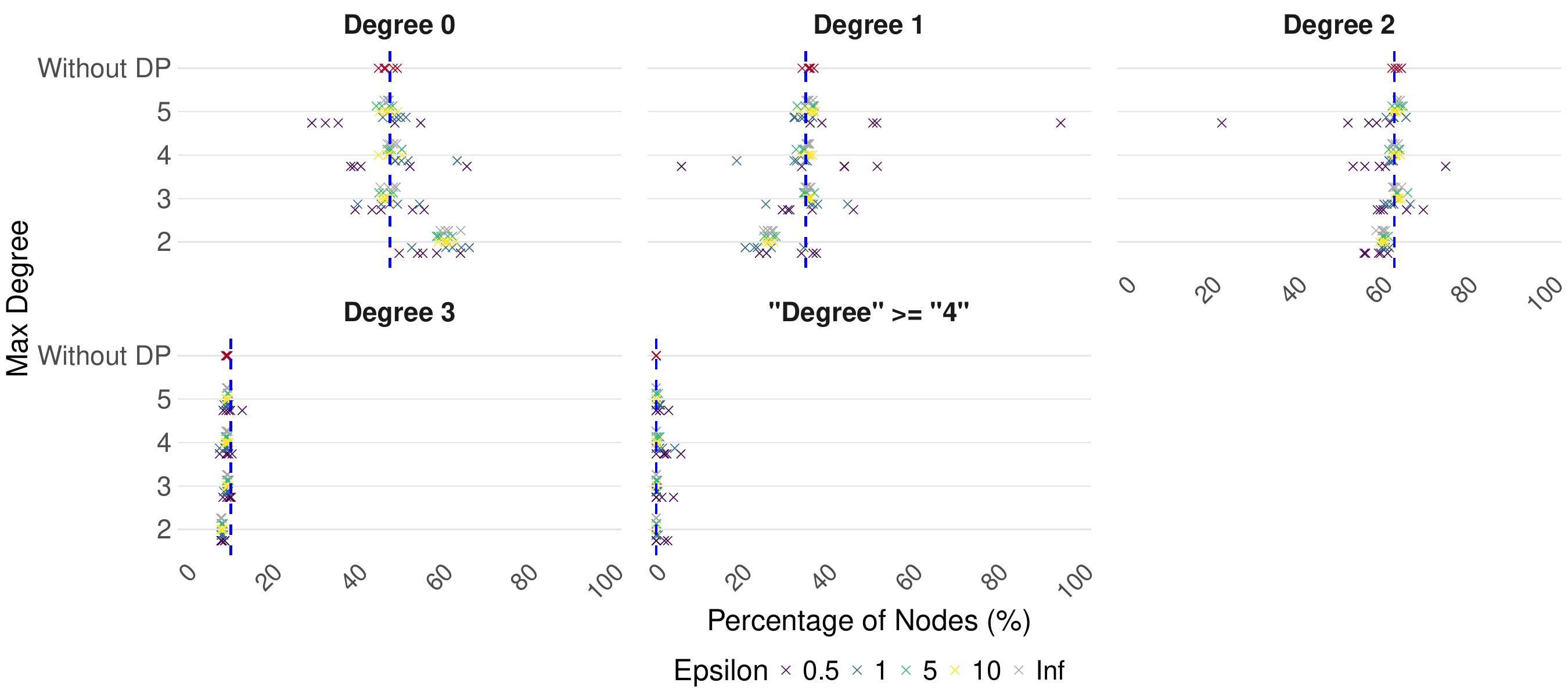}
        \caption{ERGM}
        \label{fig:node_metrics_ergm}
    \end{subfigure}
    \begin{subfigure}[b]{\textwidth}
        \centering
        \includegraphics[width=\textwidth]{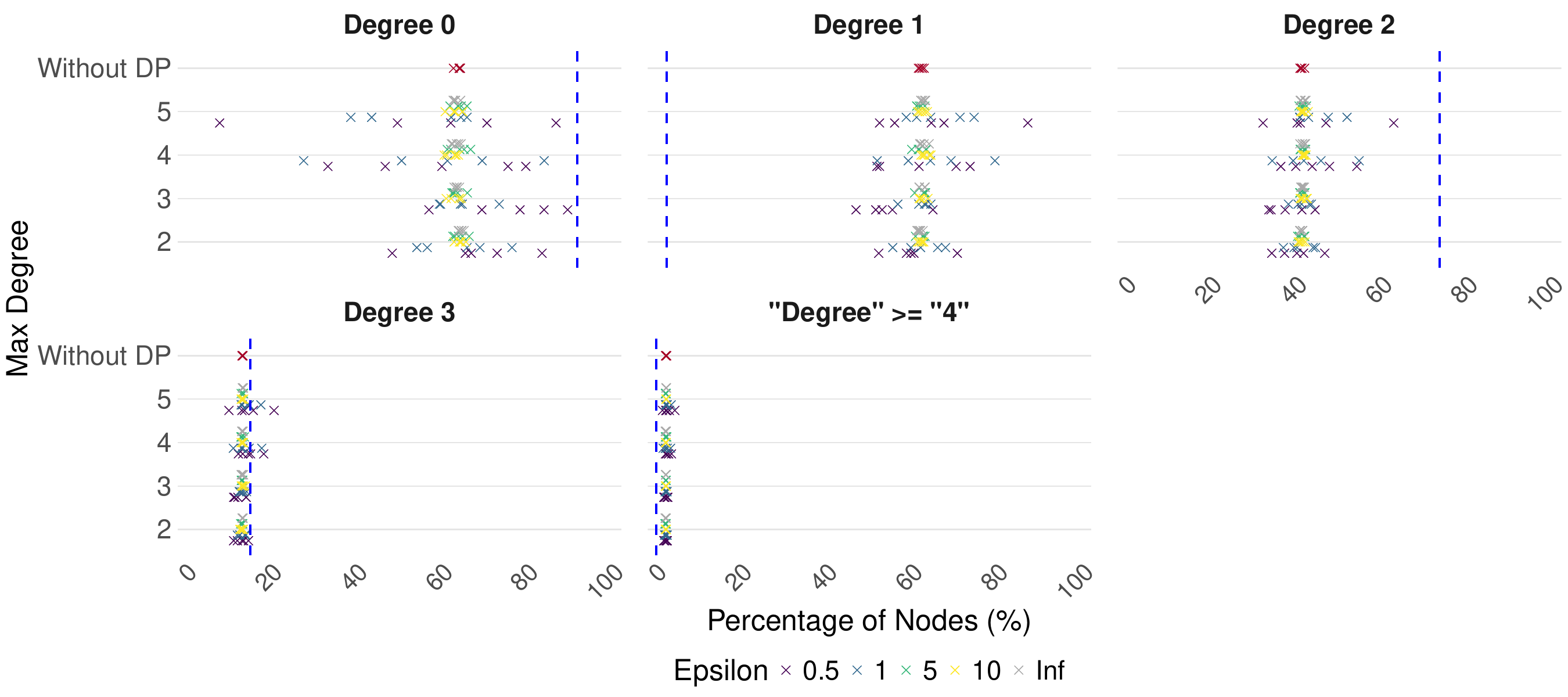}
        \caption{SBM}
        \label{fig:node_metrics_bm}
    \end{subfigure}
    \caption{Degree-level network statistics across different modeling approaches. The x-axis represents the proportion of nodes with each degree value as a percentage of the total network (10,000 nodes).}
    \label{fig:node_metrics}
\end{figure}

\clearpage

\shnote{if not referred, delete!}
\subsection{Group-Specific Analysis}
\label{sec:demographic_analysis}

\begin{figure}[h!]
    \centering
    \begin{subfigure}[b]{0.48\textwidth}
        \centering
        \includegraphics[width=\textwidth]{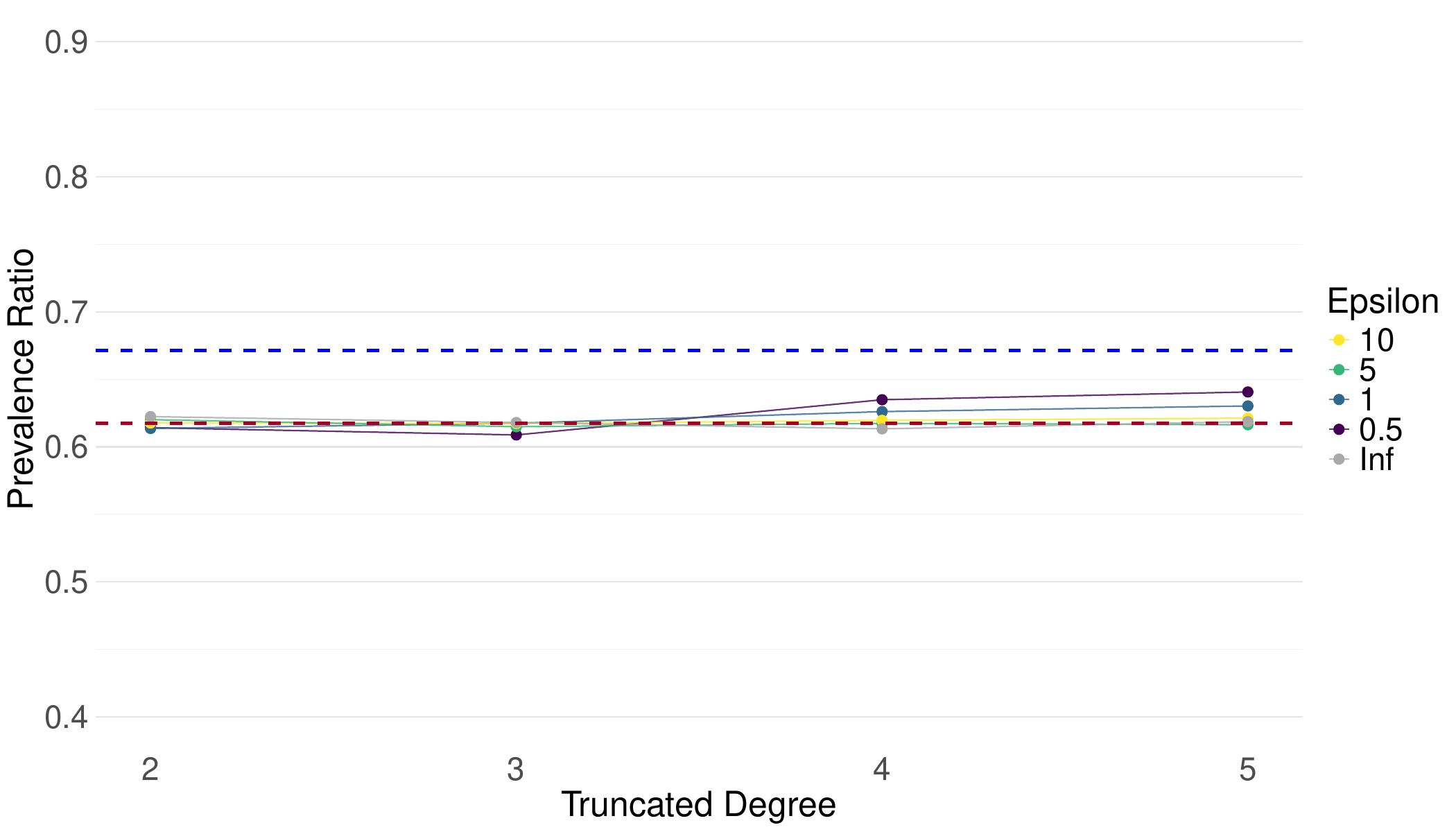}
        \caption{SBM - High Prevalence}
    \end{subfigure}
    \hfill
    \begin{subfigure}[b]{0.48\textwidth}
        \centering
        \includegraphics[width=\textwidth]{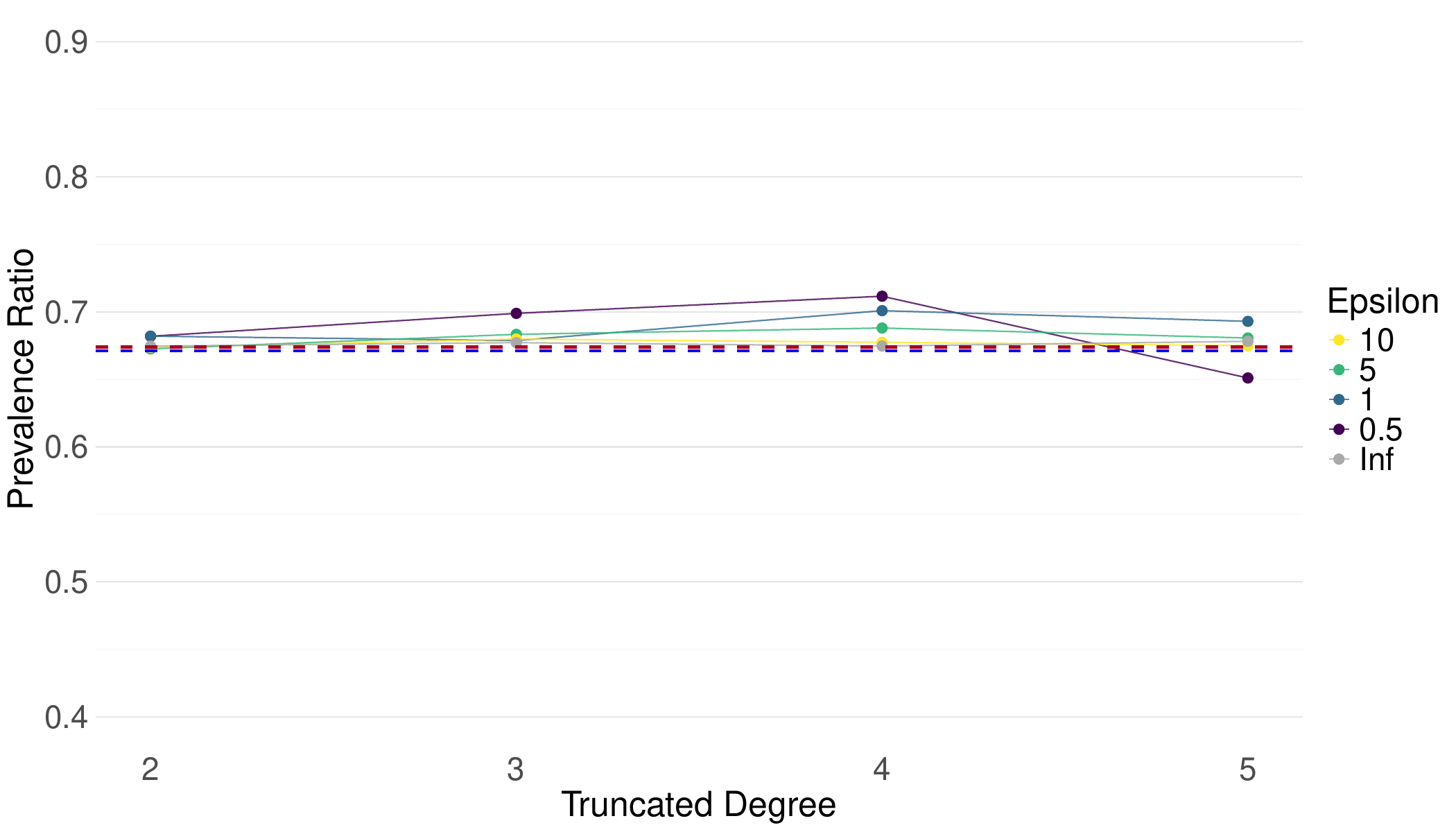}
        \caption{ERGM - High Prevalence}
    \end{subfigure}
    
    \vspace{0.5cm}
    
    \begin{subfigure}[b]{0.48\textwidth}
        \centering
        \includegraphics[width=\textwidth]{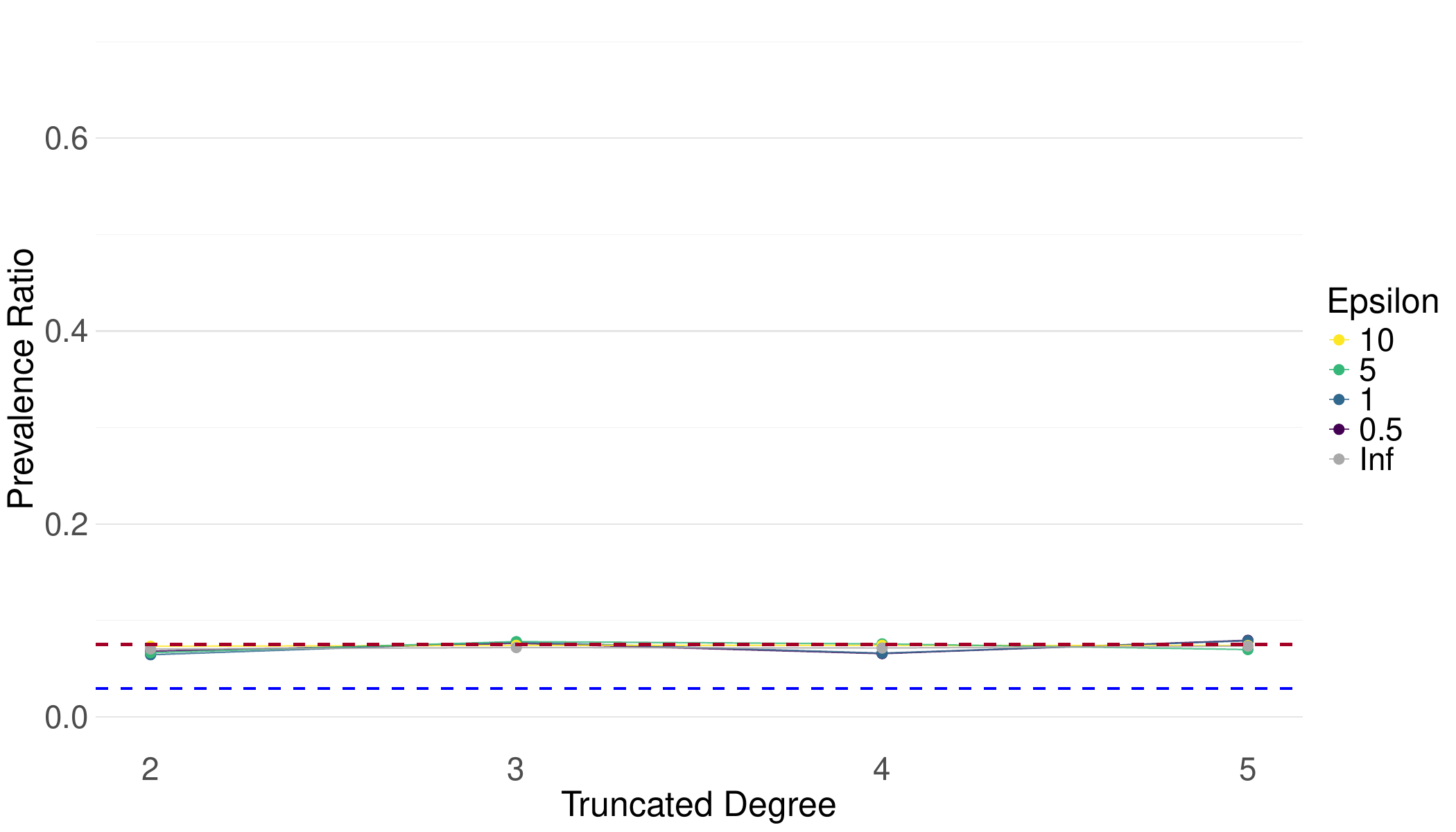}
        \caption{SBM - Low Prevalence}
    \end{subfigure}
    \hfill
    \begin{subfigure}[b]{0.48\textwidth}
        \centering
        \includegraphics[width=\textwidth]{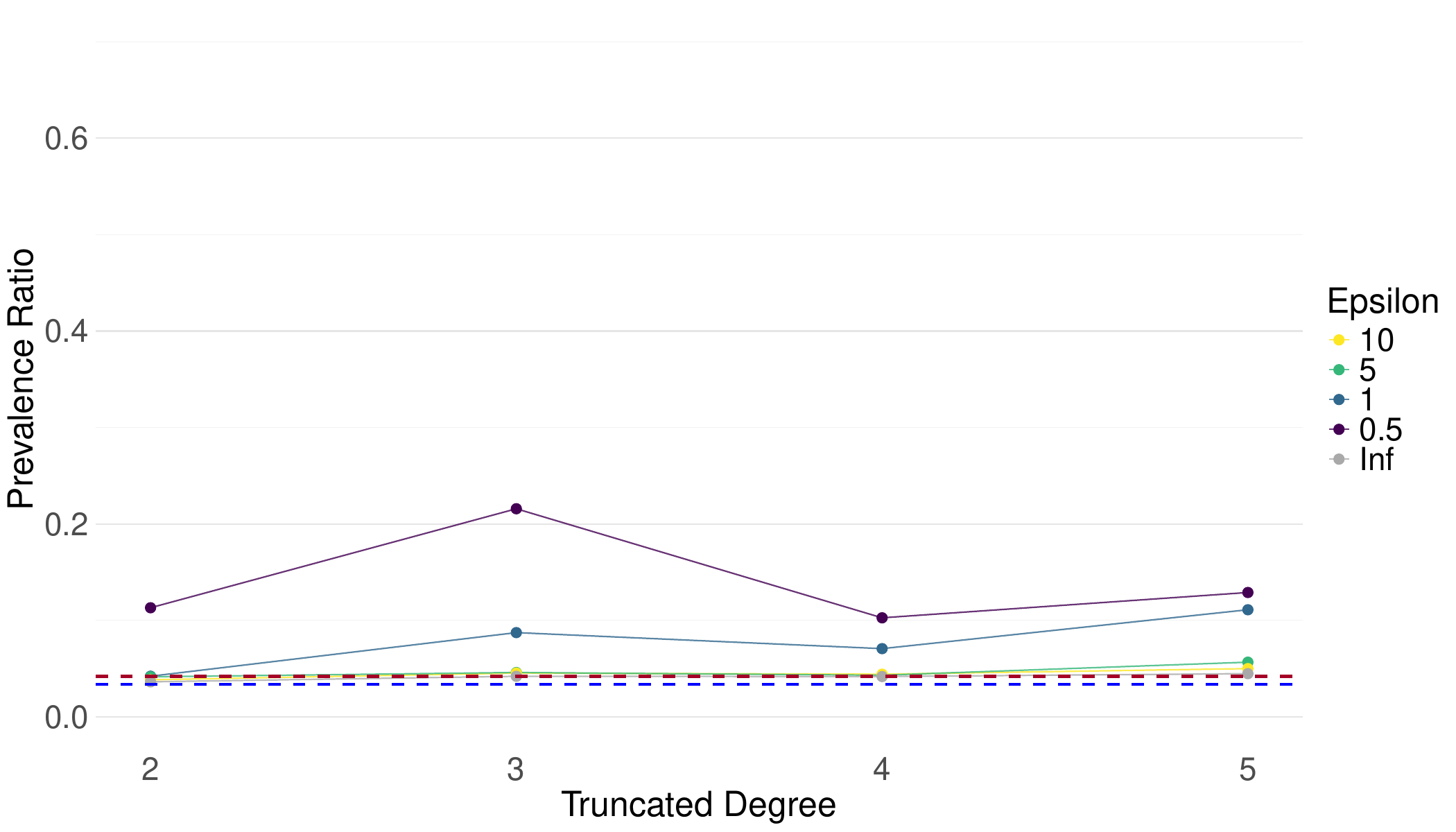}
        \caption{ERGM - Low Prevalence}
    \end{subfigure}

    \caption{Prevalence ratio across different maximum degree thresholds for age group 25-34 across all four modeling scenarios.}
    \label{fig:demographic_age_25_34}
\end{figure}

\begin{figure}[h!]
    \centering
    \begin{subfigure}[b]{0.48\textwidth}
        \centering
        \includegraphics[width=\textwidth]{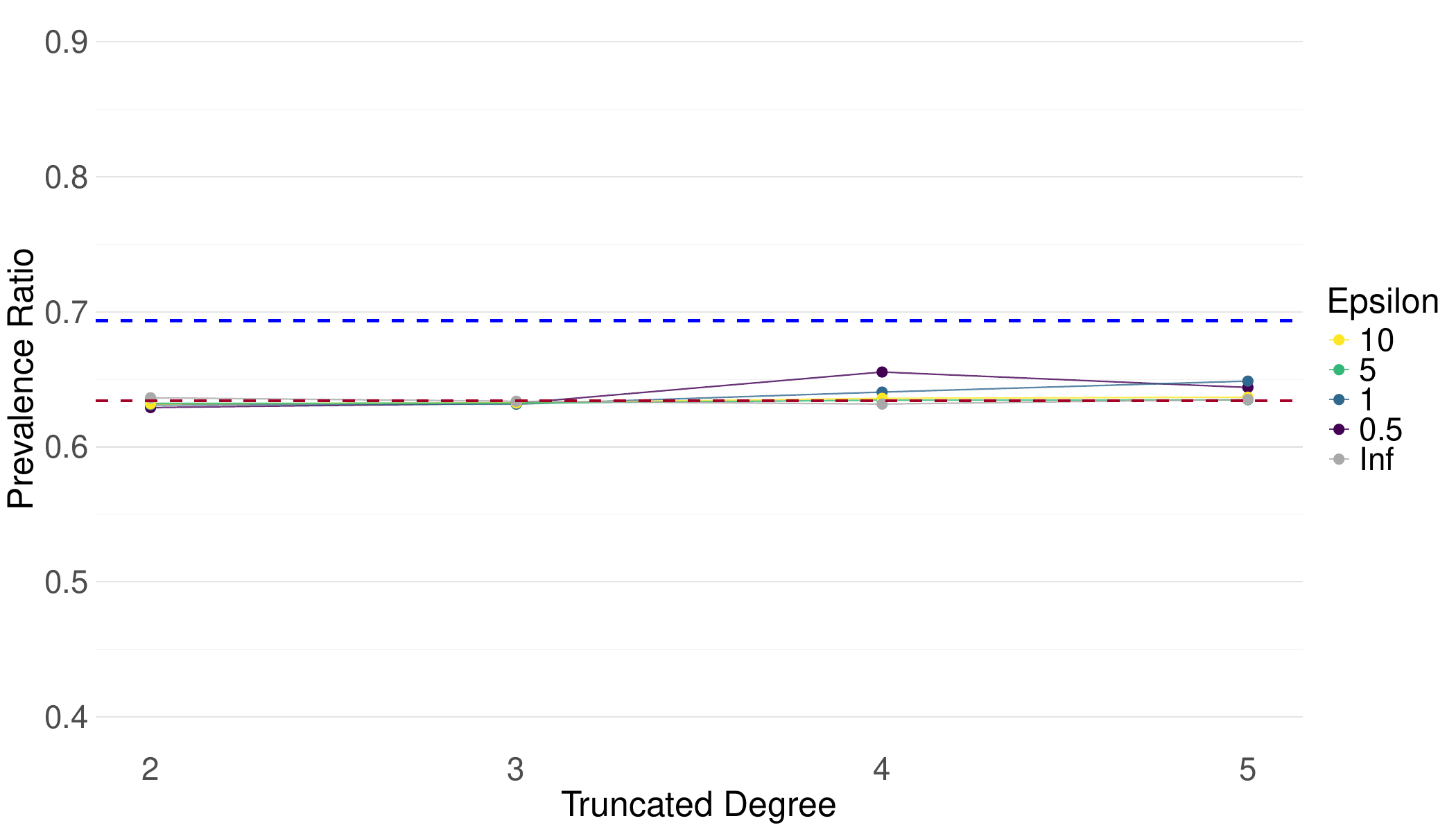}
        \caption{SBM - High Prevalence}
    \end{subfigure}
    \hfill
    \begin{subfigure}[b]{0.48\textwidth}
        \centering
        \includegraphics[width=\textwidth]{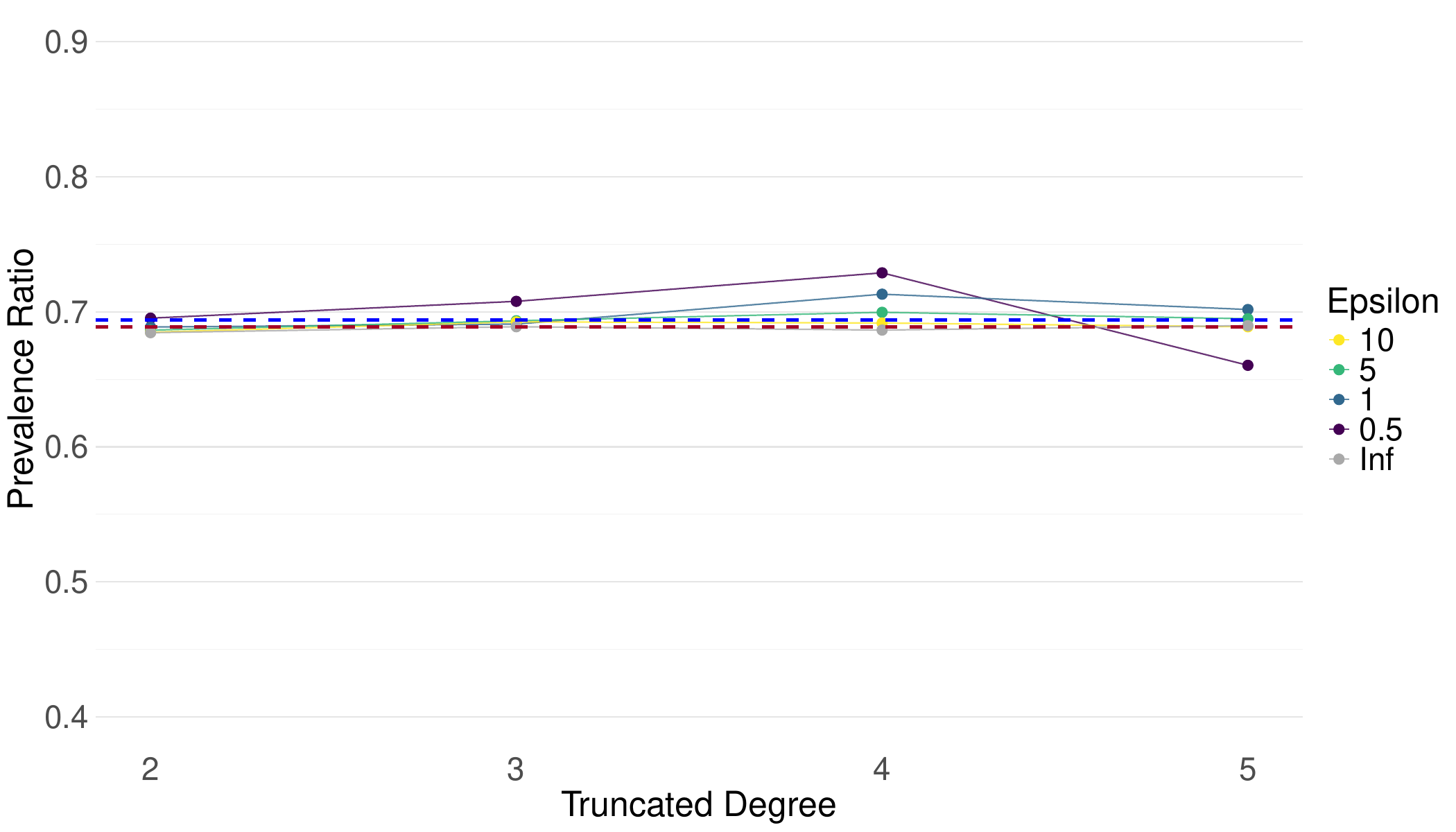}
        \caption{ERGM - High Prevalence}
    \end{subfigure}
    
    \vspace{0.5cm}
    
    \begin{subfigure}[b]{0.48\textwidth}
        \centering
        \includegraphics[width=\textwidth]{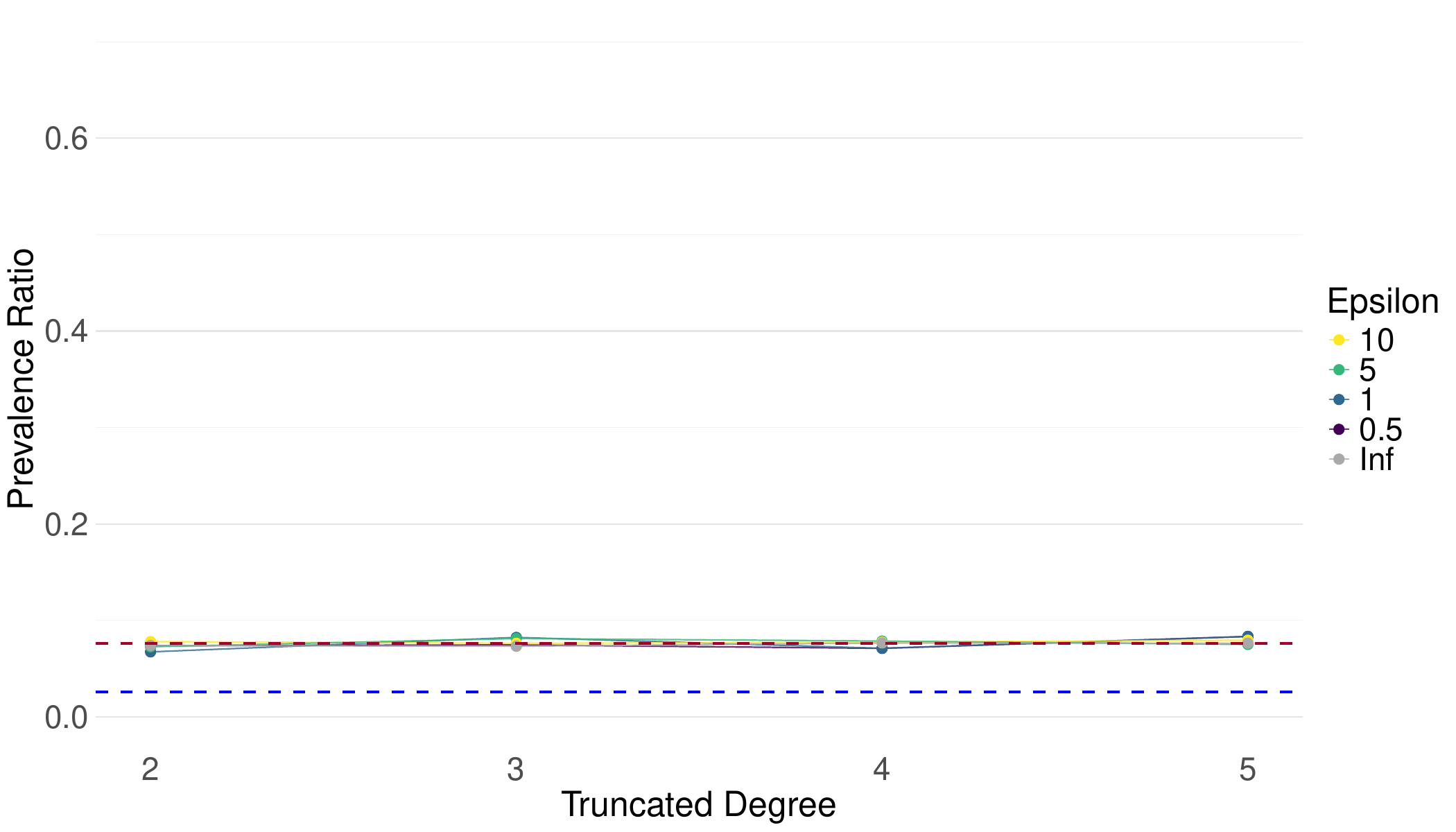}
        \caption{SBM - Low Prevalence}
    \end{subfigure}
    \hfill
    \begin{subfigure}[b]{0.48\textwidth}
        \centering
        \includegraphics[width=\textwidth]{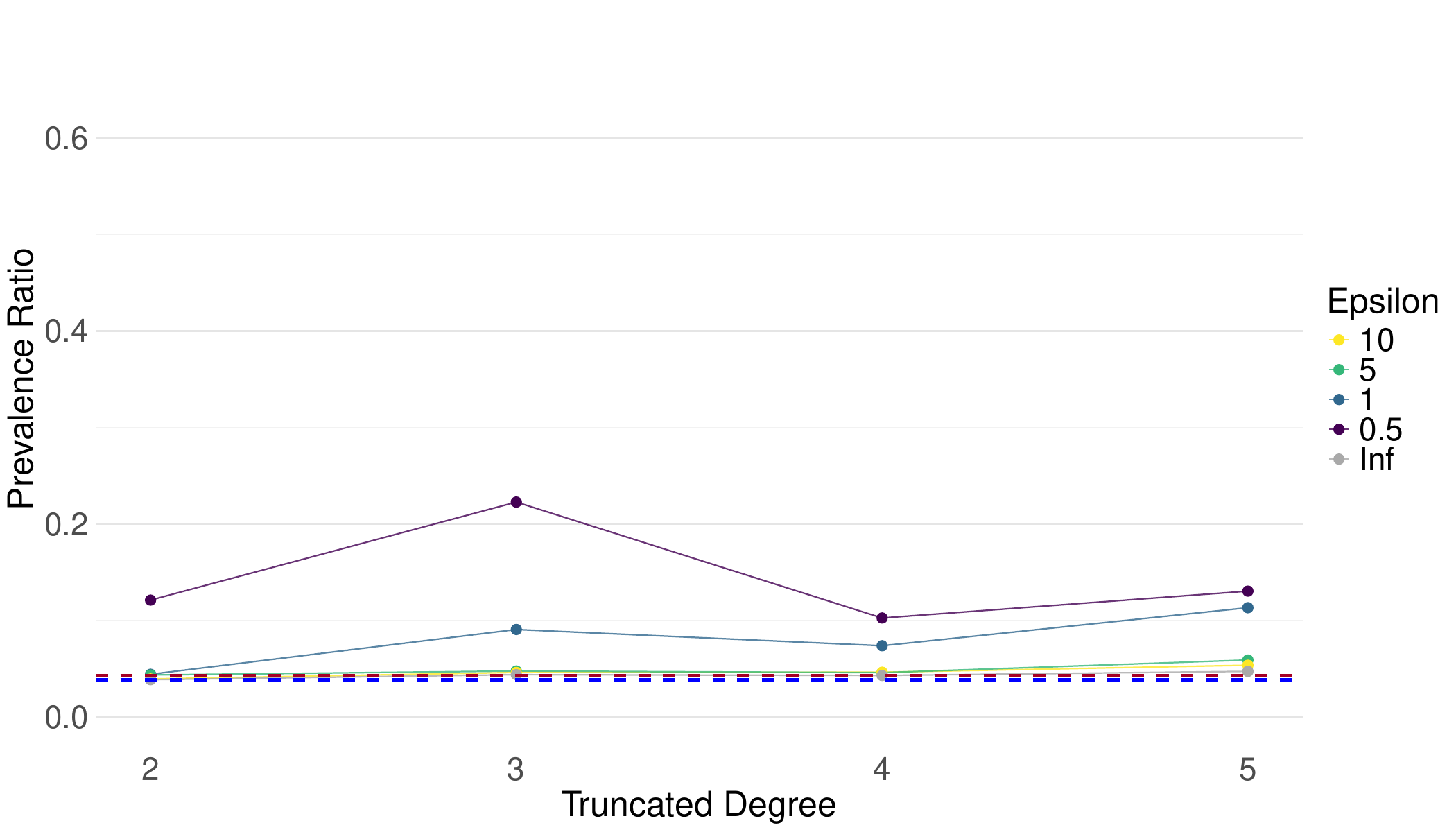}
        \caption{ERGM - Low Prevalence}
    \end{subfigure}

    \caption{Prevalence ratio across different maximum degree thresholds for Black population across all four modeling scenarios.}
    \label{fig:demographic_race_black}
\end{figure}

\clearpage

\end{document}